\documentclass[
  aps,
  reprint,
  superscriptaddress,
  bibnotes
]{revtex4-2}
\usepackage[utf8]{inputenc}
\usepackage[american]{babel}
\usepackage[normalem]{ulem}
\usepackage{mathpazo}
\usepackage{lmodern}
\usepackage{physics}
\usepackage{graphicx}
\usepackage{indentfirst}
\usepackage{listings}
\usepackage{braket}
\usepackage{amsmath}
\usepackage{amssymb}
\usepackage{amsthm}
\usepackage[font=small,labelfont=bf,justification=raggedright]{caption}
\usepackage{circuitikz}
\usepackage{subcaption}
\usepackage{booktabs}
\usepackage{bbold}
\usepackage{placeins}
\usepackage[title]{appendix}
\usepackage{qcircuit}
\usepackage[unicode=true, bookmarks=false, breaklinks=true, colorlinks=true, linkcolor=blue, citecolor=blue, urlcolor=blue]{hyperref}
\usepackage[dvipsnames]{xcolor}
\usepackage{tikz}

\usepackage{float}
\makeatletter
\let\newfloat\newfloat@ltx
\makeatother

\usepackage{algorithm}
\usepackage{algorithmicx}
\usepackage{algpseudocode}
\usepackage{textcomp}
\usepackage{circledsteps}
\usepackage{bbm}
\usepackage{xspace}
\usepackage{orcidlink}

\newcommand{\iu}{\mathrm{i}} 
\newcommand{\bigO}{\mathcal{O}}

\theoremstyle{definition}
\newtheorem{definition}{Definition}
\newtheorem{theorem}{Theorem}
\newtheorem{lemma}{Lemma}

\newcommand{\HKUSTGZ}{\affiliation{Thrust of Artificial Intelligence, Information Hub, The Hong Kong University of Science and Technology (Guangzhou), Guangzhou 511453, China}}
\newcommand{\QSC}{\affiliation{Quantum Science Center of Guangdong-Hong Kong-Macao Greater Bay Area, Shenzhen 518045, China}}

\begin{document}

\title{A Compass on the Quantum State Sphere: \\ The Hopf Ansatz for Arbitrary Pure-State Optimization}

\author{Ruge Lin\,\orcidlink{0000-0002-9297-7622}}
\email[Corresponding author: ]{gogoko699@gmail.com}
\HKUSTGZ

\author{Guangxi Li\,\orcidlink{0009-0009-5481-2411}}
\QSC

\begin{abstract}
Optimizing an arbitrary pure state requires more than a circuit that can
prepare it: one also needs coordinates, local distance information, and
measurable derivative directions. We introduce the Hopf ansatz, a binary-tree
circuit for arbitrary normalized real and complex state vectors. Internal
angles split probability between subtrees, while leaf phases supply the
complex degrees of freedom. The same tree recovers physical angles from
known amplitudes, gives an analytic diagonal metric on the state-vector
sphere, and prepares each nonzero coordinate tangent as a normalized state.

For Hamiltonian objectives and other costs with the same local
transition-moment form, each gradient component is a known metric factor
times a transition moment between the current state and a tangent state.
Signed branch states express these moments through expectation-value
measurements. Indexed preparation groups magnitude tangents by tree depth
and collects complex phase tangents into one leaf-indexed block. The number
of compiled gradient-access circuit families is logarithmic in
Hilbert-space dimension; measurement repetitions needed for a prescribed
accuracy remain a separate statistical cost. In deterministic real-state benchmarks with exact objectives and
gradients and known global optima, metric-aware Hopf updates reach
numerical-precision median gaps. Relative to coordinate-Adam baselines,
they yield smaller mean final gaps in variational eigensolver tests and
concentrate more runs at numerical precision in metrology-inspired tests.
The Hopf ansatz thus supplies a compass for arbitrary pure-state
optimization, with coordinates, geometry, and measurable directions in one
circuit structure.
\end{abstract}

\maketitle

\section{Introduction}

In state-based variational algorithms, an ansatz is both a preparation
routine and a coordinate system for search. Its parameters determine which
states are reachable, how updates move the state, and which derivatives can
be accessed from circuits. For arbitrary pure-state search, reachability
alone is not enough: one also needs an initialization map, local distance
information, and measurable derivative directions. These requirements arise
in variational quantum eigensolvers (VQE) and related algorithms
\cite{peruzzo2014variational,mcclean2016theory,cerezo2021variational}, as
well as in variational sensing and metrology, where the search may optimize
a probe state rather than an energy-minimizing wavefunction
\cite{koczor2020variational,kaubruegger2021quantum,meyer2021variational,kaubruegger2023optimal,vasilyev2024optimal,le2023variational,maclellan2024end}.

Universal state-preparation and circuit-synthesis methods establish
reachability. General synthesis results compile arbitrary transformations
\cite{knill1995approximation,shende2005synthesis}, while state-preparation
constructions load arbitrary vectors through controlled rotations,
optimized depth, or configurable decompositions
\cite{mottonen2004transformation,plesch2011quantum,zhang2022quantum,sun2023asymptotically,araujo2023configurable}.
Several also give explicit amplitude-to-angle maps. In quantum machine
learning, the choice of data-encoding circuit is itself part of model
design \cite{schuld2021effect,li2022concentration}. Specialized encoders
for constrained subspaces, such as fixed-Hamming-weight sectors, address
the complementary regime where the search space is smaller than the full
state sphere \cite{farias2025quantum}. Here the goal is to keep
preparation, inverse coordinates, the analytic metric, and tangent-state
synthesis within one circuit structure.

The Hopf ansatz realizes this structure with a binary tree based on
generalized Hopf coordinates \cite{cohl2013fourier,schobel2016orthogonal}.
Each internal coordinate is a directly programmed rotation angle that
splits amplitude between two subtrees. A root-to-leaf path determines one
computational-basis amplitude. Real-state signs are encoded in the final
splits; complex states use nonnegative magnitudes and one phase per leaf.
The tree works in both directions: angles prepare a state, while a
bottom-up inverse recovers them from its known amplitudes. Zero-weight
subtrees or leaves give the usual spherical-coordinate singularities,
handled by the inverse-map convention.

The real and complex parameter counts match the dimensions of the
corresponding normalized state-vector spheres. For complex states, we
retain global phase and use the ambient round-sphere metric, not the
projective Fubini--Study metric. In this geometry, the metric is analytic
and diagonal: a magnitude weight is the probability entering a tree node,
and a phase weight is the occupation of its leaf. On the regular chart,
these weights give componentwise gradient rescaling for metric-aware
updates, without hardware metric estimation or dense matrix inversion
\cite{stokes2020quantum}. The resulting state-sphere gradients also support
geodesic updates \cite{ferreira2025quantum}.

Derivative access completes the optimization interface. Analytic-gradient
methods, including parameter-shift rules, evaluate objectives at shifted
parameter values or related evolutions
\cite{schuld2019evaluating,banchi2021measuring,wierichs2022general}.
Other approaches address gradient overhead through circuit and measurement
reuse, density-based training, or expressivity--measurement tradeoffs
\cite{abbas2023quantum,bowles2025backprop,coyle2025density,chinzei2025tradeoff}.
In the Hopf chart, explicit shifts and clamps of gate angles prepare each
nonzero coordinate tangent as a normalized state. For Hamiltonian
objectives, each derivative is a known metric factor times a transition
moment between the current state and a tangent state. Signed superpositions
of these states express the moment through expectation-value measurements.

An index register groups the magnitude tangents at one tree depth into a
common circuit family; complex phase tangents use a leaf-indexed family.
The number of compiled gradient-access families is therefore logarithmic
in Hilbert-space dimension. In the indexed readout, an observed label and
branch sign identify the component and its estimator contribution. The
number of repetitions needed for a full gradient at fixed precision remains
a separate statistical cost.

The same construction applies beyond energy minimization whenever
derivatives admit local transition moments of Hermitian chain-rule
observables. Variational quantum metrology uses parameterized circuits for
probe optimization \cite{koczor2020variational}, joint probe and measurement
design \cite{meyer2021variational}, and Ramsey preparation and decoding
\cite{kaubruegger2021quantum}. Under local-estimation assumptions,
pure-state quantum Fisher information (QFI) objectives and fixed-readout
classical Fisher information (CFI) objectives of this form fit the Hopf
template by replacing the Hamiltonian with the corresponding chain-rule
observable.

The same Hopf geometry supports quantum backpropagation in the
controlled-reflection setting studied in the companion
paper~\cite{lin2026hopfqbp}. Its additional ``mirror'' assumption is
phase-calibrated controlled access to a reflection observable, keeping the
state and its reflected image in coherent branches with a known relative
phase. The present work obtains these transition moments
through expectation-value measurements on separately prepared state,
tangent, and signed-branch states, without requiring controlled
application of the objective observable.

Section~\ref{sec:hopf_coordinates} develops the Hopf chart, its circuit
realization, and its metric. Section~\ref{sec:gradient_access} constructs
tangent and branch states and gives the compiled resource bounds.
Section~\ref{sec:numerical_simulation} tests metric-aware updates using
exact gradients, separating optimizer behavior from measurement
fluctuations. The appendices collect algorithms, optimizer formulas,
metrological constructions, and proofs. Appendix~\ref{app:hopf4_example}
provides a four-qubit worked example linking the tree, gates, metric, and
derivative circuits.

\section{A universal Hopf variational chart}
\label{sec:hopf_coordinates}

The Hopf ansatz builds a state by recursively splitting amplitude between
two orthogonal subspaces. We use ``Hopf'' in the sense of generalized Hopf
coordinates from the polyspherical-coordinate literature, where successive
sine--cosine splits form a binary tree
\cite{cohl2013fourier,schobel2016orthogonal}. In the circuit, its leaves are
computational-basis states, and each internal angle controls a split between
two subtrees.

For two qubits, the real construction gives:
\[
\begin{aligned}
\ket{\psi}={}&
\cos\theta_1\bigl(\cos\theta_2\ket{00}+\sin\theta_2\ket{01}\bigr)\\
&+\sin\theta_1\bigl(\cos\theta_3\ket{10}+\sin\theta_3\ket{11}\bigr).
\end{aligned}
\]

The root angle \(\theta_1\) distributes probability between the two
first-qubit branches. The angles \(\theta_2\) and \(\theta_3\) then split
each branch between its two leaves. Every amplitude is a product of the
sine and cosine factors along its root-to-leaf path. The complex version
uses nonnegative magnitude splits and adds one phase to each leaf.

For the general construction, \(n\ge1\) denotes the number of qubits and
\(N:=2^n\) the Hilbert-space dimension, or number of leaves. The data-register
space is \(\mathcal H=(\mathbb C^2)^{\otimes n}\). We write the imaginary
unit as \(\iu\), reserving \(i,j,k\) for indices. Unsubscripted norms
are Hilbert-space norms. The symbols \(X,Y,Z\)
denote the Pauli operators, \(I\) the identity on the indicated space, and
\(\delta_{ij}\) the Kronecker delta. A dagger denotes the Hermitian adjoint;
a superscript \(\top\) denotes transpose.

\begin{definition}[Hopf binary tree]
\label{def:binary_tree}
Fix \(n\in\mathbb{N}\) and set \(N=2^n\).
The Hopf binary tree is the complete rooted binary tree with root at depth \(0\)
and leaves at depth \(n\).

Each node has a breadth--first tree index
\[
j\in\{1,2,\ldots,2N-1\},
\]
with root \(j=1\), left child \(2j\), and right child \(2j+1\).

Nodes \(1\le j\le N-1\) are internal nodes and carry split parameters
\(\theta_j\); nodes \(N\le j\le 2N-1\) are leaves.

A root-to-leaf path is encoded by bits \(q_n\cdots q_1\), read from
most significant bit (MSB) to least significant bit (LSB): left branch \(0\), right branch \(1\).  Let
\[
i=(q_n\cdots q_1)_2\in\{0,\ldots,N-1\}.
\]
The corresponding leaf has tree index \(j_{\mathrm{leaf}}=N+i\), labels the
computational-basis state
\[
|b_i\rangle \equiv |q_n\cdots q_1\rangle,
\]
and carries amplitude \(x_i\).
\end{definition}

\begin{definition}[Hopf coordinate system]
\label{def:hopf_coordinates}
On the tree of Definition~\ref{def:binary_tree}, each computational-basis
index $i\in\{0,\ldots,N-1\}$ has the binary expansion
\[
i = (q_n q_{n-1}\cdots q_1)_2, \qquad q_k\in\{0,1\}.
\]
The bits select the unique path from the root to that leaf.

Starting with $x_i=1$ and $j=1$, iterate over $k=n,n-1,\ldots,1$:
\[
\begin{aligned}
x_i &\leftarrow
\begin{cases}
x_i \cos\theta_j, & q_k=0,\\
x_i \sin\theta_j, & q_k=1,
\end{cases}
\\[4pt]
j &\leftarrow
\begin{cases}
2j, & q_k=0,\\
2j+1, & q_k=1.
\end{cases}
\end{aligned}
\]

To separate the coordinate blocks when needed, write
\[
\boldsymbol{\theta}_{\mathrm{mag}}
:=
(\theta_1,\ldots,\theta_{N-1}),
\qquad
\boldsymbol{\theta}_{\mathrm{ph}}
:=
(\theta_N,\ldots,\theta_{2N-1}).
\]
The real chart uses
\(\boldsymbol{\theta}=\boldsymbol{\theta}_{\mathrm{mag}}\).
The complex chart uses both blocks $\boldsymbol{\theta}=(\boldsymbol{\theta}_{\mathrm{mag}},\boldsymbol{\theta}_{\mathrm{ph}})$.

(A) Real Hopf coordinates.
The real Hopf coordinate system uses the internal-node angles
$\theta_1,\ldots,\theta_{N-1}$ with depth-dependent ranges
\[
\begin{aligned}
\theta_j &\in [0,\tfrac{\pi}{2}], && 1 \le j \le N/2-1,\\
\theta_j &\in [0,2\pi), && N/2 \le j \le N-1.
\end{aligned}
\]
The final-layer range $\theta_j\in[0,2\pi)$ allows the last split in each sibling
leaf pair to encode arbitrary real signs. With these ranges, the construction
produces
\[
\mathbf{x}=(x_0,\ldots,x_{N-1}) \in \mathbb{R}^{N},
\qquad
\sum_{i=0}^{N-1} x_i^2 = 1.
\]
Thus the real Hopf coordinates give a dimension-matched parameterization of the
real unit sphere $\mathbb{S}^{N-1}$ by $N-1$ continuous parameters.

(B) Complex Hopf coordinates.
The complex Hopf coordinate system uses the same internal-node tree, but now all
magnitude angles lie in
\[
\theta_j \in [0,\tfrac{\pi}{2}], \qquad 1 \le j \le N-1,
\]
so that the tree produces nonnegative magnitudes. Attach independent leaf phases
\[
\theta_{N+i} \in [0,2\pi), \qquad i=0,\ldots,N-1,
\]
through the assignment
\[
x_i \;\leftarrow\; x_i\,e^{\iu\theta_{N+i}}.
\]
This yields
\[
\mathbf{x}\in\mathbb{C}^{N},
\qquad
\sum_{i=0}^{N-1} |x_i|^2 = 1.
\]
The target manifold is the unit sphere in \(\mathbb{C}^{N}\), equivalently
\(\mathbb{S}^{2N-1}\) as a real manifold. The complex Hopf representation uses
\(N-1\) magnitude angles and \(N\) leaf phases, giving \(2N-1\) real
parameters. This dimension count retains global phase: the target is the
normalized state-vector sphere, not projective Hilbert space.

In both the real and complex cases, the Hopf coordinates should be understood as
a full sphere parameterization rather than as a single global one-to-one chart.
At boundary values or in zero-weight subtrees, coordinate representatives
need not be unique.
\end{definition}

Reading the tree backward gives the inverse map. Starting from known
amplitudes, sum their squared magnitudes upward through the tree. The
norms of the two child subtrees determine each magnitude split. At the
final real layer, the signed sibling amplitudes also determine the angle's
quadrant; in the complex chart, the leaf phases are read directly from
the amplitudes. This bottom-up procedure takes \(\bigO(N)\) classical time.
Lemma~\ref{lem:inverse_hopf_map} gives the explicit formulas and fixes the
conventions at zero-weight subtrees and leaves.

\begin{definition}[Hopf ansatz]
\label{def:hopf_ansatz}
The Hopf ansatz realizes Definition~\ref{def:hopf_coordinates} as a circuit
mapping $\ket{0}^{\otimes n}$ to
\[
\ket{\psi(\boldsymbol{\theta})}
=
\sum_{i=0}^{N-1} x_i(\boldsymbol{\theta}) \ket{b_i},
\]
with amplitudes given exactly by the Hopf recursion.

Each internal node is assigned to a rotation on the next qubit, with its
angle directly programming that split. The native ordering preserves the
tree's amplitude recursion from \(\ket{0}^{\otimes n}\) while omitting
controls that are unnecessary on the populated subspace.
Appendix~\ref{app:gate_schedule_routines} specifies the control masks and
gate order. Each split parameter remains tied to its designated gate
location. In the complex case, each final-layer rotation also carries the
two sibling leaf phases, so phases add no tree layers. All resource counts
below use this native ordering.

The Hopf ansatz uses the following single-qubit gates.

(A) Real Hopf ansatz.
In the real case, each gate is a possibly multi-controlled $R_y$ rotation,
\begin{equation}
R_y(\theta)
=
\begin{bmatrix}
\cos\theta & -\sin\theta \\
\sin\theta & \ \cos\theta
\end{bmatrix},
\end{equation}
which corresponds to the standard quantum gate $R_y(2\theta)$ under the
conventional definition $R_y(\varphi)=e^{-\iu\varphi Y/2}$. Its action on the
computational basis is
\begin{equation}
\begin{aligned}
R_y(\theta)\ket{0} &= \cos\theta\ket{0} + \sin\theta\ket{1},\\
R_y(\theta)\ket{1} &= -\sin\theta\ket{0} + \cos\theta\ket{1}.
\end{aligned}
\end{equation}

Given angles
\[
\boldsymbol{\theta}=(\theta_1,\ldots,\theta_{N-1})
\]
in the real Hopf coordinate ranges, Algorithm~\ref{alg:order_real}
(Appendix~\ref{app:gate_schedule_routines}) generates the corresponding gate
list. Applying the resulting circuit gives
\begin{equation}
\mathrm{HopfReal}(\boldsymbol{\theta})\ket{0}^{\otimes n}
=
\sum_{i=0}^{N-1} x_i \ket{b_i},
\end{equation}
with real amplitudes $x_i$.

(B) Complex Hopf ansatz.
In the complex case, the magnitude-splitting structure is identical to the real
case, while phases are attached at the leaf level. Accordingly, the
depth-$(n-1)$ rotations are promoted to generalized single-qubit rotations.
Here \(\theta_a\) is the split angle, and \(\theta_b,\theta_c\) are the
left and right leaf phases:
\begin{equation}
R_{\mathbb{C}}(\theta_a,\theta_b,\theta_c)
=
\begin{bmatrix}
e^{\iu\theta_b}\cos\theta_a & -e^{-\iu\theta_c}\sin\theta_a \\
e^{\iu\theta_c}\sin\theta_a & \ \ e^{-\iu\theta_b}\cos\theta_a
\end{bmatrix},
\end{equation}
with action
\begin{equation}
\begin{aligned}
R_{\mathbb{C}}(\theta_a,\theta_b,\theta_c)\ket{0}
&=
e^{\iu\theta_b}\cos\theta_a\ket{0}
+
e^{\iu\theta_c}\sin\theta_a\ket{1},\\
R_{\mathbb{C}}(\theta_a,\theta_b,\theta_c)\ket{1}
&=
- e^{-\iu\theta_c}\sin\theta_a\ket{0}
+
e^{-\iu\theta_b}\cos\theta_a\ket{1}.
\end{aligned}
\end{equation}
This gate reduces to $R_y(\theta_a)$ when $\theta_b=\theta_c=0$. For state preparation only the action on \(\ket{0}\) is used; the action on \(\ket{1}\) specifies a unitary extension and supports controlled use of the same gate family.

Given parameters
\[
\boldsymbol{\theta}=(\theta_1,\ldots,\theta_{2N-1}),
\]
Algorithm~\ref{alg:order_complex}
(Appendix~\ref{app:gate_schedule_routines}) generates the corresponding gate
list. Applying the resulting circuit gives
\begin{equation}
\mathrm{HopfComplex}(\boldsymbol{\theta})\ket{0}^{\otimes n}
=
\sum_{i=0}^{N-1} x_i \ket{b_i},
\end{equation}
with complex amplitudes satisfying $\sum_i |x_i|^2=1$.
\end{definition}

For tangent-state synthesis we also use gate-parameter assignments on the same
Hopf skeleton outside the canonical ranges of
Definition~\ref{def:hopf_coordinates}; for instance, some internal-node angles
are shifted by \(+\pi/2\). These remain valid physical gate settings. Thus
``prepared by the Hopf ansatz'' means prepared by the same ordered Hopf skeleton
and gate types, not necessarily by a point in the canonical coordinate domain.

\subsection*{Hopf geometry on the state sphere}
\label{sec:hopf_geometry}

For the same small angle change, a split in a low-weight subtree moves
the full state less than the root split does. The metric measures this
local sensitivity.
For the real two-qubit state above,
\[
\boldsymbol{g}^{\mathbb{R}}
=\operatorname{diag}\!\bigl(1,\cos^2\theta_1,\sin^2\theta_1\bigr).
\]
The root has unit weight; each final split carries its branch probability.
For every magnitude coordinate, the metric weight is the probability
entering its node.

The absence of cross terms also follows from the tree. Splits in disjoint
subtrees change disjoint groups of amplitudes. When one split lies below
another, it changes the normalized state inside a child subtree without
changing the probability entering it. Its tangent is therefore orthogonal
to the direction that transfers weight between that child and its sibling.

For real normalized states, the pullback Fubini--Study metric coincides
with the sphere metric. For complex states, we use the ambient round-sphere
metric on normalized state vectors, retaining the global-phase direction
rather than taking the projective quotient. This is the convention used
by the state-sphere optimizer updates \cite{ferreira2025quantum}.
A leaf phase changes only its own amplitude, so its metric weight is
that leaf's occupation probability. Distinct phase directions have disjoint
leaf support, and magnitude and phase variations are orthogonal in the real
inner product used by this metric. We now give the formal definitions
for the Hopf state of Definition~\ref{def:hopf_ansatz}.

\begin{definition}[Metrics in parameter space]
\label{def:hopf_metric_convention}
Let $\boldsymbol{\theta}$ denote the Hopf parameters and define the raw
coordinate tangent associated with $\theta_i$ by
\[
\ket{\dot\psi_i(\boldsymbol{\theta})}
:=
\partial_{\theta_i}\ket{\psi(\boldsymbol{\theta})}.
\]

(A) Real-state Hopf ansatz. Assume $x_k(\boldsymbol{\theta})\in\mathbb{R}$ and
$\sum_{k=0}^{N-1}x_k(\boldsymbol{\theta})^2=1$.
We define the parameter-space metric as the pullback of the Fubini--Study metric
\cite{stokes2020quantum}:
\begin{equation}
\begin{split}
g^{\mathbb{R}}_{i\ell}(\boldsymbol{\theta})
:= \mathrm{Re}\!\Big[
&\braket{\dot\psi_i(\boldsymbol{\theta})|\dot\psi_\ell(\boldsymbol{\theta})} \\
&- \braket{\dot\psi_i(\boldsymbol{\theta})|\psi(\boldsymbol{\theta})}
   \braket{\psi(\boldsymbol{\theta})|\dot\psi_\ell(\boldsymbol{\theta})}
\Big].
\end{split}
\label{eq:FS_real_def}
\end{equation}
Because $\sum_{k=0}^{N-1}x_k^2=1$ implies
$\braket{\psi|\dot\psi_i}=\sum_{k=0}^{N-1}x_k\,\partial_{\theta_i} x_k=0$, the projection term
vanishes, and hence
\begin{equation}
g^{\mathbb{R}}_{i\ell}(\boldsymbol{\theta})
=\braket{\dot\psi_i(\boldsymbol{\theta})|\dot\psi_\ell(\boldsymbol{\theta})}
=\sum_{k=0}^{N-1}\frac{\partial x_k}{\partial\theta_i}
                  \frac{\partial x_k}{\partial\theta_\ell}.
\label{eq:real_metric_JTJ}
\end{equation}

(B) Complex-state Hopf ansatz.
Assume $x_k(\boldsymbol{\theta})\in\mathbb{C}$ and
\(
\sum_{k=0}^{N-1}|x_k(\boldsymbol{\theta})|^2 = 1,
\)
so that
\(
\psi(\boldsymbol{\theta}) \in \mathbb{S}^{2N-1} \subset \mathbb{C}^{N}.
\)
We define the parameter-space metric as the pullback of the ambient round-sphere
metric on $\mathbb{S}^{2N-1}$:
\begin{equation}
g^{\mathbb{C}}_{i\ell}(\boldsymbol{\theta})
:= \mathrm{Re}\,
\braket{\dot\psi_i(\boldsymbol{\theta})|\dot\psi_\ell(\boldsymbol{\theta})}.
\label{eq:round_complex_def}
\end{equation}
Passing to the projective Fubini--Study metric would subtract
$\braket{\dot\psi_i|\psi}\braket{\psi|\dot\psi_\ell}$ from the inner
product before taking its real part.

\end{definition}

The following theorem gives these subtree and leaf weights explicitly
and shows that all off-diagonal metric entries vanish.

\begin{theorem}[Diagonal metric of the Hopf ansatz]
\label{thm:hopf_metric_diagonal}
Let $\ket{\psi(\boldsymbol{\theta})}$ be prepared by the Hopf ansatz in
Definition~\ref{def:hopf_ansatz}.

(A) Real-state version. For $\boldsymbol{\theta}=(\theta_1,\ldots,\theta_{N-1})$, the metric
$\boldsymbol{g}^{\mathbb{R}}$ of~\eqref{eq:real_metric_JTJ} is diagonal and
$g^{\mathbb{R}}_{1,1}=1$. For $i>1$, let $d_i=\lfloor\log_2 i\rfloor$, and let
$q_1\cdots q_{d_i}$ be the length-$d_i$ binary expansion of $i-2^{d_i}$, padded
with leading zeros and written most-significant bit first. Let
$s_1,\ldots,s_{d_i}$ be the ancestor indices from the root $s_1=1$ to the parent
of node $i$. Then
\begin{equation}
g^{\mathbb{R}}_{i,i}(\boldsymbol{\theta})
=\prod_{m=1}^{d_i}
\begin{cases}
\cos^2\theta_{s_m}, & q_m=0,\\
\sin^2\theta_{s_m}, & q_m=1.
\end{cases}
\label{eq:g_real_diag}
\end{equation}

(B) Complex-state version. For parameters ordered as $\theta_1,\ldots,\theta_{N-1}$ (magnitudes)
followed by $\theta_N,\ldots,\theta_{2N-1}$ (phases), the round-sphere pullback metric $\boldsymbol{g}^{\mathbb{C}}$ of Definition~\ref{def:hopf_metric_convention}(B) is diagonal with $g^{\mathbb{C}}_{1,1}=1$ and
\begin{equation}
g^{\mathbb{C}}_{i,i}(\boldsymbol{\theta})
=
\begin{cases}
g^{\mathbb{R}}_{i,i}(\boldsymbol{\theta}_{\mathrm{mag}}), & 1\le i\le N-1,\\[4pt]
|x_{i-N}(\boldsymbol{\theta})|^2, & N\le i\le 2N-1.
\end{cases}
\label{eq:g_complex_diag}
\end{equation}

Thus, in the complex Hopf ansatz, the diagonal metric has two blocks: a
magnitude block \(1\le i\le N-1\), identical to the real Hopf metric, and a
phase block \(N\le i\le 2N-1\), whose entries are the leaf occupations
\(|x_{i-N}(\boldsymbol{\theta})|^2\). Hence phase directions are weighted by
local squared amplitudes, while magnitude directions inherit the real Hopf tree
metric.

\end{theorem}

Proofs of the main results are collected in Appendix~\ref{app:proof}.

On the regular set, the diagonal metric converts coordinate gradients
into Riemannian descent directions by componentwise rescaling, without
hardware estimation or inversion of a dense quantum geometric tensor.
Appendix~\ref{app:optimizers} gives the Jacobian formulas and state-sphere
maps used by the optimizer family. The same metric weights also normalize
the coordinate tangents. The next section shows how to prepare the
resulting tangent states exactly within the Hopf circuit.

\section{Hopf gradient access and scaling}
\label{sec:gradient_access}

The metric fixes the length of each coordinate tangent. To use that direction
in an optimization step, we must also prepare it and measure the objective's
response. We first normalize and synthesize the tangent, then express its
response through expectation values on signed branch states. An index register
organizes these preparations by magnitude layer or, for complex phases, by leaf.

\begin{definition}[Normalized coordinate tangent]
\label{def:normalized_partial_general}
Let \(g_{i\ell}(\boldsymbol{\theta})\) be the parameter-space metric of
Definition~\ref{def:hopf_metric_convention}. Whenever
\(g_{i,i}(\boldsymbol{\theta})>0\), define
\begin{equation}
\ket{e_i(\boldsymbol{\theta})}
:=
\frac{\ket{\dot\psi_i(\boldsymbol{\theta})}}
{\sqrt{g_{i,i}(\boldsymbol{\theta})}}.
\label{eq:normalized_partial_general}
\end{equation}
Thus \(\ket{\dot\psi_i}\) denotes the raw coordinate tangent, whereas
\(\ket{e_i}\) denotes its normalized tangent state, with
\(\|\ket{e_i(\boldsymbol{\theta})}\|=1\).
\end{definition}

If \(g_{i,i}(\boldsymbol{\theta})=0\), the corresponding coordinate lies on a
boundary stratum of the Hopf parameterization, such as a zero-weight subtree or
zero-occupation leaf, and the normalized coordinate tangent is not part of the
regular chart. The identities below are therefore stated on the regular set;
a zero metric weight means that the raw derivative itself vanishes. Its
coordinate contribution to a differentiable state-dependent objective is
therefore zero,
even though no normalized tangent is defined there. Coordinate representatives
at such points use the inverse-map convention of
Appendix~\ref{app:inverse_hopf_map}.

For an energy objective, the derivative separates into a known geometric
scale and a transition moment.

\begin{theorem}[Gradient estimation for VQE]
\label{thm:gradient_estimation_short}
For a real or complex Hopf state \(\ket{\psi(\boldsymbol{\theta})}\)
of Definition~\ref{def:hopf_ansatz}, define
\[
E(\boldsymbol{\theta})
=
\bra{\psi(\boldsymbol{\theta})}H\ket{\psi(\boldsymbol{\theta})}
\]
for a parameter-independent Hermitian Hamiltonian \(H\).
For each coordinate \(\theta_i\) with \(g_{i,i}>0\),
\begin{equation}
\partial_{\theta_i} E(\boldsymbol{\theta})
=
2\sqrt{g_{i,i}(\boldsymbol{\theta})}\;
\mathrm{Re}\!\left[
\bra{e_i(\boldsymbol{\theta})}\,H\,\ket{\psi(\boldsymbol{\theta})}
\right].
\label{eq:hopf_grad_short}
\end{equation}
Here \(\ket{e_i(\boldsymbol{\theta})}\) denotes the normalized coordinate tangent of Definition~\ref{def:normalized_partial_general}. The formula \eqref{eq:hopf_grad_short} holds for both the real-state and complex-state versions of the Hopf ansatz, with the metric \(g_{i,i}\) chosen according to Definition~\ref{def:hopf_metric_convention}.
\end{theorem}

The Hopf-specific step is to prepare the normalized tangent on the same circuit
skeleton. For a magnitude parameter, the ancestor angles route all amplitude
into the selected subtree. A \(\pi/2\) shift at its first split produces the
derivative, while the descendant angles and any leaf phases retain the
subtree's internal structure. For a phase parameter, the circuit isolates a
single leaf and shifts its phase by \(\pi/2\). The following assignments make
these operations explicit.

\begin{theorem}[Normalized coordinate tangents as quantum states (Hopf)]
\label{thm:synthesis_normalized_direction_hopf}
Let
\[
\ket{\psi(\boldsymbol{\theta})}
=
\begin{cases}
\mathrm{HopfReal}(\boldsymbol{\theta})\ket{0}^{\otimes n},\\
\mathrm{HopfComplex}(\boldsymbol{\theta})\ket{0}^{\otimes n},
\end{cases}
\]
with parameters arranged on the Hopf binary tree as in Definitions~\ref{def:hopf_coordinates} and~\ref{def:hopf_ansatz}.
For any parameter index \(i\) with \(g_{i,i}(\boldsymbol{\theta})>0\), there exists an explicit gate-parameter assignment \(\boldsymbol{\theta}^{(i)}\) on the same ordered Hopf circuit skeleton such that
\begin{equation}
\ket{e_i(\boldsymbol{\theta})}
=
\begin{cases}
\mathrm{HopfReal}\!\big(\boldsymbol{\theta}^{(i)}\big)\ket{0}^{\otimes n},\\
\mathrm{HopfComplex}\!\big(\boldsymbol{\theta}^{(i)}\big)\ket{0}^{\otimes n},
\end{cases}
.
\label{eq:tangent_state_preparation}
\end{equation}
Here \(\mathrm{HopfReal}\) and \(\mathrm{HopfComplex}\) follow the off-chart
gate-assignment convention given above.

Fix a magnitude-node index \(i\in\{1,\ldots,N-1\}\) and define
\(
d_i=\lfloor\log_2 i\rfloor
\).
Let \(s_1,\ldots,s_{d_i}\) be the ancestor indices on the unique path from the root \(s_1=1\) to the parent of node \(i\), and let
\(
q_1\cdots q_{d_i}\in\{0,1\}^{d_i}
\)
be the corresponding left/right bits (MSB first), where \(q_m=0\) for a left branch and \(q_m=1\) for a right branch.
Define
\[
A_L^{(i)} := \{s_m \mid q_m=0\},
\qquad
A_R^{(i)} := \{s_m \mid q_m=1\},
\]
let \(D^{(i)}\) denote the set of strict descendants of node \(i\) (internal nodes only), and let \(\mathcal{L}(i)\subseteq\{0,\ldots,N-1\}\) be the set of
computational-basis leaf labels in the subtree rooted at \(i\).

(A) Magnitude parameters (real and complex cases).
This construction applies identically to the real Hopf gate skeleton and to the magnitude block of the complex Hopf gate skeleton. The resulting assignments \(\boldsymbol{\theta}^{(i)}\) are gate settings for that skeleton and need not belong to the canonical Hopf chart domain.

Define \(\boldsymbol{\theta}^{(i)}\) by
\begin{align}
\theta^{(i)}_i &= \theta_i + \tfrac{\pi}{2}, \label{eq:ti_shift_mag}\\
\theta^{(i)}_a &= 0, \qquad a\in A_L^{(i)}, \label{eq:anc_left_zero_mag}\\
\theta^{(i)}_a &= \tfrac{\pi}{2}, \qquad a\in A_R^{(i)}, \label{eq:anc_right_pi2_mag}\\
\theta^{(i)}_d &= \theta_d, \qquad d\in D^{(i)}, \label{eq:desc_keep_mag}\\
\theta^{(i)}_j &= 0, \qquad \text{otherwise}, \qquad 1\le j\le N-1. \label{eq:else_zero_mag}
\end{align}
In the complex case, the phase block is further set to
\begin{equation}
\theta^{(i)}_{N+\ell} =
\begin{cases}
\theta_{N+\ell}, & \ell\in\mathcal{L}(i),\\
0, & \text{otherwise}.
\end{cases}
\label{eq:phase_keep_on_subtree}
\end{equation}
Then
\begin{equation}
\begin{aligned}
\mathrm{HopfReal}\!\big(\boldsymbol{\theta}^{(i)}\big)\ket{0}^{\otimes n}
&=
\ket{e_i(\boldsymbol{\theta})},\\
\mathrm{HopfComplex}\!\big(\boldsymbol{\theta}^{(i)}\big)\ket{0}^{\otimes n}
&=
\ket{e_i(\boldsymbol{\theta})}.
\end{aligned}
\end{equation}
If node \(i\) lies at depth \(d\) (root at depth \(0\)), the resulting state is supported only on the \(N/2^d\) computational basis states corresponding to the leaves of the subtree rooted at \(i\). Hence its support size is at most \(N/2^d\), with equality when all leaves in that subtree have nonzero amplitude.

(B) Phase parameters (complex-state Hopf ansatz only).
Let \(i=N+\ell\) with \(0\le \ell\le N-1\).
Write
\[
x_\ell(\boldsymbol{\theta})
=
r_\ell(\boldsymbol{\theta})\,e^{\iu\theta_{N+\ell}},
\qquad
r_\ell(\boldsymbol{\theta})\ge 0.
\]
Let \((t_1,\ldots,t_n)\) be the internal nodes encountered along the root-to-leaf path of leaf \(\ell\), with corresponding bits \((b_1,\ldots,b_n)\) (MSB first).

Define \(\boldsymbol{\theta}^{(i)}\) by clamping the magnitude block so that \(|x_\ell|=1\):
\begin{equation}
\begin{aligned}
\theta^{(i)}_{t_m}
&=
\begin{cases}
0, & b_m=0,\\
\tfrac{\pi}{2}, & b_m=1,
\end{cases}
\\
\theta^{(i)}_j &= 0,
\qquad j\notin\{t_1,\ldots,t_n\},\quad 1\le j\le N-1,
\end{aligned}
\label{eq:phase_case_clamp}
\end{equation}
and in the phase block set
\begin{equation}
\theta^{(i)}_{N+\ell}=\theta_{N+\ell}+\tfrac{\pi}{2},
\qquad
\theta^{(i)}_{N+m}=0 \quad \text{for } m\neq \ell.
\label{eq:phase_case_shift}
\end{equation}
Then
\[
\mathrm{HopfComplex}\!\big(\boldsymbol{\theta}^{(i)}\big)\ket{0}^{\otimes n}
=
\ket{e_{N+\ell}(\boldsymbol{\theta})},
\]
since
\(
\ket{\dot\psi_{N+\ell}(\boldsymbol{\theta})}
=
\iu\,x_\ell(\boldsymbol{\theta})\ket{b_\ell}
\)
and
\(
g_{i,i}(\boldsymbol{\theta})=|x_\ell(\boldsymbol{\theta})|^2
\).

\end{theorem}

\subsection*{Gradient estimation from signed branch states}

The transition moment can be read from the energy of a signed superposition of
the state and its tangent. We assume coherent controlled preparation of these
states with a fixed relative phase. The coherent control acts on the known
preparation circuits; the Hamiltonian is accessed through expectation-value
measurements. Define the signed branch states
\begin{equation}
\ket{\varphi_i^{(s)}(\boldsymbol{\theta})}
:=
\frac{1}{\sqrt{2}}
\left(
\ket{\psi(\boldsymbol{\theta})}
+
s\,\ket{e_i(\boldsymbol{\theta})}
\right),
\; \;
s\in\{+1,-1\},
\label{eq:signed_branch_state}
\end{equation}
State normalization implies \(\mathrm{Re}\braket{\psi|e_i}=0\), so both
branch states have unit norm, including in the complex chart. Write their
energies as
\begin{equation}
\begin{aligned}
E_{\psi}
:=
\bra{\psi(\boldsymbol{\theta})}&H\ket{\psi(\boldsymbol{\theta})},
\\
E_i^{\mathrm{tan}}
:=
\bra{e_i(\boldsymbol{\theta})}&H\ket{e_i(\boldsymbol{\theta})},
\\
E_{\varphi_i^{(s)}}
:=
\bra{\varphi_i^{(s)}(\boldsymbol{\theta})}&H\ket{\varphi_i^{(s)}(\boldsymbol{\theta})}.
\end{aligned}
\label{eq:signed_branch_energies}
\end{equation}

The branch energy contains the average of the two diagonal energies and a
sign-dependent interference term:
\begin{equation}
E_{\varphi_i^{(s)}}
=
\frac{1}{2}\big(E_{\psi}+E_i^{\mathrm{tan}}\big)
+
s\,\mathrm{Re}\!\left[
\bra{e_i(\boldsymbol{\theta})}H\ket{\psi(\boldsymbol{\theta})}
\right].
\label{eq:signed_branch_energy_expansion}
\end{equation}
Rearranging gives, for either sign \(s=\pm1\),
\begin{equation}
\mathrm{Re}\!\left[
\bra{e_i(\boldsymbol{\theta})}H\ket{\psi(\boldsymbol{\theta})}
\right]
=
s\!\left(
E_{\varphi_i^{(s)}}-\tfrac{1}{2}\big(E_{\psi}+E_i^{\mathrm{tan}}\big)
\right).
\label{eq:Re_overlap_single_branch}
\end{equation}

When \(E_\psi\) and \(E_i^{\mathrm{tan}}\) are estimated separately,
either sign-conditioned branch average recovers the transition moment. The
observed sign is a known prefactor, so outcomes with either sign can be retained.
Combining the two branch energies instead cancels the diagonal terms:
\begin{equation}
\mathrm{Re}\!\left[
\bra{e_i(\boldsymbol{\theta})}H\ket{\psi(\boldsymbol{\theta})}
\right]
=
\frac{1}{2}
\Big(
E_{\varphi_i^{(+)}}-E_{\varphi_i^{(-)}}
\Big).
\label{eq:re_overlap_from_signed_energies}
\end{equation}

Substituting \eqref{eq:Re_overlap_single_branch} into the general gradient formula \eqref{eq:hopf_grad_short} gives the single-branch estimator
\begin{equation}
\partial_{\theta_i} E(\boldsymbol{\theta})
=
2\sqrt{g_{i,i}(\boldsymbol{\theta})}\;
s\!\left(
E_{\varphi_i^{(s)}}-\tfrac{1}{2}\big(E_{\psi}+E_i^{\mathrm{tan}}\big)
\right),
\label{eq:single_branch_gradient_estimator}
\end{equation}
valid for either \(s=+1\) or \(s=-1\).
Equivalently, if both branch energies are estimated, one may write
\begin{equation}
\partial_{\theta_i} E(\boldsymbol{\theta})
=
\sqrt{g_{i,i}(\boldsymbol{\theta})}\;
\Big(
E_{\varphi_i^{(+)}}-E_{\varphi_i^{(-)}}
\Big).
\label{eq:symmetric_branch_gradient_estimator}
\end{equation}

The metric factor is known analytically from
Theorem~\ref{thm:hopf_metric_diagonal}. The first estimator uses one baseline
energy \(E_\psi\) and, for each coordinate, a tangent energy and a
sign-conditioned branch energy. The symmetric form uses the difference of the
two branch energies instead. Below we count settings for the first form; the
symmetric form uses the same preparation constructions and has the same
asymptotic setting count.

\subsection*{Magnitude-layer batching}

The estimators above refer to one coordinate at a time. At a fixed tree depth,
an index register lets one compiled circuit family serve all tangents in the
layer, and a second family serves their signed branches. Here a \emph{setting}
means a fixed compiled gate pattern and indexed preparation choice, with the
Hopf angles supplied as physical rotation parameters. An \emph{execution} is
one preparation-and-measurement run. Its observed label identifies the
component informed by that run.

For each depth \(d\in\{0,1,\ldots,n-1\}\), define the internal-node layer
\begin{equation}
\begin{aligned}
\mathcal{L}_d
:&=
\{\, i\in\{1,\ldots,N-1\} : \lfloor\log_2 i\rfloor = d \,\}\\
&=
\{2^d,2^d+1,\ldots,2^{d+1}-1\}.
\end{aligned}
\label{eq:def_layer_i}
\end{equation}
Introduce a reduced index register \(I\) of \(d\) qubits storing the layer offset
\begin{equation}
r := i-2^d \in \{0,\ldots,2^d-1\},
\; \;
i=2^d+r \in \mathcal{L}_d.
\label{eq:reduced_layer_index}
\end{equation}
Thus \(\ket{r}_I\) labels internal node \(i=2^d+r\).

Prepare the uniform superposition over \(\mathcal{L}_d\),
\begin{equation}
\ket{\mathrm{unif}(\mathcal{L}_d)}_I
:=
\frac{1}{\sqrt{|\mathcal{L}_d|}}
\sum_{r=0}^{2^d-1}\ket{r}_I,
\; \;
|\mathcal{L}_d|=2^d.
\label{eq:unif_layer}
\end{equation}
For \(d\ge 1\), this is \(\ket{+}^{\otimes d}\); for \(d=0\), the index register is empty and the layer consists only of the root node.

Using the explicit synthesis maps \(i\mapsto \boldsymbol{\theta}^{(i)}\), we implement a coherent index-controlled preparation on \(I\otimes\mathcal{H}\) that maps
\[
\ket{r}_I\ket{0}^{\otimes n}
\longmapsto
\ket{r}_I\ket{e_{2^d+r}(\boldsymbol{\theta})}
\; \;
\forall\,r\in\{0,\ldots,2^d-1\}.
\]
The index selects the shifted or clamped assignments of
Theorem~\ref{thm:synthesis_normalized_direction_hopf} within the same Hopf gate
skeleton. Each gate location has only a constant number of possible angle
assignments, so it requires only a constant number of index-conditioned
multi-controlled rotations, without angle multiplexing.

Acting on \eqref{eq:unif_layer}, this preparation yields the indexed tangent-state preparation
\begin{equation}
\ket{\Phi_d(\boldsymbol{\theta})}
:=
\frac{1}{\sqrt{|\mathcal{L}_d|}}
\sum_{r=0}^{2^d-1}
\ket{r}_I\ket{e_{2^d+r}(\boldsymbol{\theta})}.
\label{eq:indexed_derivative_layer}
\end{equation}
Measuring \(I\) in the computational basis labels the system's tangent state.
Standard Pauli measurements on the system then provide samples of
\[
E_i^{\mathrm{tan}}
=
\bra{e_i(\boldsymbol{\theta})}H\ket{e_i(\boldsymbol{\theta})},
\qquad
i=2^d+r,
\]
for the observed \(i\in\mathcal{L}_d\). Sorting outcomes by label gives a
separate conditional average for each coordinate, recovering all tangent
energies in the ideal-expectation limit.

To obtain branch energies with the same layer labels, add one ancilla qubit
\(a\) and prepare
\begin{equation}
\begin{aligned}
&\ket{\Omega_d(\boldsymbol{\theta})}\\
:=
&\frac{1}{\sqrt{|\mathcal{L}_d|}}
\sum_{r=0}^{2^d-1}
\ket{r}_I
\left(
\frac{
\ket{0}_a\ket{\psi(\boldsymbol{\theta})}
+
\ket{1}_a\ket{e_{2^d+r}(\boldsymbol{\theta})}
}{\sqrt{2}}
\right).
\end{aligned}
\label{eq:indexed_branch_layer}
\end{equation}
Measuring \(a\) in the \(X\) basis gives a sign \(s\in\{+1,-1\}\).
Conditioned on this sign and the observed index \(r\), the system is in
\(\ket{\varphi_i^{(s)}(\boldsymbol{\theta})}\) with \(i=2^d+r\).

The ancilla readout uses a final Hadamard followed by a \(Z\)-basis measurement:
the measured bit \(b\) gives \(s=(-1)^b\). The pair \((r,s)\) and the
Hamiltonian measurement outcome contribute to
Eq.~\eqref{eq:single_branch_gradient_estimator} for the labeled coordinate.
Every execution is used, with no postselection on either label or sign.

\subsection*{Leaf-indexed batching for complex phases}

The complex Hopf ansatz contains an additional block of leaf-phase parameters
\[
\theta_{N+\ell},
\qquad
\ell=0,\ldots,N-1,
\]
whose normalized tangents are supported on their corresponding leaves
(Theorem~\ref{thm:synthesis_normalized_direction_hopf}). The leaf index therefore
replaces the magnitude-layer index, and one batch serves the entire phase block.

Introduce an \(n\)-qubit phase-index register \(P\), with the data register's
MSB-to-LSB ordering, and prepare
\begin{equation}
\ket{\mathrm{unif}_{\mathrm{ph}}}_P
:=
\frac{1}{\sqrt{N}}
\sum_{\ell=0}^{N-1}\ket{\ell}_P.
\label{eq:def_phase_uniform}
\end{equation}
Using the explicit synthesis map
\(
\ell\mapsto \boldsymbol{\theta}^{(N+\ell)}
\)
from Theorem~\ref{thm:synthesis_normalized_direction_hopf}(B), we implement a coherent index-controlled preparation on \(P\otimes\mathcal{H}\) such that
\[
\ket{\ell}_P\ket{0}^{\otimes n}
\longmapsto
\ket{\ell}_P\ket{e_{N+\ell}(\boldsymbol{\theta})}
\; \;
\forall\,\ell\in\{0,\ldots,N-1\}.
\]
This yields the indexed phase-derivative state
\begin{equation}
\ket{\Phi_{\mathrm{ph}}(\boldsymbol{\theta})}
:=
\frac{1}{\sqrt{N}}
\sum_{\ell=0}^{N-1}
\ket{\ell}_P
\ket{e_{N+\ell}(\boldsymbol{\theta})}.
\label{eq:def_Phi_phase}
\end{equation}
As in the magnitude construction, measuring the index register \(P\) labels
the tangent, and Pauli measurements on the system provide samples of
\[
E_{N+\ell}^{\mathrm{tan}}
=
\bra{e_{N+\ell}(\boldsymbol{\theta})}
H
\ket{e_{N+\ell}(\boldsymbol{\theta})}
\]
for the observed leaf \(\ell\). Conditional averages for each label recover
the phase-block tangent energies in the ideal-expectation limit.

The signed-branch family uses one ancilla qubit \(a\):
\begin{equation}
\begin{aligned}
&\ket{\Omega_{\mathrm{ph}}(\boldsymbol{\theta})}\\
:=
&\frac{1}{\sqrt{N}}
\sum_{\ell=0}^{N-1}
\ket{\ell}_P
\left(
\frac{
\ket{0}_a\ket{\psi(\boldsymbol{\theta})}
+
\ket{1}_a\ket{e_{N+\ell}(\boldsymbol{\theta})}
}{\sqrt{2}}
\right).
\end{aligned}
\label{eq:def_Omega_phase}
\end{equation}
Measuring \(a\) in the \(X\) basis and \(P\) in the computational basis yields, for the observed label \(\ell\) and sign \(s\in\{+1,-1\}\), label- and sign-conditioned samples of
\[
E_{\varphi_{N+\ell}^{(s)}}
=
\bra{\varphi_{N+\ell}^{(s)}(\boldsymbol{\theta})}
H
\ket{\varphi_{N+\ell}^{(s)}(\boldsymbol{\theta})}.
\]
The pair \((\ell,s)\) and the Hamiltonian outcome feed the same
single-branch estimator \eqref{eq:single_branch_gradient_estimator}, now for
coordinate \(N+\ell\). All labels and both signs are retained.

\subsection*{Complete estimator and scaling analysis}

The accounting separates the charge of one compiled Hopf preparation,
the number of preparation settings, and the measurement repetitions needed
for a chosen accuracy. Let \(C_{\mathrm{cfg}}\) denote the assigned logical
CNOT charge of one such preparation under
Lemma~\ref{lem:cnot_model}. The explicit index and branch registers are
part of the protocol; the no-clean-ancilla condition concerns additional
workspace for decomposing its controlled gates. Observable-dependent
readout operations and measurement-basis choices are counted separately.

\begin{theorem}[Uniform cost of Hopf gradient settings]
\label{thm:hopf_gradient_setting_cost}
Under the assigned CNOT ledger of Lemma~\ref{lem:cnot_model}, every compiled
Hopf preparation in the gradient-access protocol has the same asymptotic
assigned charge as a forward Hopf preparation. This covers the
indexed tangent and branch settings \(\ket{\Phi_d(\boldsymbol{\theta})}\) and
\(\ket{\Omega_d(\boldsymbol{\theta})}\) of
Eqs.~\eqref{eq:indexed_derivative_layer} and~\eqref{eq:indexed_branch_layer}
for every magnitude depth \(d\in\{0,\ldots,n-1\}\) in both the real and
complex ansatze. In the complex ansatz it also covers
\(\ket{\Phi_{\mathrm{ph}}(\boldsymbol{\theta})}\) and
\(\ket{\Omega_{\mathrm{ph}}(\boldsymbol{\theta})}\) of
Eqs.~\eqref{eq:def_Phi_phase} and~\eqref{eq:def_Omega_phase}. In all cases,
\[
C_{\mathrm{cfg}}(\text{one gradient setting})=\bigO(nN).
\]
\end{theorem}

Write \(N_{\mathrm{set}}^{\mathbb R}\) and
\(N_{\mathrm{set}}^{\mathbb C}\) for the numbers of Hopf preparation
settings in the real and complex protocols. Since \(N=2^n\), the magnitude
tree has exactly \(n=\log_2 N\) layers. For the real ansatz, the
baseline-plus-tangent estimator uses one baseline setting and, for each
magnitude layer, one tangent setting and one branch setting:
\begin{equation}
\boxed{
N_{\mathrm{set}}^{\mathbb{R}}
=
1+2n
=
\bigO(\log N).
}
\label{eq:setting_count_real}
\end{equation}
For the complex ansatz, the magnitude block contributes the same \(2n\)
settings and the phase block adds one tangent and one branch setting:
\[
\boxed{
N_{\mathrm{set}}^{\mathbb{C}}
=
1+2(n+1)
=
\bigO(\log N).
}
\]
Since the number of magnitude parameters is \(N-1\), these setting counts are
also logarithmic in the magnitude-parameter count.

For a baseline energy evaluation, \(C_{\mathrm{cfg}}(E)\) charges its
forward Hopf preparation. The native-ordering count of
Theorem~\ref{thm:hopf_cnot} gives
\begin{equation}
\boxed{
C_{\mathrm{cfg}}(E)
=
\bigO(nN)
=
\bigO(N\log N).
}
\label{eq:cfg_energy_scaling}
\end{equation}

The indexed magnitude and phase preparations reuse the forward Hopf structure
through subtree-local tangent synthesis. By
Theorem~\ref{thm:hopf_gradient_setting_cost}, each gradient-access setting has
the same asymptotic assigned charge as a forward setting. Let
\(C_{\mathrm{grad}}(\nabla E)\) denote the sum of assigned preparation
charges over the distinct gradient-access settings, each counted once. Then
\begin{equation}
\boxed{
C_{\mathrm{grad}}(\nabla E)
=
\bigO(n)\,C_{\mathrm{cfg}}(E)
=
\bigO(\log N)\,C_{\mathrm{cfg}}(E).
}
\label{eq:cfg_gradient_scaling}
\end{equation}

These bounds count distinct Hopf preparation families. Each execution is
assigned to a component by its measured index. The execution budget for a
prescribed gradient accuracy depends on the observable decomposition,
label--sign sampling probabilities, estimator variances, and error
criterion. Observable-dependent readout configurations are counted
separately.

Related approaches reduce gradient overhead through algebraic circuit
structure, density-based models, or expressivity--measurement tradeoffs
\cite{abbas2023quantum,bowles2025backprop,coyle2025density,chinzei2025tradeoff}.
The Hopf construction organizes the required gradient settings within the same
circuit skeleton through exact tangent-state synthesis and analytic binary-tree
geometry, with the attainable precision of each component set by the associated
measurement statistics.

Although stated for the VQE objective
\(E(\boldsymbol{\theta})=\bra{\psi}H\ket{\psi}\), the
construction is not tied to Hamiltonian minimization. The objective derivative need only reduce to transition moments of a
Hermitian observable \(O\),
\begin{equation}
\operatorname{Re}\!\left[
\bra{e_i(\boldsymbol{\theta})}
O
\ket{\psi(\boldsymbol{\theta})}
\right].
\end{equation}
The same framework applies, after local linearization, to metrological costs
built from finitely many expectation values, probabilities, or probability
phase-derivatives, including local QFI and fixed-readout CFI objectives
\cite{braunstein1994statistical,paris2009quantum,giovannetti2011advances,degen2017quantum,pezze2018quantum}.
Appendix~\ref{app:metrology} gives two local examples: pure-state QFI as an
ultimate probe-state benchmark and fixed-readout Ramsey CFI
objective~\cite{ramsey1950molecular} for a specified parity readout at a
calibrated operating point. In both, Hopf tangent preparation and
compiled-setting organization are unchanged after replacing \(H\) by the
metrological chain-rule observable. The optimal measurement, estimator, noise
model, and adaptive strategy remain specific to the problem.

\section{Numerical simulations}
\label{sec:numerical_simulation}

The preceding section explains how Hopf coordinate derivatives can be accessed
from circuits. Here we isolate how the geometry uses those derivatives in an
optimization step. All runs are deterministic real-state simulations, with
objectives and first derivatives evaluated exactly from the current state.
The comparison therefore tests the update rules of
Appendix~\ref{app:optimizers}, rather than the finite-shot measurement
protocols of Sec.~\ref{sec:gradient_access}.

We compare three metric-aware Hopf optimizers with two coordinate-Adam
baselines. All five methods use the same stress instances and exact
cost-and-gradient information. The distinction is how that information is
converted into a state update: the geometric methods use the diagonal Hopf
metric, state-sphere geodesics, and vector transport, whereas the Adam
baselines update circuit coordinates directly.

\begin{figure*}[t]
\centering
\includegraphics[width=\textwidth]{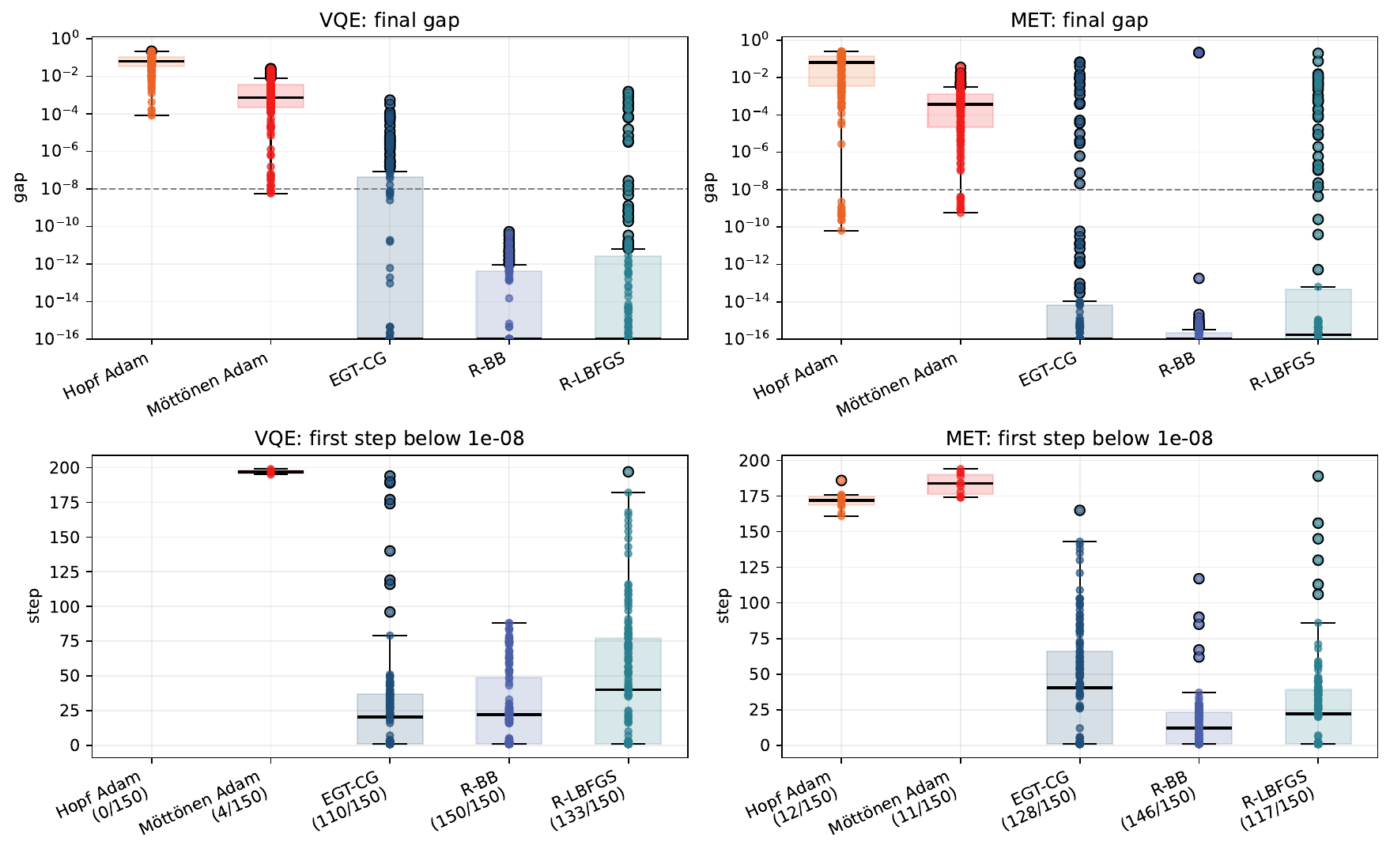}
\caption{
Aggregate final precision and threshold hitting over all stress tests and
system sizes $n=6,\ldots,10$. The top row shows final optimality gaps after
$200$ accepted updates; the gray dashed line marks the $10^{-8}$ reference
threshold. The bottom row shows the first accepted update at which each crossing trace
reaches that threshold; the parenthetical labels report the number crossing
within $200$ updates out of the total number of traces. In every box plot, the box spans the interquartile range, the horizontal line marks the median, and the whiskers extend to the most extreme values within $1.5$ interquartile ranges. Colored circles are the individual
task-size-seed traces. Optimizer colors are used consistently in this figure and
Fig.~\ref{fig:hopf_geometric_n10_convergence}: orange for Hopf-Adam, red for
M\"ott\"onen-ideal-PS-Adam, dark blue for EGT-CG, indigo for R-BB, and teal for
R-LBFGS. Lower values are better. The metric-aware Hopf methods concentrate
near the numerical floor in VQE and have numerical-precision medians in the
metrology-inspired aggregate; their main distinction is the size of the
slow-convergence tail.
}
\label{fig:hopf_geometric_summary}
\end{figure*}

\begin{figure*}[t]
\centering
\includegraphics[width=\textwidth]{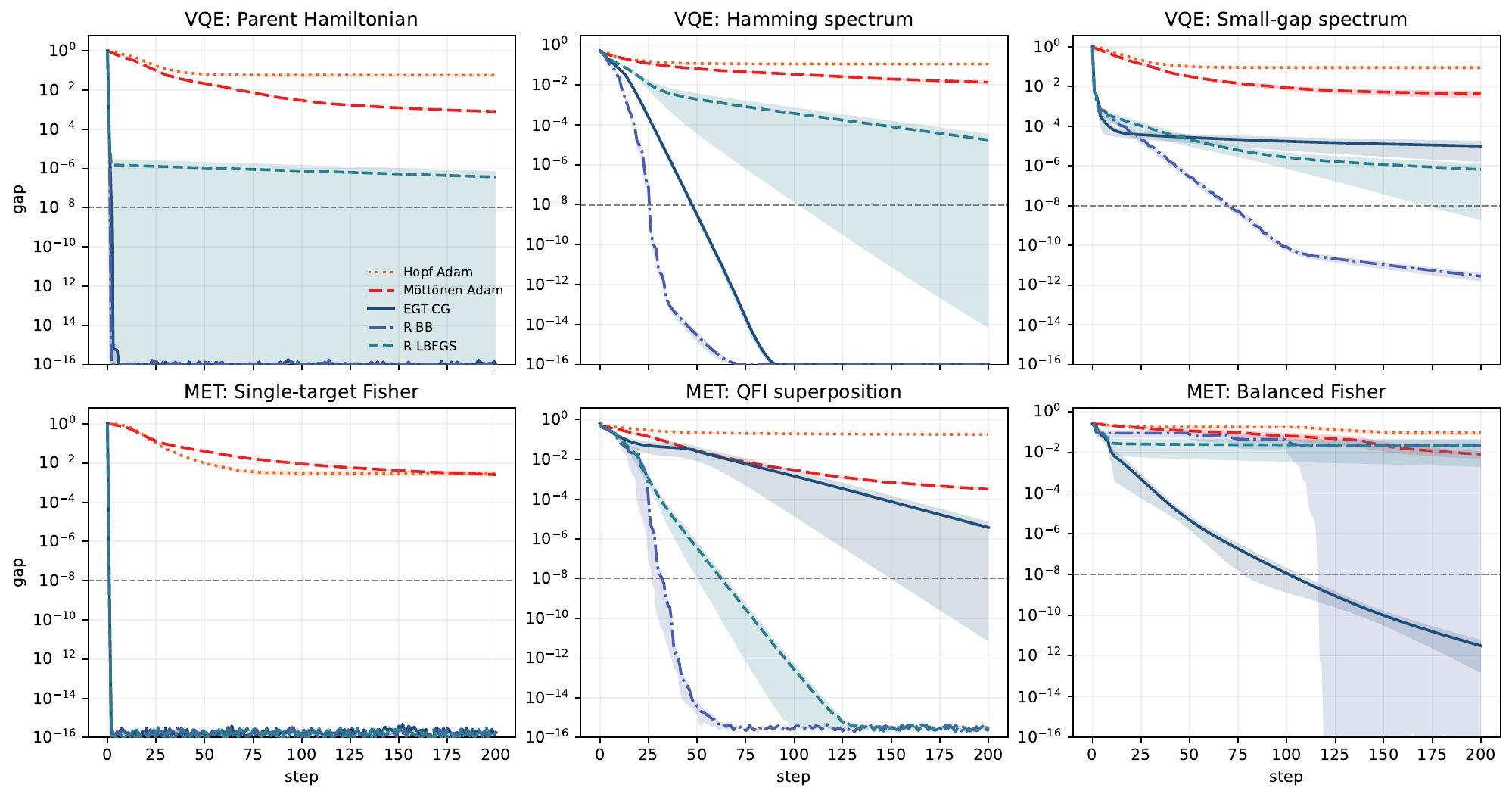}
\caption{
Per-task convergence at $n=10$, with ten deterministic pseudo-random
initial-state seeds per task. The top row shows the parent-Hamiltonian,
scrambled-Hamming-spectrum, and small-gap-spectrum VQE tasks. The bottom row
shows the single-target Fisher proxy, QFI extremal-superposition, and balanced
two-target Fisher proxy tasks. Each curve is the mean optimality gap across the ten
seeds, and each translucent band is the corresponding standard error of the
mean. Optimizer colors match Fig.~\ref{fig:hopf_geometric_summary}, and the
legend's line styles are used consistently in all six panels. The gray dashed
line marks the $10^{-8}$ reference threshold. The vertical scale is logarithmic
and lower values are better. Because several metric-aware traces reach the
plotting floor, these panels are primarily per-task convergence diagnostics;
the final-gap and threshold-hit distributions are summarized in
Fig.~\ref{fig:hopf_geometric_summary}.
}
\label{fig:hopf_geometric_n10_convergence}
\end{figure*}

The geometric methods share an update pipeline. The analytic diagonal metric
maps coordinate gradients to state-sphere gradients, the sphere exponential
map generates normalized trial states, and the inverse Hopf map returns
accepted states to circuit angles.

Hopf-EGT-CG (exact geodesic transport with conjugate gradients) implements the update of
Ref.~\cite{ferreira2025quantum}. Hopf-Riemannian-BB (R-BB) uses a Riemannian
Barzilai--Borwein spectral-gradient step~\cite{barzilai1988two}, with
transported two-point secant information setting the step scale.
Hopf-Riemannian-LBFGS (R-LBFGS) uses a limited-memory quasi-Newton update on
the state sphere~\cite{liu1989limited,absil2008optimization}.

Hopf-Adam applies Adam~\cite{kingma2014adam} directly to the Hopf coordinates,
using the same exact coordinate gradient as the three geometric methods. It
uses the chart without the induced metric or state-sphere updates.

M\"ott\"onen-ideal-PS-Adam (``PS'' denotes parameter shift)
instead applies Adam to the physical post-multiplexing \(R_y\) angles of the
real M\"ott\"onen state-preparation
circuit~\cite{mottonen2004transformation}. Its gradient is the exact
parameter-shift-equivalent gradient in those physical angles. Both Adam
baselines use adaptive cost-only backtracking for trial steps.

All methods have the same accepted-update budget and cost-plus-gradient
information interface. The curves compare convergence per accepted update;
evaluation and sampling costs are separate resource measures.

\subsection*{Stress-test design}

Each task is scrambled by a fixed real orthogonal circuit to test optimizer
behavior on generic real states rather than exploit problem-specific
structure. We run five system sizes, \(n=6,7,8,9,10\), with ten
deterministic pseudo-random initial-state seeds per fixed task instance. Each
run records \(201\) states, including the initial state and \(200\)
accepted updates. Three tasks per application class give
\(5\times3\times10=150\) traces per optimizer and class. The scrambling circuit and objective stay fixed for each task; the seeds
vary only the initial variational state.

The VQE stress tests are:
\begin{itemize}
    \item \textbf{Parent Hamiltonian:}
    \(H=I-\ket{\tau}\!\bra{\tau}\), where \(\ket{\tau}\) is a scrambled target
    state.  The optimal energy gap is zero.

    \item \textbf{Scrambled Hamming spectrum:}
    a normalized diagonal Hamming-distance spectrum is conjugated by the real
    scrambling circuit.  The unique ground state is again a scrambled basis
    state.

    \item \textbf{Small-gap scrambled spectrum:}
    a scrambled diagonal Hamiltonian with one ground state, one nearby excited
    state separated by a gap \(10^{-2}\), and all other levels at unit energy.
\end{itemize}

The metrology-inspired stress tests are:
\begin{itemize}
    \item \textbf{Single-target Fisher proxy:}
    for the probe density operator
    \(\rho=\ket{\psi}\!\bra{\psi}\), set
    \(F=\langle\tau|\rho|\tau\rangle^2\) and minimize \(1-F\).

    \item \textbf{QFI extremal-superposition objective:}
    for a scrambled Hermitian generator \(G\), define
    \(\Delta_G=\lambda_{\max}(G)-\lambda_{\min}(G)>0\), where
    \(\lambda_{\max}\) and \(\lambda_{\min}\) are its largest and smallest
    eigenvalues.
    The cost is \(1-4\operatorname{Var}_{\psi}(G)/\Delta_G^2\), where
    \(\operatorname{Var}_{\psi}(G)=\bra{\psi}G^2\ket{\psi}
    -(\bra{\psi}G\ket{\psi})^2\). An equal-weight superposition of extremal eigenstates is optimal.

    \item \textbf{Balanced two-target Fisher proxy:}
    two orthonormal targets \(\ket{\tau_1}\) and \(\ket{\tau_2}\) are
    obtained by scrambling distinct basis states. Define their occupation
    probabilities \(e_k=|\braket{\tau_k|\psi}|^2\) and proxy values
    \(F_k=e_k^2\), for \(k=1,2\). With smoothing parameter \(\beta=20\),
    set
    \[
    F_{\mathrm{bal}}=-\frac{1}{\beta}
    \log\!\left(\frac{e^{-\beta F_1}+e^{-\beta F_2}}{2}\right).
    \]
    The minimized cost is \(1/4-F_{\mathrm{bal}}\); its minimum is zero
    at \(e_1=e_2=1/2\).
\end{itemize}

The single-target and balanced-target tasks use synthetic Fisher-inspired
proxies. The third task uses pure-state QFI, and
Appendix~\ref{app:metrology} gives the operational fixed-readout CFI
construction. VQE panels show energy gaps, and metrology
panels show proxy or normalized-QFI optimality gaps, with lower values
better throughout. Roundoff-level negative
final values are clipped to a positive floor on logarithmic axes. All scheduled
task-size-seed traces are included in the aggregate statistics.

Zero-weight subtrees make some Hopf angles nonunique, and \(g_{i,i}=0\)
marks directions outside the regular tangent chart. The real-state simulations
use the inverse-map convention of Appendix~\ref{app:algorithms}, with angle
clipping (\(10^{-6}\) for non-final angles and \(10^{-9}\) near zeros
of final-layer sine or cosine factors). These latter zeros need not be
singularities of the real chart: a nonzero sibling-pair weight still gives
a nonzero metric entry. A boundary-safe tree contraction gives the coordinate
gradients. The state-gradient lift uses a \(10^{-12}\) floor for metric
denominators.

\subsection*{Aggregate results}

Figure~\ref{fig:hopf_geometric_summary} shows two distinct patterns. In VQE,
metric-aware updates substantially reduce mean final gaps relative to the
coordinate-Adam baselines. In the metrology-inspired tests, the main
distinction is the distribution of final gaps: the geometric methods place
more traces at numerical precision, while the stronger coordinate baseline
has comparable or smaller mean final gaps. Table~\ref{tab:hopf_geometric_final_gaps}
reports the mean and median final gaps for all three geometric optimizers.

\begin{table}[tb]
\centering
\small
\begin{tabular}{@{}lcc@{}}
\toprule
Optimizer & Mean final gap & Median final gap \\
\midrule
\multicolumn{3}{c}{VQE} \\
\midrule
EGT-CG & \(9.23\times10^{-6}\) & \(1.02\times10^{-25}\) \\
R-BB & \(1.61\times10^{-12}\) & \(1.02\times10^{-25}\) \\
R-LBFGS & \(3.21\times10^{-5}\) & \(5.30\times10^{-17}\) \\
\midrule
\multicolumn{3}{c}{Metrology-inspired} \\
\midrule
EGT-CG & \(1.56\times10^{-3}\) & \(0\) \\
R-BB & \(5.74\times10^{-3}\) & \(0\) \\
R-LBFGS & \(2.36\times10^{-3}\) & \(1.67\times10^{-16}\) \\
\bottomrule
\end{tabular}
\caption{Final optimality gaps for the metric-aware Hopf optimizers, aggregated
across \(150\) task-size-seed traces per application class after \(200\)
accepted updates. The entries are floating-point residuals before
application of the positive plotting floor. Their smallest values depend on
how the objective is evaluated.}
\label{tab:hopf_geometric_final_gaps}
\end{table}

For VQE, R-BB gives the smallest mean final gap,
\(1.61\times10^{-12}\). EGT-CG and R-LBFGS also reach numerical-precision
median gaps, but their means are larger because of slow-convergence tails.
The stronger coordinate baseline, M\"ott\"onen-ideal-PS-Adam, has mean final
gap \(3.47\times10^{-3}\) and median \(7.35\times10^{-4}\); Hopf-Adam has
mean final gap \(7.01\times10^{-2}\). Both baselines therefore remain
several orders of magnitude above the best geometric mean gap. The
comparison with Hopf-Adam keeps the chart and coordinate gradients fixed,
highlighting the role of the metric-aware update rather than the coordinates
alone.

In the metrology-inspired tests, the geometric medians are at roundoff
scale, while slow or stalled runs set the means.
M\"ott\"onen-ideal-PS-Adam has mean final gap \(1.53\times10^{-3}\), close
to EGT-CG's \(1.56\times10^{-3}\) and below those of R-BB and R-LBFGS,
but a larger median of \(3.67\times10^{-4}\). The largest mean final gap,
\(7.96\times10^{-2}\), occurs with Hopf-Adam.

\subsection*{Per-task convergence}

Figure~\ref{fig:hopf_geometric_n10_convergence} shows where the aggregate
differences arise at \(n=10\). The parent-Hamiltonian and single-target Fisher
objectives are the easiest tasks: most Hopf-geometric traces rapidly reach the
target scale. On the Hamming, small-gap, QFI, and balanced Fisher tasks, the
mean curves often overlap after the fast initial descent. The important
differences are therefore the slow tails and the traces that fail to cross the
\(10^{-8}\) threshold within \(200\) accepted updates. The balanced two-target Fisher objective is the clearest nonlinear test
of slow-tail behavior.

The Hopf-Adam comparison tests the use of geometry within a fixed chart.
Differences among EGT-CG, R-BB, and R-LBFGS concern optimizer robustness
within that common coordinate-and-metric framework.

\section{Conclusion}

The Hopf ansatz connects arbitrary pure-state preparation to optimization
through one executable coordinate system. Its inverse map, analytic metric,
and preparable tangent directions let the same tree translate between
state-space motion, circuit angles, and measurable objective responses.
The indexed construction organizes gradient access into logarithmically
many Hopf preparation families, with sampling requirements determined
separately by the target accuracy and measurement statistics.

The deterministic real-state tests illustrate the optimizer role of this
geometry. The tested metric-aware pipelines improve mean final VQE gaps
relative to coordinate-Adam baselines. In the metrology-inspired tests,
they place more runs at numerical precision, while the stronger coordinate
baseline has comparable or smaller mean final gaps. These comparisons
concern complete update procedures on the specified tasks.

The framework applies most directly when the search object is an arbitrary
normalized quantum state vector and the quantity being optimized admits a usable
local derivative formula. The same geometric role can be explored in state
discrimination \cite{bae2015quantum}, variational quantum hypothesis testing
\cite{subramanian2024shallow}, variational sensing and metrology
\cite{koczor2020variational,kaubruegger2019variational,kaubruegger2021quantum,meyer2021variational,le2023variational,yang2021hybrid,maclellan2024end},
quantum control \cite{rabitz2000whither,glaser2015training}, and variational
compilation \cite{khatri2019quantumassisted}. Across these settings,
objectives, readouts, and performance gains remain problem-specific, but the
need for an executable bridge between state-space geometry and circuit
parameters is shared.

In the companion paper's global quantum backpropagation
construction~\cite{lin2026hopfqbp}, the inverse Hopf frame resolves the
objective response along labeled magnitude directions. Each joint all-$X$
measurement outcome informs every magnitude derivative, while a separate
signed leaf record supplies the complex phase derivatives. The distinction
is how the information is read out: indexed component records in the
present work, shared magnitude records under the mirror assumption.
Both use the same state coordinates and differential geometry.

The main message is structural. When arbitrary-state freedom is part of the
model, the circuit can carry its own coordinate map, metric, and differential
probes. The Hopf ansatz turns universal state preparation into a navigable circuit
framework, a compass for moving on the quantum state sphere.

\vspace{0.5cm}

\begin{acknowledgments}
The authors are sincerely grateful to Allan Tosta, Andr\'e J. Ferreira-Martins,
Thiago O. Maciel, and Tobias Haug for thoughtful discussions and perspectives
that helped shape this work, and to Leandro Aolita and the Technology Innovation
Institute (TII) for their valued support of the project. 

Special thanks are due to Giancarlo Camilo and Renato M. S. Farias for sustained
guidance, encouragement, and perspective, and to Ranyiliu Chen and Xin Wang for
help with the final revision.

Guangxi Li was supported by National Natural Science Foundation of China (Grant No. 62402206). This work was partially supported by the Guangdong Provincial Quantum Science Strategic Initiative (Grant Nos.~GDZX2403008, GDZX2503001, and GDZX2403001).
\end{acknowledgments}

\vspace{0.5cm}

\section*{Data and code availability}

The manuscript uses analytic derivations and deterministic state-vector
simulations; no experimental data were used. The accompanying public
repository~\cite{lin2026hopfansatzcode} provides documented commands to
regenerate the reported real-Hopf stress-test data, diagnostics, and figures.
It also includes auxiliary checks: an \(n=6\) complex-Hopf stress test using the
same three VQE and three metrology-inspired task families, finite-shot
signed-branch gradient tests, assigned CNOT-count checks, and Qibo~\cite{efthymiou2022qibo}
circuit-level gradient-access examples for local four-qubit real and complex
VQE toys. 

Generated CSV files are omitted from source control because they are derived
artifacts and can be large, but they can be regenerated from the documented
commands. Selected figures and diagnostic panels are included as release
artifacts for direct comparison with the manuscript results.

\vspace{0.3cm}

\section*{Generative AI disclosure}

The idea behind the Hopf ansatz and its binary-tree construction originated
with the first author, with the tree first drawn on paper before any AI
assistance. Over more than a year, OpenAI's ChatGPT assisted at successive
stages as the construction took shape in code, proofs, and the manuscript.
GPT-4 was initially used for coding, with its role gradually extending to
verification and proof formulation. GPT-5.5 Pro subsequently assisted with
manuscript polishing, and GPT-5.6 Pro and GPT-6 Pro assisted with the final
revision. The dialogue served as a working blackboard for implementing the
construction, examining arguments, and refining their presentation. The
authors reviewed the mathematical arguments and ran the code, and retain
full responsibility for the scientific content and conclusions.

\bibliographystyle{apsrev4-2}
\bibliography{references}

@article{farias2025quantum,
   title={{\href{https://doi.org/10.1103/PhysRevApplied.23.044014}{Quantum encoder for fixed-Hamming-weight subspaces}}},
   journal={Physical Review Applied},
   author={Farias, Renato MS and Maciel, Thiago O and Camilo, Giancarlo and Lin, Ruge and Ramos-Calderer, Sergi and Aolita, Leandro},
   year={2025},
   DOI={10.1103/PhysRevApplied.23.044014},
   url={https://doi.org/10.1103/PhysRevApplied.23.044014}
}

@misc{ferreira2025quantum,
  author={Ferreira-Martins, Andr{\'e} J. and Farias, Renato M. S. and Camilo, Giancarlo and Maciel, Thiago O. and Tosta, Allan and Lin, Ruge and Alhajri, Abdulla and Haug, Tobias and Aolita, Leandro},
  archivePrefix={arXiv},
  eprint={2506.17395},
  primaryClass={quant-ph},
  year={2025}
}

@article{schobel2016orthogonal,
   title={{\href{https://doi.org/10.3842/SIGMA.2016.041}{Are orthogonal separable coordinates really classified?}}},
   journal={SIGMA. Symmetry, Integrability and Geometry: Methods and Applications},
   author={Sch{\"o}bel, Konrad},
   year={2016},
   DOI={10.3842/SIGMA.2016.041},
   url={https://doi.org/10.3842/SIGMA.2016.041}
}

@article{cohl2013fourier,
   title={{\href{https://doi.org/10.3842/SIGMA.2013.042}{Fourier, Gegenbauer and Jacobi expansions for a power-law fundamental solution of the polyharmonic equation and polyspherical addition theorems}}},
   journal={SIGMA. Symmetry, Integrability and Geometry: Methods and Applications},
   author={Cohl, Howard S},
   year={2013},
   DOI={10.3842/SIGMA.2013.042},
   url={https://doi.org/10.3842/SIGMA.2013.042}
}

@article{sun2023asymptotically,
   title={{\href{https://doi.org/10.1109/TCAD.2023.3244885}{Asymptotically optimal circuit depth for quantum state preparation and general unitary synthesis}}},
   journal={IEEE Transactions on Computer-Aided Design of Integrated Circuits and Systems},
   author={Sun, Xiaoming and Tian, Guojing and Yang, Shuai and Yuan, Pei and Zhang, Shengyu},
   year={2023},
   DOI={10.1109/TCAD.2023.3244885},
   url={https://doi.org/10.1109/TCAD.2023.3244885}
}

@article{stokes2020quantum,
   title={{\href{https://doi.org/10.22331/q-2020-05-25-269}{Quantum natural gradient}}},
   journal={Quantum},
   author={Stokes, James and Izaac, Josh and Killoran, Nathan and Carleo, Giuseppe},
   year={2020},
   DOI={10.22331/q-2020-05-25-269},
   url={https://doi.org/10.22331/q-2020-05-25-269}
}

@article{abbas2023quantum,
  title={{\href{https://dl.acm.org/doi/10.5555/3666122.3668062}{On quantum backpropagation, information reuse, and cheating measurement collapse}}},
  journal={Advances in Neural Information Processing Systems},
  author={Abbas, Amira and King, Robbie and Huang, Hsin-Yuan and Huggins, William J and Movassagh, Ramis and Gilboa, Dar and McClean, Jarrod},
  year={2023},
  DOI={10.5555/3666122.3668062},
  url={https://dl.acm.org/doi/10.5555/3666122.3668062}
}

@article{vale2023circuit,
  title={{\href{https://dl.acm.org/doi/10.1109/TCAD.2023.3327102}{Circuit decomposition of multicontrolled special unitary single-qubit gates}}},
  journal={IEEE Transactions on Computer-Aided Design of Integrated Circuits and Systems},
  author={Vale, Rafaella and Azevedo, Thiago Melo D and Ara{\'u}jo, Ismael CS and Araujo, Israel F and Da Silva, Adenilton J},
  year={2023},
  DOI={10.1109/TCAD.2023.3327102},
  url={https://dl.acm.org/doi/10.1109/TCAD.2023.3327102}
}

@misc{shende2008cnot,
  author={Shende, Vivek V. and Markov, Igor L.},
  archivePrefix={arXiv},
  eprint={0803.2316},
  primaryClass={quant-ph},
  year={2008}
}

@article{barenco1995elementary,
  title={{\href{https://doi.org/10.1103/PhysRevA.52.3457}{Elementary gates for quantum computation}}},
  journal={Physical review A},
  author={Barenco, Adriano and Bennett, Charles H and Cleve, Richard and DiVincenzo, David P and Margolus, Norman and Shor, Peter and Sleator, Tycho and Smolin, John A and Weinfurter, Harald},
  year={1995},
  DOI={10.1103/PhysRevA.52.3457},
  url={https://doi.org/10.1103/PhysRevA.52.3457}
}

@misc{mottonen2004transformation,
  author={M{\"o}tt{\"o}nen, Mikko and Vartiainen, Juha J. and Bergholm, Ville and Salomaa, Martti M.},
  archivePrefix={arXiv},
  eprint={quant-ph/0407010},
  year={2004}
}

@article{shende2005synthesis,
  title={{\href{https://doi.org/10.1109/TCAD.2005.855930}{Synthesis of quantum-logic circuits}}},
  journal={IEEE Transactions on Computer-Aided Design of Integrated Circuits and Systems},
  author={Shende, Vivek V and Bullock, Stephen S and Markov, Igor L},
  year={2005},
  DOI={10.1109/TCAD.2005.855930},
  url={https://doi.org/10.1109/TCAD.2005.855930}
}

@article{plesch2011quantum,
  title={{\href{https://doi.org/10.1103/PhysRevA.83.032302}{Quantum-state preparation with universal gate decompositions}}},
  journal={Physical Review A},
  author={Plesch, Martin and Brukner, {\v{C}}aslav},
  year={2011},
  DOI={10.1103/PhysRevA.83.032302},
  url={https://doi.org/10.1103/PhysRevA.83.032302}
}

@misc{knill1995approximation,
  author={Knill, Emanuel},
  archivePrefix={arXiv},
  eprint={quant-ph/9508006},
  year={1995}
}

@article{zhang2022quantum,
  title={{\href{https://doi.org/10.1103/PhysRevLett.129.230504}{Quantum state preparation with optimal circuit depth: Implementations and applications}}},
  journal={Physical Review Letters},
  author={Zhang, Xiao-Ming and Li, Tongyang and Yuan, Xiao},
  year={2022},
  DOI={10.1103/PhysRevLett.129.230504},
  url={https://doi.org/10.1103/PhysRevLett.129.230504}
}

@article{peruzzo2014variational,
  title={{\href{https://doi.org/10.1038/ncomms5213}{A variational eigenvalue solver on a photonic quantum processor}}},
  journal={Nature communications},
  author={Peruzzo, Alberto and McClean, Jarrod and Shadbolt, Peter and Yung, Man-Hong and Zhou, Xiao-Qi and Love, Peter J and Aspuru-Guzik, Al{\'a}n and O’brien, Jeremy L},
  year={2014},
  DOI={10.1038/ncomms5213},
  url={https://doi.org/10.1038/ncomms5213}
}

@article{mcclean2016theory,
  title={{\href{https://doi.org/10.1088/1367-2630/18/2/023023}{The theory of variational hybrid quantum-classical algorithms}}},
  journal={New Journal of Physics},
  author={McClean, Jarrod R and Romero, Jonathan and Babbush, Ryan and Aspuru-Guzik, Al{\'a}n},
  year={2016},
  DOI={10.1088/1367-2630/18/2/023023},
  url={https://doi.org/10.1088/1367-2630/18/2/023023}
}

@article{cerezo2021variational,
  title={{\href{https://doi.org/10.1038/s42254-021-00348-9}{Variational quantum algorithms}}},
  journal={Nature Reviews Physics},
  author={Cerezo, Marco and Arrasmith, Andrew and Babbush, Ryan and Benjamin, Simon C and Endo, Suguru and Fujii, Keisuke and McClean, Jarrod R and Mitarai, Kosuke and Yuan, Xiao and Cincio, Lukasz and others},
  year={2021},
  DOI={10.1038/s42254-021-00348-9},
  url={https://doi.org/10.1038/s42254-021-00348-9}
}

@article{schuld2019evaluating,
  title={{\href{https://doi.org/10.1103/PhysRevA.99.032331}{Evaluating analytic gradients on quantum hardware}}},
  journal={Physical Review A},
  author={Schuld, Maria and Bergholm, Ville and Gogolin, Christian and Izaac, Josh and Killoran, Nathan},
  year={2019},
  DOI={10.1103/PhysRevA.99.032331},
  url={https://doi.org/10.1103/PhysRevA.99.032331}
}

@article{banchi2021measuring,
  title={{\href{https://doi.org/10.22331/q-2021-01-25-386}{Measuring analytic gradients of general quantum evolution with the stochastic parameter shift rule}}},
  journal={Quantum},
  author={Banchi, Leonardo and Crooks, Gavin E},
  year={2021},
  DOI={10.22331/q-2021-01-25-386},
  url={https://doi.org/10.22331/q-2021-01-25-386}
}

@article{wierichs2022general,
  title={{\href{https://doi.org/10.22331/q-2022-03-30-677}{General parameter-shift rules for quantum gradients}}},
  journal={Quantum},
  author={Wierichs, David and Izaac, Josh and Wang, Cody and Lin, Cedric Yen-Yu},
  year={2022},
  DOI={10.22331/q-2022-03-30-677},
  url={https://doi.org/10.22331/q-2022-03-30-677}
}

@article{bowles2025backprop,
  title={{\href{https://doi.org/10.22331/q-2025-10-02-1873}{Backpropagation scaling in parameterised quantum circuits}}},
  journal={Quantum},
  author={Bowles, Joseph and Wierichs, David and Park, Chae-Yeun},
  year={2025},
  DOI={10.22331/q-2025-10-02-1873},
  url={https://doi.org/10.22331/q-2025-10-02-1873}
}

@article{chinzei2025tradeoff,
  title={{\href{https://doi.org/10.1038/s41534-025-01036-7}{Trade-off between gradient measurement efficiency and expressivity in deep quantum neural networks}}},
  journal={npj Quantum Information},
  author={Chinzei, Koki and Yamano, Shinichiro and Tran, Quoc Hoan and Endo, Yasuhiro and Oshima, Hirotaka},
  year={2025},
  DOI={10.1038/s41534-025-01036-7},
  url={https://doi.org/10.1038/s41534-025-01036-7}
}

@article{coyle2025density,
  title={{\href{https://doi.org/10.1038/s41534-025-01099-6}{Training-efficient density quantum machine learning}}},
  journal={npj Quantum Information},
  author={Coyle, Brian and Raj, Snehal and Mathur, Natansh and Cherrat, El Amine and Jain, Nishant and Kazdaghli, Skander and Kerenidis, Iordanis},
  year={2025},
  DOI={10.1038/s41534-025-01099-6},
  url={https://doi.org/10.1038/s41534-025-01099-6}
}

@article{meyer2021variational,
  title={{\href{https://doi.org/10.1038/s41534-021-00425-y}{A variational toolbox for quantum multi-parameter estimation}}},
  journal={npj Quantum Information},
  author={Meyer, Johannes Jakob and Borregaard, Johannes and Eisert, Jens},
  year={2021},
  DOI={10.1038/s41534-021-00425-y},
  url={https://doi.org/10.1038/s41534-021-00425-y}
}

@article{kaubruegger2023optimal,
  title={{\href{https://doi.org/10.1103/PRXQuantum.4.020333}{Optimal and variational multiparameter quantum metrology and vector-field sensing}}},
  journal={PRX Quantum},
  author={Kaubruegger, Raphael and Shankar, Athreya and Vasilyev, Denis V and Zoller, Peter},
  year={2023},
  DOI={10.1103/PRXQuantum.4.020333},
  url={https://doi.org/10.1103/PRXQuantum.4.020333}
}

@misc{vasilyev2024optimal,
  author={Vasilyev, Denis V. and Shankar, Athreya and Kaubruegger, Raphael and Zoller, Peter},
  archivePrefix={arXiv},
  eprint={2404.14194},
  primaryClass={quant-ph},
  year={2024}
}

@article{le2023variational,
   title={{\href{https://doi.org/10.1038/s41598-023-44786-0}{Variational quantum metrology for multiparameter estimation under dephasing noise}}},
   journal={Scientific Reports},
   author={Le, Trung Kien and Nguyen, Hung Q and Ho, Le Bin},
   year={2023},
   DOI={10.1038/s41598-023-44786-0},
   url={https://doi.org/10.1038/s41598-023-44786-0}
}

@article{maclellan2024end,
   title={{\href{https://doi.org/10.1038/s41534-024-00914-w}{End-to-end variational quantum sensing}}},
   journal={npj Quantum Information},
   author={MacLellan, Benjamin and Roztocki, Piotr and Czischek, Stefanie and Melko, Roger G},
   year={2024},
   DOI={10.1038/s41534-024-00914-w},
   url={https://doi.org/10.1038/s41534-024-00914-w}
}

@article{subramanian2024shallow,
   title={{\href{https://doi.org/10.1103/PhysRevA.110.032424}{Shallow-depth variational quantum hypothesis testing}}},
   journal={Physical Review A},
   author={Subramanian, Mahadevan and Vinjanampathy, Sai},
   year={2024},
   DOI={10.1103/PhysRevA.110.032424},
   url={https://doi.org/10.1103/PhysRevA.110.032424}
}

@article{kaubruegger2021quantum,
   title={{\href{https://doi.org/10.1103/PhysRevX.11.041045}{Quantum variational optimization of Ramsey interferometry and atomic clocks}}},
   journal={Physical review X},
   author={Kaubruegger, Raphael and Vasilyev, Denis V and Schulte, Marius and Hammerer, Klemens and Zoller, Peter},
   year={2021},
   DOI={10.1103/PhysRevX.11.041045},
   url={https://doi.org/10.1103/PhysRevX.11.041045}
}

@article{araujo2023configurable,
   title={{\href{https://doi.org/10.1007/s11128-023-03869-7}{Configurable sublinear circuits for quantum state preparation}}},
   journal={Quantum Information Processing},
   author={Araujo, Israel F and Park, Daniel K and Ludermir, Teresa B and Oliveira, Wilson R and Petruccione, Francesco and Da Silva, Adenilton J},
   year={2023},
   DOI={10.1007/s11128-023-03869-7},
   url={https://doi.org/10.1007/s11128-023-03869-7}
}

@article{braunstein1994statistical,
  title={{\href{https://doi.org/10.1103/PhysRevLett.72.3439}{Statistical distance and the geometry of quantum states}}},
  author={Braunstein, Samuel L. and Caves, Carlton M.},
  journal={Physical Review Letters},
  year={1994},
  doi={10.1103/PhysRevLett.72.3439},
  url={https://doi.org/10.1103/PhysRevLett.72.3439}
}

@article{paris2009quantum,
  title={{\href{https://doi.org/10.1142/S0219749909004839}{Quantum estimation for quantum technology}}},
  author={Paris, Matteo G. A.},
  journal={International Journal of Quantum Information},
  year={2009},
  doi={10.1142/S0219749909004839},
  url={https://doi.org/10.1142/S0219749909004839}
}

@article{giovannetti2011advances,
  title={{\href{https://doi.org/10.1038/nphoton.2011.35}{Advances in quantum metrology}}},
  author={Giovannetti, Vittorio and Lloyd, Seth and Maccone, Lorenzo},
  journal={Nature Photonics},
  year={2011},
  doi={10.1038/nphoton.2011.35},
  url={https://doi.org/10.1038/nphoton.2011.35}
}

@article{degen2017quantum,
  title={{\href{https://doi.org/10.1103/RevModPhys.89.035002}{Quantum sensing}}},
  author={Degen, Christian L. and Reinhard, Friedemann and Cappellaro, Paola},
  journal={Reviews of Modern Physics},
  year={2017},
  doi={10.1103/RevModPhys.89.035002},
  url={https://doi.org/10.1103/RevModPhys.89.035002}
}

@article{pezze2018quantum,
  title={{\href{https://doi.org/10.1103/RevModPhys.90.035005}{Quantum metrology with nonclassical states of atomic ensembles}}},
  author={Pezz{\`e}, Luca and Smerzi, Augusto and Oberthaler, Markus K. and Schmied, Roman and Treutlein, Philipp},
  journal={Reviews of Modern Physics},
  year={2018},
  doi={10.1103/RevModPhys.90.035005},
  url={https://doi.org/10.1103/RevModPhys.90.035005}
}

@article{ramsey1950molecular,
  title={{\href{https://doi.org/10.1103/PhysRev.78.695}{A molecular beam resonance method with separated oscillating fields}}},
  author={Ramsey, Norman F.},
  journal={Physical Review},
  year={1950},
  doi={10.1103/PhysRev.78.695},
  url={https://doi.org/10.1103/PhysRev.78.695}
}

@article{koczor2020variational,
  title={{\href{https://doi.org/10.1088/1367-2630/ab965e}{Variational-state quantum metrology}}},
  author={Koczor, B{\'a}lint and Endo, Suguru and Jones, Tyson and Matsuzaki, Yuichiro and Benjamin, Simon C.},
  journal={New Journal of Physics},
  year={2020},
  doi={10.1088/1367-2630/ab965e},
  url={https://doi.org/10.1088/1367-2630/ab965e}
}

@article{yang2021hybrid,
  title={{\href{https://doi.org/10.1038/s41598-020-80070-1}{Hybrid quantum-classical approach to enhanced quantum metrology}}},
  author={Yang, Xiaodong and Chen, Xi and Li, Jun and Peng, Xinhua and Laflamme, Raymond},
  journal={Scientific Reports},
  year={2021},
  doi={10.1038/s41598-020-80070-1},
  url={https://doi.org/10.1038/s41598-020-80070-1}
}

@article{kaubruegger2019variational,
  title={{\href{https://doi.org/10.1103/PhysRevLett.123.260505}{Variational spin-squeezing algorithms on programmable quantum sensors}}},
  author={Kaubruegger, Raphael and Silvi, Pietro and Kokail, Christian and van Bijnen, Rick and Rey, Ana Maria and Ye, Jun and Kaufman, Adam M. and Zoller, Peter},
  journal={Physical Review Letters},
  year={2019},
  doi={10.1103/PhysRevLett.123.260505},
  url={https://doi.org/10.1103/PhysRevLett.123.260505}
}

@article{bae2015quantum,
  title={{\href{https://doi.org/10.1088/1751-8113/48/8/083001}{Quantum state discrimination and its applications}}},
  author={Bae, Joonwoo and Kwek, Leong-Chuan},
  journal={Journal of Physics A: Mathematical and Theoretical},
  year={2015},
  doi={10.1088/1751-8113/48/8/083001},
  url={https://doi.org/10.1088/1751-8113/48/8/083001}
}

@article{khatri2019quantumassisted,
  title={{\href{https://doi.org/10.22331/q-2019-05-13-140}{Quantum-assisted quantum compiling}}},
  author={Khatri, Sumeet and LaRose, Ryan and Poremba, Alexander and Cincio, Lukasz and Sornborger, Andrew T. and Coles, Patrick J.},
  journal={Quantum},
  year={2019},
  doi={10.22331/q-2019-05-13-140},
  url={https://doi.org/10.22331/q-2019-05-13-140}
}

@article{rabitz2000whither,
  title={{\href{https://doi.org/10.1126/science.288.5467.824}{Whither the future of controlling quantum phenomena?}}},
  author={Rabitz, Herschel and de Vivie-Riedle, Regina and Motzkus, Marcus and Kompa, Karl},
  journal={Science},
  year={2000},
  doi={10.1126/science.288.5467.824},
  url={https://doi.org/10.1126/science.288.5467.824}
}

@article{glaser2015training,
  title={{\href{https://doi.org/10.1140/epjd/e2015-60464-1}{Training Schr{\"o}dinger's cat: quantum optimal control}}},
  author={Glaser, Steffen J. and Boscain, Ugo and Calarco, Tommaso and Koch, Christiane P. and K{\"o}ckenberger, Walter and Kosloff, Ronnie and Kuprov, Ilya and Luy, Burkhard and Schirmer, Sophie and Schulte-Herbr{\"u}ggen, Thomas and others},
  journal={The European Physical Journal D},
  year={2015},
  doi={10.1140/epjd/e2015-60464-1},
  url={https://doi.org/10.1140/epjd/e2015-60464-1}
}

@misc{lin2026hopfansatzcode,
  author = {Lin, Ruge and Li, Guangxi},
  title = {{\href{https://github.com/GoGoKo699/Hopf-ansatz}{Hopf ansatz GitHub repository}}},
  year = {2026}
}

@misc{wilkens2023quantumcircuitlibrary,
  author = {Wilkens, Jadwiga},
  title = {{\href{https://github.com/wilkensJ/quantum-circuit-drawio-library}{Quantum Circuit Library}}},
  year = {2023}
}

@article{dai1999nonlinear,
  title={{\href{https://doi.org/10.1137/S1052623497318992}{A nonlinear conjugate gradient method with a strong global convergence property}}},
  author={Dai, Yu-Hong and Yuan, Yaxiang},
  journal={SIAM Journal on optimization},
  year={1999},
  doi={10.1137/S1052623497318992},
  url={https://doi.org/10.1137/S1052623497318992}
}

@article{liu1989limited,
  title={{\href{https://doi.org/10.1007/BF01589116}{On the limited memory BFGS method for large scale optimization}}},
  author={Liu, Dong C and Nocedal, Jorge},
  journal={Mathematical programming},
  year={1989},
  doi={10.1007/BF01589116},
  url={https://doi.org/10.1007/BF01589116}
}

@book{absil2008optimization,
  title={{\href{https://dl.acm.org/doi/10.5555/1557548}{Optimization algorithms on matrix manifolds}}},
  author={Absil, P-A and Mahony, Robert and Sepulchre, Rodolphe},
  publisher={Princeton University Press},
  year={2008},
  doi={10.5555/1557548},
  url={https://dl.acm.org/doi/10.5555/1557548}
}

@article{barzilai1988two,
  title={{\href{https://doi.org/10.1093/imanum/8.1.141}{Two-point step size gradient methods}}},
  author={Barzilai, Jonathan and Borwein, Jonathan M},
  journal={IMA journal of numerical analysis},
  year={1988},
  doi={10.1093/imanum/8.1.141},
  url={https://doi.org/10.1093/imanum/8.1.141}
}

@article{hestenes1952methods,
  title={{\href{https://doi.org/10.6028/JRES.049.044}{Methods of conjugate gradients for solving linear systems}}},
  author={Hestenes, Magnus R and Stiefel, Eduard and others},
  journal={Journal of research of the National Bureau of Standards},
  year={1952},
  doi={10.6028/JRES.049.044},
  url={https://doi.org/10.6028/JRES.049.044}
}

@article{efthymiou2022qibo,
  title={{\href{https://doi.org/10.1088/2058-9565/ac39f5}{Qibo: a framework for quantum simulation with hardware acceleration}}},
  author={Efthymiou, Stavros and Ramos-Calderer, Sergi and Bravo-Prieto, Carlos and P{\'e}rez-Salinas, Adri{\'a}n and Garc{\'\i}a-Mart{\'\i}n, Diego and Garcia-Saez, Artur and Latorre, Jos{\'e} Ignacio and Carrazza, Stefano},
  journal={Quantum Science \& Technology},
  year={2022},
  doi={10.1088/2058-9565/ac39f5},
  url={https://doi.org/10.1088/2058-9565/ac39f5}
}

@misc{kingma2014adam,
  author={Kingma, Diederik P. and Ba, Jimmy},
  archivePrefix={arXiv},
  eprint={1412.6980},
  primaryClass={cs.LG},
  year={2014}
}

@article{schuld2021effect,
  title={{\href{https://doi.org/10.1103/PhysRevX.11.041011}{Effect of data encoding on the expressive power of variational quantum-machine-learning models}}},
  author={Schuld, Maria and Sweke, Ryan and Meyer, Johannes Jakob},
  journal={Physical Review X},
  volume={11},
  pages={041011},
  year={2021},
  url={https://doi.org/10.1103/PhysRevX.11.041011}
}

@article{li2022concentration,
  title={{\href{https://dl.acm.org/doi/10.5555/3600270.3601684}{Concentration of data encoding in parameterized quantum circuits}}},
  author={Li, Guangxi and Ye, Ruilin and Zhao, Xuanqiang and Wang, Xin},
  journal={Advances in Neural Information Processing Systems},
  volume={35},
  pages={19456--19469},
  year={2022},
  url={https://dl.acm.org/doi/10.5555/3600270.3601684}
}

@misc{lin2026hopfqbp,
  author = {Lin, Ruge and Chen, Ranyiliu},
  title = {{\href{https://arxiv.org/abs/2610.14177}{Compass in the Mirror:
            Quantum Backpropagation with the Hopf Ansatz}}},
  year = {2026},
  eprint = {2610.14177},
  archivePrefix = {arXiv},
  primaryClass = {quant-ph}
}

\clearpage

\onecolumngrid
\appendix

\FloatBarrier

\section{Algorithms and subroutines}
\label{app:algorithms}

\subsection*{Inverse Hopf map}
\label{app:inverse_hopf_map}

\begin{lemma}[Inverse Hopf map and classical cost]
\label{lem:inverse_hopf_map}
Let $n\ge1$ and $N=2^n$. We use the quadrant-aware angle
$\operatorname{atan2}(y,x)$, whose sine and cosine have the signs of $y$
and $x$, respectively; reduction modulo $2\pi$ selects $[0,2\pi)$.

(A) Real case. Let $\mathbf{x}=(x_0,\ldots,x_{N-1})\in\mathbb{R}^{N}$ satisfy
$\sum_{\ell=0}^{N-1}x_\ell^2=1$. For each internal node
$j\in\{1,\ldots,N-1\}$, define its left- and right-subtree leaf sets by
\[
\mathcal{L}_L(j):=\{\ell: \ket{b_\ell}\text{ is a leaf in the left subtree of }j\},
\qquad
\mathcal{L}_R(j):=\{\ell: \ket{b_\ell}\text{ is a leaf in the right subtree of }j\},
\]
and the associated subtree norms
\begin{equation}
\label{eq:inverse_hopf_subtree_norms_real}
S_L(j):=\sqrt{\sum_{\ell\in\mathcal{L}_L(j)}x_\ell^2},
\qquad
S_R(j):=\sqrt{\sum_{\ell\in\mathcal{L}_R(j)}x_\ell^2},
\qquad
S(j):=\sqrt{S_L(j)^2+S_R(j)^2}.
\end{equation}
For the non-final internal nodes, namely $1\le j\le N/2-1$, set
\begin{equation}
\label{eq:inverse_hopf_theta_internal_real}
\theta_j:=
\begin{cases}
\operatorname{atan2}\!\big(S_R(j),S_L(j)\big), & S(j)>0,\\[2pt]
0, & S(j)=0,
\end{cases}
\qquad
\theta_j\in[0,\pi/2].
\end{equation}
For the final-layer internal nodes, write $j=N/2+k$ with
$k=0,\ldots,N/2-1$. These nodes split the sibling leaf pair
$(2k,2k+1)$, and their angles are chosen with the signed quadrant convention
\begin{equation}
\label{eq:inverse_hopf_theta_final_real}
\theta_{N/2+k}:=
\begin{cases}
\operatorname{atan2}\!\big(x_{2k+1},x_{2k}\big) \pmod{2\pi},
& x_{2k}^2+x_{2k+1}^2>0,\\[2pt]
0, & x_{2k}=x_{2k+1}=0.
\end{cases}
\end{equation}
With the ranges of Definition~\ref{def:hopf_coordinates}(A), the forward Hopf
recursion then reproduces the signed vector $\mathbf{x}$ exactly on the regular
set. At zero-subtree and zero-leaf-pair points, the inverse is not unique; the
choices above fix one convention.

(B) Complex case. Let $\mathbf{x}=(x_0,\ldots,x_{N-1})\in\mathbb{C}^{N}$ satisfy
$\sum_\ell |x_\ell|^2=1$. Write
\[
x_\ell=r_\ell e^{\iu\phi_\ell},
\qquad
r_\ell=|x_\ell|\ge0,
\]
where
\[
\phi_\ell := \arg(x_\ell)\pmod{2\pi}\in[0,2\pi),
\]
with the convention \(\phi_\ell:=0\) if \(x_\ell=0\).
For each internal node $j\in\{1,\ldots,N-1\}$, define
\[
S_{L,r}(j):=\sqrt{\sum_{\ell\in\mathcal{L}_L(j)}r_\ell^2},
\qquad
S_{R,r}(j):=\sqrt{\sum_{\ell\in\mathcal{L}_R(j)}r_\ell^2},
\qquad
S_r(j):=\sqrt{S_{L,r}(j)^2+S_{R,r}(j)^2}.
\]
Set
\begin{equation}
\label{eq:inverse_hopf_theta_complex_magnitudes}
\theta_j:=
\begin{cases}
\operatorname{atan2}\!\big(S_{R,r}(j),S_{L,r}(j)\big), & S_r(j)>0,\\[2pt]
0, & S_r(j)=0,
\end{cases}
\qquad
j=1,\ldots,N-1.
\end{equation}
The leaf phases are
\begin{equation}
\label{eq:inverse_hopf_leaf_phases_complex}
\theta_{N+\ell}:=\phi_\ell\in[0,2\pi),
\qquad
\ell=0,\ldots,N-1.
\end{equation}
Then the complex Hopf recursion of Definition~\ref{def:hopf_coordinates}(B)
reproduces $\mathbf{x}$ up to the chosen phase convention at zero entries. The
inverse is not unique whenever some magnitude subtree or leaf amplitude
vanishes.

We denote by $\mathrm{Hopf}^{-1}$ any choice of the inverse Hopf map constructed
above on the regular set. There is an $\bigO(N)$-time, $\bigO(N)$-memory procedure to
compute it. Use a tree-indexed array of weights with leaves
$w_{N+\ell}=x_\ell^2$ in the real case and $w_{N+\ell}=|x_\ell|^2$ in the
complex case. A single bottom-up pass gives
\[
w_j=w_{2j}+w_{2j+1},
\qquad
j=N-1,N-2,\ldots,1.
\]
The subtree norms are $S_L(j)=\sqrt{w_{2j}}$ and
$S_R(j)=\sqrt{w_{2j+1}}$. The real final-layer angles are then obtained directly
from the signed pairs by Eq.~\eqref{eq:inverse_hopf_theta_final_real}; the
complex leaf phases require one additional linear pass. Thus,
\[
\boxed{
\text{Time}_{\mathrm{Hopf}^{-1}}(N)=\bigO(N),
\qquad
\text{Memory}_{\mathrm{Hopf}^{-1}}(N)=\bigO(N).
}
\]
\end{lemma}

\subsection*{Gate-schedule routines}
\label{app:gate_schedule_routines}

The gate-ordering routines output four parallel lists
\[
\mathtt{Ctrl},\qquad
\mathtt{Anti},\qquad
\mathtt{Targ},\qquad
\mathtt{Index}.
\]
All entries are read position-wise.  The masks \(\mathtt{Ctrl}\) and
\(\mathtt{Anti}\) specify, respectively, controls on \(\ket{1}\) and controls
on \(\ket{0}\).  The mask \(\mathtt{Targ}\) is a power of two specifying the
single target qubit.  We use the bitmask convention
\[
(q_n,\ldots,q_1)
\]
from left to right, so the integer bit \(2^{t-1}\) corresponds to qubit \(q_t\).
Thus, for example, for \(n=4\), the mask \(1000\) targets \(q_4\), while
\(0001\) targets \(q_1\).

In the real ansatz, every \(\mathtt{Index}\) entry is a single integer \(j\)
and the corresponding gate is \(R_y(\theta_j)\). In the complex ansatz,
\(\mathtt{Index}\) entries have the following convention:
\[
\mathtt{Index}_r
\in
\{1,\ldots,N-1\}
\quad\text{or}\quad
\mathtt{Index}_r=[j,N+\ell_L,N+\ell_R].
\]
Here \(r\) is the position in the gate list, and \(\ell_L,\ell_R\) are
the computational-basis labels of the left and right sibling leaves of
node \(j\). An integer entry denotes a magnitude-only \(R_y(\theta_j)\)
gate. A length-three entry denotes a promoted final-layer gate
\[
R_{\mathbb C}(\theta_j,\theta_{N+\ell_L},\theta_{N+\ell_R}).
\]
We do not use length-one lists for magnitude-only complex gates.
The list routines use zero-based indexing, \(|A|\) denotes the length of a
list \(A\), and \(HW[k]\) collects basis labels of Hamming
weight \(k\), the number of one bits. Their stated complexities use a
word model in which an \(n\)-bit mask fits in one word; bit-counting and
sorting costs are included as shown.

\providecommand{\theHalgorithm}{\arabic{algorithm}}
\renewcommand{\theHalgorithm}{hopfmain.\arabic{algorithm}}
{
\floatname{algorithm}{Algorithm }
\begin{algorithm}[H]
    \algrenewcommand{\algorithmicrequire}{\textbf{Input}}
    \algrenewcommand{\algorithmicensure}{\textbf{Output}}
    \algnewcommand{\algorithmiccomplexity}{\textbf{Complexity}}
    \newcommand{\Complexity}{\item[\algorithmiccomplexity]}
    \caption{Generate the order of quantum gates (real)}
    \label{alg:order_real}
    \begin{algorithmic}[1]
    \Require Integer $n$ (number of qubits).
    \Ensure Lists $\mathtt{Ctrl}$, $\mathtt{Anti}$, $\mathtt{Targ}$, $\mathtt{Index}$.
    \Complexity $\bigO\!\left(nN\right)$.    
    \Procedure{\color{OliveGreen}{HopfReal}}{$n$}
    \State $N \gets 2^n$
    \State $HW \gets \textsc{\color{OliveGreen}{HammingWeight}}(n)$
    \State $\mathtt{Ctrl} \gets [0,\dots,0]$ \Comment{$n$ zeros}
    \State $\mathtt{Anti} \gets [0, N-2^{n-1},\dots, N-2]$
    \Comment{$0$; then $N-2^{n-i}$ for $i=1,\ldots,n-1$}
    \State $\mathtt{Targ} \gets [2^{n-1},2^{n-2},\dots,2^0]$ \Comment{$2^{n-i}$ for $i=1,\ldots,n$}
    \State $\mathtt{Index} \gets [1,2,4,\dots,2^{n-1}]$ \Comment{$n$ entries: $\mathtt{Index}[i]=2^{i}$ for $i=0,\dots,n-1$ (one per level on the initial root-to-leaf chain)}
    \For{$k=1$ \textbf{to} $n-1$}
        \State $A,B \gets \textsc{\color{OliveGreen}{FindPairs}}(HW[k], HW[k+1])$
        \State $\mathtt{Ctrl}.\text{extend}(A)$
        \State $\mathtt{Anti}.\text{extend}(\textsc{\color{OliveGreen}{Anti}}(A))$
        \State $\mathtt{Targ}.\text{extend}(B-A)$ \Comment{elementwise $B_i-A_i$}
        \State $L \gets |A|$
        \For{$i=0$ \textbf{to} $L-1$}
            \State $\mathtt{Index}.\text{append}(\textsc{\color{OliveGreen}{ThetaReal}}(n,A[i],B[i]))$
        \EndFor
    \EndFor
    \State \Return $\mathtt{Ctrl}$, $\mathtt{Anti}$, $\mathtt{Targ}$, $\mathtt{Index}$
    \EndProcedure
    \end{algorithmic}
\end{algorithm}
}

\setcounter{algorithm}{0}
\renewcommand{\theHalgorithm}{hopfsub.\arabic{algorithm}}
{
\floatname{algorithm}{Subroutine }
\begin{algorithm}[H]
    \caption{Group numbers by Hamming weight}
    \algrenewcommand{\algorithmicrequire}{\textbf{Input}}
    \algrenewcommand{\algorithmicensure}{\textbf{Output}}
    \algnewcommand{\algorithmiccomplexity}{\textbf{Complexity}}
    \newcommand{\Complexity}{\item[\algorithmiccomplexity]}
    
    \label{alg:hamming_weight}
    \begin{algorithmic}[1]
    
    \Require Integer $n$ for the bit length.
    \Ensure A list $HW$ of length $n+1$, where $HW[k]$ contains all integers with Hamming weight $k$.
    \Complexity $\bigO\!\left(nN\right)$.
    
    \Procedure{\color{OliveGreen}{HammingWeight}}{$n$}
        \State $N \gets 2^n$
        \State Initialize $HW$ as $[\,[\,],\,\hdots,\,[\,]\,]$ \Comment{$n+1$ empty lists}
        \For{$i = 0$ \textbf{to} $N-1$}
            \State $k \gets \text{count\_ones}(i)$
            \State $HW[k].\text{append}(i)$
        \EndFor
        \State \Return $HW$
    \EndProcedure  
\end{algorithmic}
\end{algorithm}
}

{
\floatname{algorithm}{Subroutine }
\begin{algorithm}[H]
    \algrenewcommand{\algorithmicrequire}{\textbf{Input}}
    \algrenewcommand{\algorithmicensure}{\textbf{Output}}
    \algnewcommand{\algorithmiccomplexity}{\textbf{Complexity}}
    \newcommand{\Complexity}{\item[\algorithmiccomplexity]}
    \caption{Find pairs between elements of two lists}
    \label{alg:find_pairs}
    \begin{algorithmic}[1]
    
    \Require Sorted lists $A$ and $B$ (ascending).
    \Ensure Two lists $A'$ and $B'$ (of equal length) obtained by scanning $B$ from largest to smallest and pairing each $b$ with the current largest $a\in A$ such that $a<b$ (with $a$ reusable), and emitting the resulting pairs ordered by increasing $a$, with $b$ decreasing within each group.
    \Complexity $\bigO(|A| + |B|)$.
    
    \Procedure{\color{OliveGreen}{FindPairs}}{$A$, $B$}
    \State $\widetilde A \gets [\,]$, \ $\widetilde B \gets [\,]$
    \State $i \gets |A|-1$
    \For{$b$ \textbf{in} $B$ \textbf{traversed in reverse order}}
        \While{$i \ge 0$ \textbf{and} $A[i] \ge b$}
           \State $i \gets i-1$
        \EndWhile
        \If{$i < 0$}
           \State \textbf{break}
        \EndIf
        \State append $A[i]$ to $\widetilde A$
        \State append $b$ to $\widetilde B$
    \EndFor

    \State $A' \gets [\,]$, \ $B' \gets [\,]$
    \State $r \gets |\widetilde A|-1$
    \While{$r \ge 0$}
        \State $a \gets \widetilde A[r]$
        \State $\ell \gets r$
        \While{$\ell>0$ \textbf{and} $\widetilde A[\ell-1]=a$}
            \State $\ell \gets \ell-1$
        \EndWhile
        \For{$t=\ell$ \textbf{to} $r$}
            \State append $\widetilde A[t]$ to $A'$
            \State append $\widetilde B[t]$ to $B'$
        \EndFor
        \State $r \gets \ell-1$
    \EndWhile

    \State \Return $A', B'$
    \EndProcedure
    \end{algorithmic}
\end{algorithm}
}

{
\floatname{algorithm}{Subroutine }
\begin{algorithm}[H]
    \algrenewcommand{\algorithmicrequire}{\textbf{Input}}
    \algrenewcommand{\algorithmicensure}{\textbf{Output}}
    \algnewcommand{\algorithmiccomplexity}{\textbf{Complexity}}
    \newcommand{\Complexity}{\item[\algorithmiccomplexity]}
    \caption{Find anti-controls based on control multiplicities}
    \label{alg:anti_controls}
    \begin{algorithmic}[1]

    \Require A list $A$ of integers representing control bitmasks.
    \Ensure A list $\mathtt{Anti}$ of integers, constructed as follows:
    for each distinct value $a\in A$ with multiplicity $c$,
    the list contains the values $2^{c}-2^{i}$ for $i=c,c-1,\dots,1$.
    The output depends only on multiplicities in $A$ and is ordered
    deterministically by ascending distinct values of $a$.
    \Complexity $\bigO\!\left(|A|\log|A|\right)$.

    \Procedure{\color{OliveGreen}{Anti}}{$A$}
        \State $\mathtt{Anti} \gets [\,]$
        \State $\mathtt{Dict} \gets$ frequency map of elements in $A$
        \State $\mathtt{Keys} \gets$ sorted list of distinct elements of $A$ (ascending)
        \For{each $a$ \textbf{in} $\mathtt{Keys}$}
            \State $c \gets \mathtt{Dict}[a]$
            \For{$i = c$ \textbf{down to} $1$}
                \State append $(2^{c} - 2^{i})$ to $\mathtt{Anti}$
            \EndFor
        \EndFor
        \State \Return $\mathtt{Anti}$
    \EndProcedure
    \end{algorithmic}
\end{algorithm}
}

{
\floatname{algorithm}{Subroutine }
\begin{algorithm}[H]
    \algrenewcommand{\algorithmicrequire}{\textbf{Input}}
    \algrenewcommand{\algorithmicensure}{\textbf{Output}}
    \algnewcommand{\algorithmiccomplexity}{\textbf{Complexity}}
    \newcommand{\Complexity}{\item[\algorithmiccomplexity]}
    \caption{Find $\theta$ index for the gate (real)}
    \label{alg:theta_real}
    \begin{algorithmic}[1]
    
    \Require Integer $n$ for the bit length, integer $a$ encoding initial basis state, integer $b$ encoding target basis state.
    \Ensure An integer giving the $\theta$ index.
    \Complexity $\bigO(1)$.
    
    \Procedure{\color{OliveGreen}{ThetaReal}}{$n$, $a$, $b$}
    \State $N \gets 2^n$
    \State $d \gets b - a$
    \State \Return $\dfrac{N+a}{2d}$
    \EndProcedure
    \Statex
    \Statex \textbf{Remark (integrality).} For every $(a,b)$ produced by $\textsc{FindPairs}(HW[k],HW[k+1])$, $d=b-a$ is the least significant set bit of $b$, and $a=b-d$ is divisible by $2d$. Since $2d$ also divides $N$, $(N+a)/(2d)\in\mathbb{Z}$.                                     
    \end{algorithmic}
\end{algorithm}
}

{
\floatname{algorithm}{Subroutine }
\begin{algorithm}[H]
    \algrenewcommand{\algorithmicrequire}{\textbf{Input}}
    \algrenewcommand{\algorithmicensure}{\textbf{Output}}
    \algnewcommand{\algorithmiccomplexity}{\textbf{Complexity}}
    \newcommand{\Complexity}{\item[\algorithmiccomplexity]}
    \caption{Find $\theta$ index for the gate (complex)}
    \label{alg:theta_complex}
    \begin{algorithmic}[1]
    
    \Require Integer $n$, integer $a$ (initial state), integer $b$ (target state).
    \Ensure Either one integer index, or a length-three list of indices.
    \Complexity $\bigO(1)$.
    
    \Procedure{\color{OliveGreen}{ThetaComplex}}{$n$, $a$, $b$}
    \State $N \gets 2^n$
    \State $d \gets b - a$
    \State $j \gets \dfrac{N+a}{2d}$
    \If{$d=1$}
       \State \Return $[j,\,N+a,\,N+b]$ \Comment{leaf-level transition, requiring phase angles}
    \Else
       \State \Return $j$
    \EndIf
    \EndProcedure
    \Statex
    \Statex \textbf{Remark (integrality).} For every $(a,b)$ produced by $\textsc{FindPairs}(HW[k],HW[k+1])$, $d=b-a$ is the least significant set bit of $b$, and $a=b-d$ is divisible by $2d$. Since $2d$ also divides $N$, $(N+a)/(2d)\in\mathbb{Z}$.   \end{algorithmic}
\end{algorithm}
}

\setcounter{algorithm}{1}
\renewcommand{\theHalgorithm}{hopfmain.\arabic{algorithm}}
{
\floatname{algorithm}{Algorithm }
\begin{algorithm}[H]
    \algrenewcommand{\algorithmicrequire}{\textbf{Input}}
    \algrenewcommand{\algorithmicensure}{\textbf{Output}}
    \algnewcommand{\algorithmiccomplexity}{\textbf{Complexity}}
    \newcommand{\Complexity}{\item[\algorithmiccomplexity]}
    \caption{Generate the order of quantum gates (complex)}
    \label{alg:order_complex}
    \begin{algorithmic}[1]
    
    \Require Integer $n$.
    \Ensure Lists $\mathtt{Ctrl}$, $\mathtt{Anti}$, $\mathtt{Targ}$, $\mathtt{Index}$, where each $\mathtt{Index}$ entry is either an integer or a length-$3$ list.
    \Complexity $\bigO\!\left(nN\right)$.
    
    \Procedure{\color{OliveGreen}{HopfComplex}}{$n$}
    \If{$n=1$}
        \State \Return $[0]$, $[0]$, $[1]$, $[[1,2,3]]$
    \EndIf
    \State $N \gets 2^n$
    \State $HW \gets \textsc{\color{OliveGreen}{HammingWeight}}(n)$
    \State $\mathtt{Ctrl} \gets [0,\dots,0]$ 
    \State $\mathtt{Anti} \gets [0, N-2^{n-1},\dots, N-2]$
    \State $\mathtt{Targ} \gets [2^{n-1},2^{n-2},\dots,2^0]$
    \State $\mathtt{Index} \gets [1,2,4,\dots,2^{n-2},\,[2^{n-1},\,N,\,N+1]]$ \Comment{$n$ initial gates: the first $n-1$ are magnitude-only; the last is promoted to $R_{\mathbb{C}}$ to attach the first sibling-pair phases}

    \For{$k=1$ \textbf{to} $n-1$}
        \State $A,B \gets \textsc{\color{OliveGreen}{FindPairs}}(HW[k], HW[k+1])$
        \State $\mathtt{Ctrl}.\text{extend}(A)$
        \State $\mathtt{Anti}.\text{extend}(\textsc{\color{OliveGreen}{Anti}}(A))$
        \State $\mathtt{Targ}.\text{extend}(B-A)$
        \State $L \gets |A|$
        \For{$i=0$ \textbf{to} $L-1$}
            \State $\mathtt{Index}.\text{append}(\textsc{\color{OliveGreen}{ThetaComplex}}(n,A[i], B[i]))$
        \EndFor
    \EndFor
    \State \Return $\mathtt{Ctrl}$, $\mathtt{Anti}$, $\mathtt{Targ}$, $\mathtt{Index}$
    \EndProcedure
                                         
    \end{algorithmic}
\end{algorithm}
}

\FloatBarrier
\subsection*{Why the native ordering prepares the Hopf state}
\label{app:native_schedule_invariant}

The ordering can be described directly in terms of basis labels. For
\(a>0\), let \(c\) be the number of trailing zero bits of \(a\) and
\(k\) its Hamming weight; set \(c=n\) and \(k=0\) for \(a=0\).
Only groups with \(c\ge1\) contain gates. The group for \(a\) runs through
\(t=c-1,c-2,\ldots,0\), with
\[
b=a+2^t,\qquad
\mathtt{Ctrl}=a,\qquad
\mathtt{Anti}=2^c-2^{t+1},\qquad
\mathtt{Targ}=2^t,\qquad
j=\frac{N+a}{2^{t+1}}.
\]
The initial \(a=0\) group is the leftmost root-to-leaf chain. All other
groups are ordered first by \(k\), then by increasing \(a\).
Indeed, removing the least significant one bit from any weight-\((k+1)\)
label \(b\) gives the largest weight-\(k\) label smaller than \(b\).
This is precisely the pairing made by \textsc{FindPairs}; the group has
\(c\) entries, so \textsc{Anti} supplies the masks above.

Before the gate for \((a,t)\), the only potentially populated basis state
satisfying its controls is \(\ket{b_a}\); the partner \(\ket{b_b}\) has
not yet been populated. To verify this invariant, note that the current
pass contains states of weight at most \(k+1\). States of weight below
\(k\) cannot satisfy the positive mask. At weight \(k\), only \(a\)
can do so. A weight-\((k+1)\) state containing all the one bits of \(a\)
has the form \(a+2^u\). If \(u<c\), it belongs to the same group:
earlier choices \(u>t\) are excluded by the anti-controls, while
\(u\le t\) have not yet been generated. If \(u>c\), its parent obtained
by removing its least significant one bit is larger than \(a\), so that
group has not been processed. This also verifies the initial group, where
the anti-controls exclude its previously generated one-bit labels.

Each gate therefore splits one occupied amplitude into its two tree
branches using only the first column of \(R_y\) or \(R_{\mathbb C}\).
Every internal node occurs once: \(a\) is its prefix padded by zeros,
and \(2^t\) identifies the next bit to split. Its ancestors occur either
in an earlier Hamming-weight group or earlier in the same group.
Induction over this order reproduces the sine--cosine recursion of
Definition~\ref{def:hopf_coordinates}; a final-layer split also attaches
its two leaf phases. The argument allows zero amplitudes and arbitrary
real gate angles, including the off-chart tangent assignments. It proves
state preparation from \(\ket{0}^{\otimes n}\), not equality of the
native and fully prefix-controlled circuits on arbitrary inputs.

\subsection*{Assigned CNOT ledger}
\label{app:cnot_model}

\begin{lemma}[Assigned CNOT ledger for mixed-sign multi-controls]
\label{lem:cnot_model}
We assign logical CNOT charges to controlled gates, with single-qubit gates
free in this ledger and no additional clean ancilla for their decomposition.
The counts are specific to this logical-gate ledger and are not optimized
over alternative compilations. Connectivity routing, observable readout,
and repeated executions are accounted for separately.

Consider a single-qubit gate acting on one target and conditioned on a set of control qubits,
where each control is specified either as a positive control (trigger on $\ket{1}$) or a negative control (trigger on $\ket{0}$).
Let $m$ denote the total number of controls (positive plus negative).
In the implementation, controls are specified by two disjoint bitmasks $\mathtt{Ctrl}$ and $\mathtt{Anti}$,
and
\[
m=\operatorname{popcount}(\mathtt{Ctrl}\lor \mathtt{Anti}).
\]
Here \(\lor\) is bitwise OR and \(\operatorname{popcount}\) counts set bits.

\begin{itemize}
\item \textbf{Sign-independence.}
Negative controls are implemented by conjugating the corresponding control qubits by Pauli-$X$
(before and after the controlled operation). Since $X$ can be taken as a CNOT-free primitive
in our accounting, the CNOT cost depends only on $m$ and not on how many controls are negative.

\item \textbf{Controlled-$R_y$.}
Let $C_R(m)$ be the assigned charge for an $m$-controlled $R_y$ gate.
We use
\[
C_R(m)=
\begin{cases}
0, & m=0,\\
2^{m+1}-2, & 1\le m\le 4,\\
16(m+1)-40, & m\ge 5,
\end{cases}
\]
The small-control values are fixed charges based on standard ancilla-free
controlled-gate constructions \cite{barenco1995elementary,shende2008cnot}.
The linear regime uses the constructions of
Ref.~\cite{vale2023circuit}, with \(n_g=m+1\) the total number of qubits
in the controlled gate.

\item \textbf{Controlled-$R_{\mathbb C}$ (special unitary single-qubit gate).}
Let $C_U(m)$ be the assigned charge for an $m$-controlled special unitary
single-qubit gate, here the promoted $R_{\mathbb C}$ gate.
We use
\[
C_U(m)=
\begin{cases}
0, & m=0,\\
2^{m+1}-2, & 1\le m\le 4,\\
\begin{cases}
20(m+1)-38, & m+1\ \mathrm{odd},\\
20(m+1)-42, & m+1\ \mathrm{even},
\end{cases}
& m\ge 5,
\end{cases}
\]
with the same small-control convention and the linear ancilla-free
construction of Ref.~\cite{vale2023circuit}.
\end{itemize}

The direct-angle condition applies to the logical Hopf tree: each split
coordinate programs its designated rotation, and each leaf phase programs
its designated phase, without recombining coordinates from different tree
locations. An elementary decomposition of one controlled gate may still
use functions or fractions of that gate's own parameters.
\end{lemma}

\begin{theorem}[Assigned CNOT counts for the native Hopf ordering]
\label{thm:hopf_cnot}
Let \(G_{\mathbb R}(n)\) and \(G_{\mathbb C}(n)\) denote the sums of
assigned CNOT charges over the real and complex native schedules on
\(n\ge2\) data qubits. Use \(C_R(m)\) and \(C_U(m)\) from
Lemma~\ref{lem:cnot_model}, where \(m\) is the number of controls.

For \(n\in\{2,3,4,5\}\), the real-state and complex-state versions of the Hopf
ansatz generally differ at the gate level, because in the complex ansatz the
depth-\((n-1)\) rotations are promoted to \(R_{\mathbb C}\) gates that inject
leaf phases (Definition~\ref{def:hopf_ansatz} and
Algorithm~\ref{alg:order_complex}).  Nevertheless, under the counting model
above, their total assigned CNOT charges coincide for these sizes:
\[
\begin{array}{c|c|c}
n & G_{\mathbb R}(n) & G_{\mathbb C}(n) \\
\hline
2 & 4   & 4   \\
3 & 24  & 24  \\
4 & 100 & 100 \\
5 & 360 & 360
\end{array}
\]
These values are obtained by applying the CNOT-counting model above to the full
generated gate lists, including the promoted \(R_{\mathbb C}\) gates in the
complex case.

For general $n\ge 2$, the control-number distribution induced by the Hopf ordering is as follows.

For the real ansatz, for each $m\in\{0,1,\dots,n-1\}$, there are exactly
$\binom{n}{m}$ $m$-controlled $R_y$ gates. Hence
\[
G_{\mathbb R}(n)
=
\sum_{m=0}^{n-1}\binom{n}{m}\,C_R(m).
\]

For the complex ansatz, the gate types split as follows:
\begin{itemize}
\item for each $m\in\{0,1,\dots,n-2\}$, there are exactly
$\binom{n-1}{m}$ $m$-controlled $R_y$ gates;
\item for each $m\in\{1,2,\dots,n-2\}$, there are exactly
$\binom{n-1}{m-1}$ $m$-controlled $R_{\mathbb C}$ gates;
\item at the top control number $m=n-1$, there are exactly $n$
$(n-1)$-controlled $R_{\mathbb C}$ gates and no $(n-1)$-controlled $R_y$ gate.
\end{itemize}
Therefore the exact sum of assigned complex charges is
\[
G_{\mathbb C}(n)
=
\sum_{m=0}^{n-1}
\left[
\binom{n-1}{m}C_R(m)
+
\binom{n-1}{m-1}C_U(m)
\right]
+
C_U(n-1)-C_R(n-1),
\]
where binomial coefficients with an invalid lower index are interpreted as zero.

In particular, $\mathrm{HopfComplex}$ promotes exactly
\[
\sum_{m=1}^{n-2}\binom{n-1}{m-1}+n
=
\sum_{r=0}^{n-3}\binom{n-1}{r}+n
=
2^{n-1}=N/2
\]
gates to $R_{\mathbb C}$ gates, in agreement with the circuit construction.

Consequently, both ansatze satisfy the asymptotic bounds
\[
G_{\mathbb R}(n)=\bigO(nN),\qquad
G_{\mathbb C}(n)=\bigO(nN).
\]
Moreover, under this assigned ledger,
\[
G_{\mathbb R}(n)<8nN,\qquad
G_{\mathbb C}(n)<9nN
\]
for all sufficiently large $n$.
\end{theorem}

\begin{proof}
The native-ordering invariant above establishes the action of the schedule.
For its control count, a gate in group \(a\), with Hamming weight \(k\),
trailing-zero count \(c\), and target bit \(2^t\), has
\(m=k+c-t-1\) controls. For \(a>0\), there are
\(\binom{n-c-1}{k-1}\) labels with given \(c\) and \(k\).
Introduce a formal marker \(z\) and let \(A_n(z)\) sum \(z^m\) over
all gates. The initial group contributes one gate at each control count
from zero to \(n-1\), while the remaining groups give
\[
\begin{aligned}
A_n(z)
&=\sum_{h=0}^{n-1}z^h
 +\sum_{c=1}^{n-1}z(1+z)^{n-c-1}\sum_{h=0}^{c-1}z^h\\
&=(1+z)^n-z^n.
\end{aligned}
\]
Here \(h\) is a summation index; the second identity follows by summing
the finite geometric series. The coefficient of \(z^m\) is therefore
\(\binom nm\) for \(0\le m<n\).

The promoted complex gates are exactly those with \(t=0\). Their
control-count polynomial, denoted \(P_n(z)\), is
\[
\begin{aligned}
P_n(z)
&=z^{n-1}+\sum_{c=1}^{n-1}z^c(1+z)^{n-c-1}\\
&=z(1+z)^{n-1}+z^{n-1}-z^n.
\end{aligned}
\]
Subtracting gives \(A_n(z)-P_n(z)=(1+z)^{n-1}-z^{n-1}\), the
magnitude-only polynomial. These coefficients give precisely the two
gate-type distributions in the statement. Applying the ledger yields the
finite-size sums and the displayed values. Finally,
\(C_R(m)=16m-24\) and \(C_U(m)\le20m-18\) for \(m\ge5\).
The binomial identities for the gate counts and their first moments give
\[
G_{\mathbb R}(n)=8nN-24N+\bigO(n^4),\qquad
G_{\mathbb C}(n)\le9nN-20N+\bigO(n^4).
\]
The polynomial remainders account for the finitely many small-control
charges and endpoint corrections. They are dominated by the negative
terms proportional to \(N=2^n\), proving the stated eventual strict
bounds and the \(\bigO(nN)\) scaling.
\end{proof}

\FloatBarrier
\section{A four-qubit example of the Hopf ansatz}
\label{app:hopf4_example}

This appendix gives the \(n=4\) real and complex Hopf ansatze as reference
examples for the notation, tables, and circuit diagrams. We write
\[
N=2^4=16,
\qquad
\ell=(q_4q_3q_2q_1)_2\in\{0,\ldots,15\},
\]
with qubits ordered from most significant \(q_4\) to least significant \(q_1\).
The real ansatz has internal-node parameters
\[
\theta_1,\ldots,\theta_{15},
\]
and the complex ansatz adds leaf phases
\[
\theta_{16+\ell},
\qquad
\ell=0,\ldots,15.
\]
The tree, gate schedule, metric table, and layerwise derivative circuits use
the same indexing convention; Fig.~\ref{fig:hopf4_tree} is the reference
diagram.

\subsection*{Hopf tree and leaf-phase indexing}
\label{app:hopf4_tree_indexing}

Figure~\ref{fig:hopf4_tree} shows the complete \(n=4\) Hopf tree. At each
internal node \(j\), a bit-\(0\) branch contributes \(\cos\theta_j\), while a
bit-\(1\) branch contributes \(\sin\theta_j\). Define
\[
f_0(\theta):=\cos\theta,
\qquad
f_1(\theta):=\sin\theta.
\]
For a leaf \(\ell=(q_4q_3q_2q_1)_2\), let
\[
s_1(\ell)=1,
\qquad
s_{t+1}(\ell)=2s_t(\ell)+q_{5-t},
\quad t=1,2,3.
\]
Its path amplitude is
\begin{equation}
a_\ell
=
\prod_{t=1}^{4}
f_{q_{5-t}}\!\left(\theta_{s_t(\ell)}\right),
\qquad
x_\ell^{\mathbb R}=a_\ell,
\qquad
x_\ell^{\mathbb C}=a_\ell e^{\iu\theta_{16+\ell}},
\label{eq:hopf4_path_product}
\end{equation}
with the corresponding real or complex angle ranges.

For example, the first leaf follows the all-left path, the last leaf follows
the all-right path, and the intermediate leaf
\(10=(1010)_2\) follows the right--left--right--left path:
\begin{align}
x_0^{\mathbb R}
&=
\cos\theta_1\cos\theta_2\cos\theta_4\cos\theta_8,
&
x_0^{\mathbb C}
&=
x_0^{\mathbb R}e^{\iu\theta_{16}},
\nonumber\\
x_{10}^{\mathbb R}
&=
\sin\theta_1\cos\theta_3\sin\theta_6\cos\theta_{13},
&
x_{10}^{\mathbb C}
&=
x_{10}^{\mathbb R}e^{\iu\theta_{26}},
\label{eq:hopf4_amplitude_examples}\\
x_{15}^{\mathbb R}
&=
\sin\theta_1\sin\theta_3\sin\theta_7\sin\theta_{15},
&
x_{15}^{\mathbb C}
&=
x_{15}^{\mathbb R}e^{\iu\theta_{31}}.
\nonumber
\end{align}
The general \(n\)-qubit formula is given in
Lemma~\ref{lem:hopf_product_formula}.

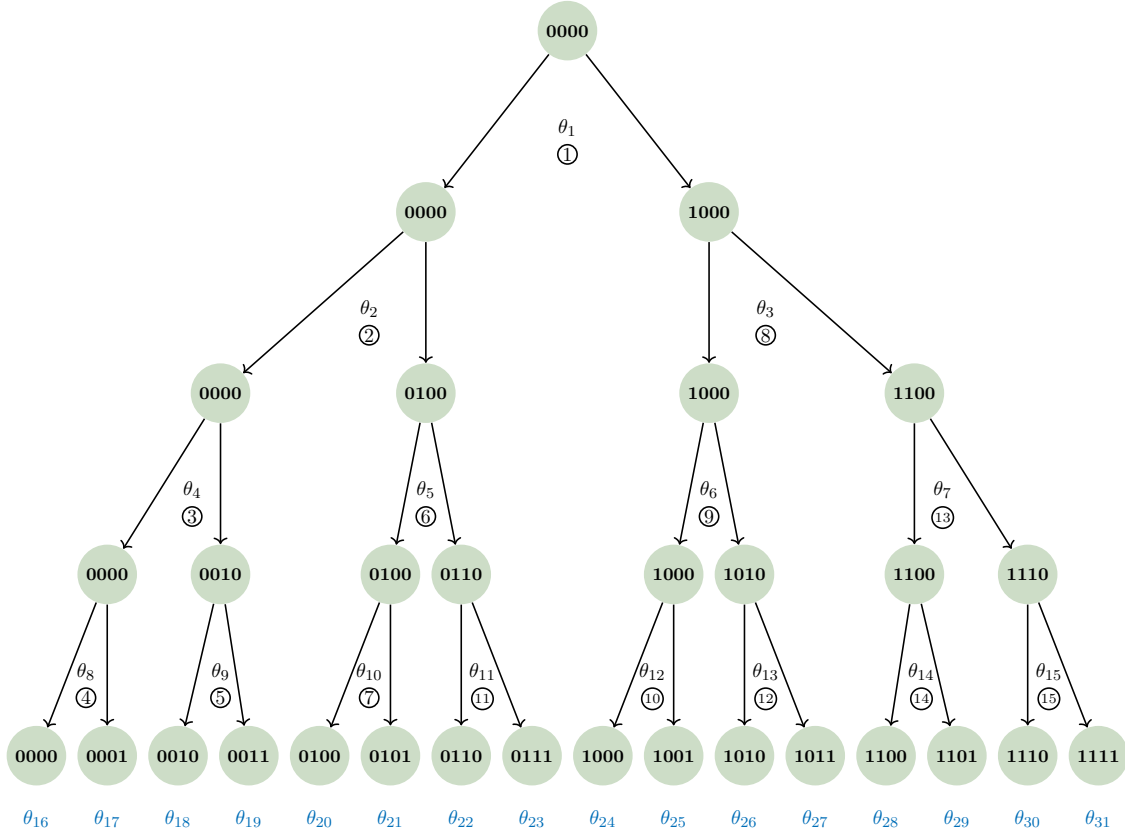
\begin{figure*}[ht]
\resizebox{15cm}{!}{
\begin{tikzpicture}[
  level distance=3.2cm,
  level 1/.style = {sibling distance=5cm},
  level 2/.style = {sibling distance=8cm},
  level 3/.style = {sibling distance=2cm},
  level 4/.style = {sibling distance=1cm},
  every node/.style = {draw=none, fill=OliveGreen!20, circle, minimum size=1cm, font=\bfseries},
  edge from parent/.style = {draw, ->, thick},
  label1/.style = {below=1.2cm, font=\normalsize, draw=none,fill=none},
  label2/.style = {below=1.4cm,circle,inner sep=0.6pt,minimum size=0.6pt, font=\normalsize, draw,fill=none},
  label3/.style = {below=1.4cm,circle,inner sep=0.6pt,minimum size=0.6pt, font=\scriptsize, draw,fill=none},
  label4/.style = {below=1.4cm,circle,inner sep=0.6pt,minimum size=0.6pt, line width=0.8pt,font=\normalsize, draw,fill=none},
  label5/.style = {below=0.6cm, font=\normalsize, draw=none,fill=none, color=NavyBlue}
]

\node (root) {0000}
  child {node (node0) {0000}
    child  [shift={(0.375cm,0)}] {node (node00) {0000}
      child [shift={(-1cm,0)}] {node (node000) {0000}
        child [shift={(-0.75cm,0)}] {node (node0000) {0000} }
        child [shift={(-0.5cm,0)}] {node (node0001) {0001} }
        [shift={(-0.375cm,0)}]  node[label1] {$\theta_{8}$}
        node[label2, below=0.6cm] {4}
      }
      child [shift={(-1cm,0)}] {node (node001) {0010}
        child [shift={(-0.25cm,0)}] {node (node0010) {0010} }
        child {node (node0011) {0011} }
        node[label1] {$\theta_{9}$}
        node[label2, below=0.6cm] {5}
      }
      [shift={(-0.5cm,0)}] node[label1] {$\theta_{4}$}
      node[label2, below=0.6cm] {3}
    }
    child  [shift={(-4cm,0)}] {node (node01) {0100}
      child [shift={(0.375cm,0)}] {node (node010) {0100}
        child [shift={(-0.75cm,0)}] {node (node0100) {0100} }
        child [shift={(-0.5cm,0)}] {node (node0101) {0101} }
        [shift={(-0.375cm,0)}] node[label1] {$\theta_{10}$}
        node[label2, below=0.6cm] {7}
      }
      child [shift={(-0.375cm,0)}] {node (node011) {0110}
        child [shift={(0.5cm,0)}] {node (node0110) {0110}}
        child [shift={(0.75cm,0)}] {node (node0111) {0111}}
        [shift={(0.375cm,0)}] node[label1] {$\theta_{11}$}
        node[label3, below=0.6cm] {11}
      }
      node[label1] {$\theta_{5}$}
      node[label2, below=0.6cm] {6}
    }
    [shift={(-1cm,0)}] node[label1] {$\theta_{2}$}
    node[label2, below=0.6cm] {2}
  }
  child  {node (node1) {1000}
    child [shift={(4cm,0)}]  {node (node10) {1000}
      child [shift={(0.375cm,0)}]{node (node100) {1000}
        child [shift={(-0.75cm,0)}] {node (node1000) {1000}}
        child [shift={(-0.5cm,0)}]{node (node1001) {1001}}
        [shift={(-0.375cm,0)}] node[label1] {$\theta_{12}$}
        node[label3, below=0.6cm] {10}
      }
      child [shift={(-0.375cm,0)}] {node (node101) {1010}
        child [shift={(0.5cm,0)}]{node (node1010) {1010}}
        child [shift={(0.75cm,0)}] {node (node1011) {1011}}
        [shift={(0.375cm,0)}] node[label1] {$\theta_{13}$}
        node[label3, below=0.6cm] {12}
      }
      node[label1] {$\theta_{6}$}
      node[label2, below=0.6cm] {9}
    }
    child [shift={(-0.375cm,0)}] {node (node11) {1100}
      child [shift={(1cm,0)}] {node (node110) {1100}
        child {node (node1100) {1100}}
        child [shift={(0.25cm,0)}] {node (node1101) {1101}}
        [shift={(0.125cm,0)}] node[label1] {$\theta_{14}$}
        node[label3, below=0.6cm] {14}
      }
      child [shift={(1cm,0)}]{node (node111) {1110}
        child [shift={(0.5cm,0)}] {node (node1110) {1110}}
        child [shift={(0.75cm,0)}] {node (node1111) {1111}}
        [shift={(0.375cm,0)}] node[label1] {$\theta_{15}$}
        node[label3, below=0.6cm] {15}
      }
      [shift={(0.5cm,0)}] node[label1] {$\theta_{7}$}
      node[label3, below=0.6cm] {13}
    }
    [shift={(1cm,0)}] node[label1] {$\theta_{3}$}
    node[label2, below=0.6cm] {8}
  }  node[label1] {$\theta_{1}$}
  node[label4, below=0.6cm] {1};

\path (node0000) node[label5] {$\theta_{16}$};
\path (node0001) node[label5] {$\theta_{17}$};
\path (node0010) node[label5] {$\theta_{18}$};
\path (node0011) node[label5] {$\theta_{19}$};
\path (node0100) node[label5] {$\theta_{20}$};
\path (node0101) node[label5] {$\theta_{21}$};
\path (node0110) node[label5] {$\theta_{22}$};
\path (node0111) node[label5] {$\theta_{23}$};
\path (node1000) node[label5] {$\theta_{24}$};
\path (node1001) node[label5] {$\theta_{25}$};
\path (node1010) node[label5] {$\theta_{26}$};
\path (node1011) node[label5] {$\theta_{27}$};
\path (node1100) node[label5] {$\theta_{28}$};
\path (node1101) node[label5] {$\theta_{29}$};
\path (node1110) node[label5] {$\theta_{30}$};
\path (node1111) node[label5] {$\theta_{31}$};

\end{tikzpicture}
}
\caption{Binary tree for preparing a $4$-qubit state in Hopf coordinates. A bit-$0$ branch multiplies the amplitude by $\cos\theta_j$, and a bit-$1$ branch multiplies it by $\sin\theta_j$. The displayed $\theta_j$ labels are the Hopf magnitude parameters associated with internal nodes, while the blue labels $\theta_{16},\ldots,\theta_{31}$ are the complex leaf phases. Circled numbers indicate the gate-application order generated by Algorithms~\ref{alg:order_real} and~\ref{alg:order_complex}. The binary labels inside the nodes are padded path labels identifying the corresponding subtree or leaf, rather than the breadth-first internal-node indices.}
\label{fig:hopf4_tree}
\end{figure*}

\subsection*{Ansatz circuit}
\label{app:hopf4_ansatz_schedule}

In the gate table, a superscript on \(R_y^t\) or \(R_{\mathbb C}^t\)
identifies target qubit \(q_t\), not a power of the gate. The prefixes
\(c_{t_1,\ldots,t_r}\) and \(a_{t_1,\ldots,t_r}\) denote controls on
\(\ket1\) and \(\ket0\), respectively, on the listed qubits.

The tree in Fig.~\ref{fig:hopf4_tree} determines both the amplitude products and the physical gate schedule. Each internal node becomes one controlled single-qubit rotation. In the real ansatz, this is an $R_y$ gate. In the complex ansatz, only the final-layer rotations are promoted to $R_{\mathbb C}$ gates, because they split sibling leaves and can therefore attach the two corresponding leaf phases.

\begin{figure*}[ht]
\centering
\includegraphics[scale=0.25]{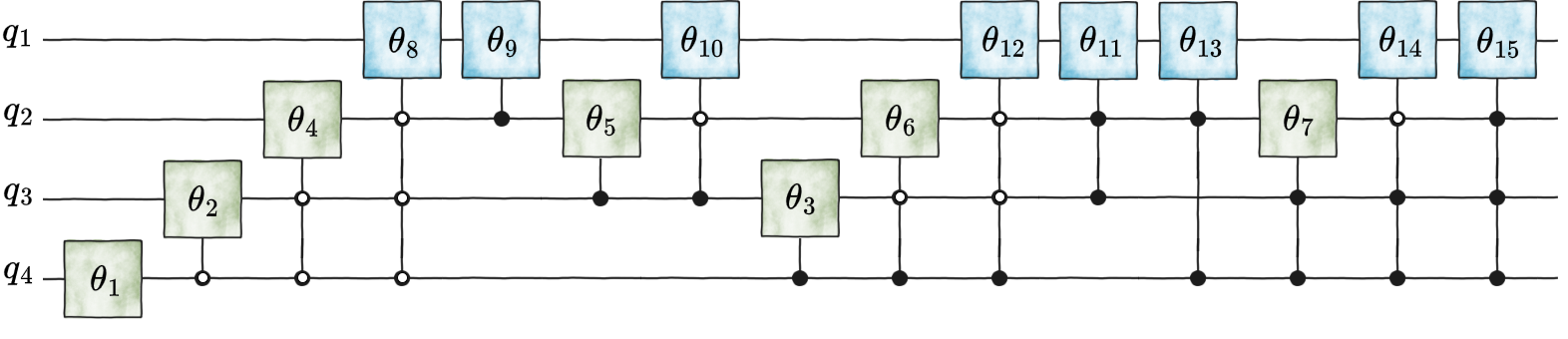}
\caption{Four-qubit Hopf ansatz circuit derived from the tree in Fig.~\ref{fig:hopf4_tree}. The green gates are the magnitude rotations shared by the real-state and complex-state versions of the Hopf ansatz. In the complex ansatz, each final-layer gate $\theta_j$ with $8\le j\le 15$ is promoted to an $R_{\mathbb C}$ gate carrying the magnitude angle $\theta_j$ and the two sibling leaf phases listed in Table~\ref{tab:hopf4_gate_schedule}. All quantum circuit diagrams in this work were produced using the quantum-circuit-drawio-library \cite{wilkens2023quantumcircuitlibrary}.}
\label{fig:hopf4_ansatz_circuit}
\end{figure*}

\begin{table*}[ht]
\centering
\scriptsize
\begin{tabular}{c c c c  | c c | c c }
\hline
Order & $\mathtt{Ctrl}$ & $\mathtt{Anti}$  & $\mathtt{Targ}$  & $\mathtt{Index}$ (real) & Gates (real) & $\mathtt{Index}$ (complex) & Gates (complex)\\
\hline
\Circled{1}  & $0000$ & $0000$ & $1000$ & $1$ & $R_y^4(\theta_{1})$ & $1$ & $R_y^4(\theta_{1})$\\
\Circled{2}  & $0000$ & $1000$ & $0100$ & $2$ & $a_4 R_y^3(\theta_{2})$ & $2$ & $a_4 R_y^3(\theta_{2})$\\
\Circled{3}  & $0000$ & $1100$ & $0010$ & $4$ & $a_{3,4}R_y^2(\theta_{4})$ & $4$ & $a_{3,4}R_y^2(\theta_{4})$\\
\Circled{4}  & $0000$ & $1110$ & $0001$ & $8$ & $a_{2,3,4}R_y^1(\theta_{8})$ & $[8,16,17]$ & $a_{2,3,4}R_{\mathbb{C}}^1(\theta_{8},\theta_{16},\theta_{17})$\\
\Circled{5}  & $0010$ & $0000$ & $0001$ & $9$ & $c_2 R_y^1 (\theta_{9})$ & $[9,18,19]$ & $c_2 R_{\mathbb{C}}^1 (\theta_{9},\theta_{18},\theta_{19})$ \\
\Circled{6}  & $0100$ & $0000$ & $0010$ & $5$ & $c_3 R_y^2(\theta_{5})$ & $5$ & $c_3 R_y^2(\theta_{5})$\\
\Circled{7}  & $0100$ & $0010$ & $0001$ & $10$ & $c_3 a_2 R_y^1(\theta_{10})$ & $[10,20,21]$ & $c_3 a_2 R_{\mathbb{C}}^1(\theta_{10},\theta_{20},\theta_{21})$\\
\Circled{8}  & $1000$ & $0000$ & $0100$ & $3$ & $c_4 R_y^3(\theta_{3})$ & $3$ & $c_4 R_y^3(\theta_{3})$\\
\Circled{9}  & $1000$ & $0100$ & $0010$ & $6$ & $c_4 a_3 R_y^2(\theta_{6})$ & $6$ & $c_4 a_3 R_y^2(\theta_{6})$\\
\Circled{10} & $1000$ & $0110$ & $0001$ & $12$ & $c_4 a_{2,3} R_y^1(\theta_{12})$ & $[12,24,25]$ & $c_4 a_{2,3} R_{\mathbb{C}}^1(\theta_{12},\theta_{24},\theta_{25})$\\
\Circled{11} & $0110$ & $0000$ & $0001$ & $11$ & $c_{2,3} R_y^1(\theta_{11})$ & $[11,22,23]$ & $c_{2,3} R_{\mathbb{C}}^1(\theta_{11},\theta_{22},\theta_{23})$\\
\Circled{12} & $1010$ & $0000$ & $0001$ & $13$ & $c_{2,4} R_y^1(\theta_{13})$ & $[13,26,27]$ & $c_{2,4} R_{\mathbb{C}}^1(\theta_{13},\theta_{26},\theta_{27})$\\
\Circled{13} & $1100$ & $0000$ & $0010$ & $7$ & $c_{3,4} R_y^2(\theta_{7})$ & $7$ & $c_{3,4} R_y^2(\theta_{7})$\\
\Circled{14} & $1100$ & $0010$ & $0001$ & $14$ & $c_{3,4} a_2 R_y^1(\theta_{14})$ & $[14,28,29]$ & $c_{3,4} a_2 R_{\mathbb{C}}^1(\theta_{14},\theta_{28},\theta_{29})$\\
\Circled{15} & $1110$ & $0000$ & $0001$ & $15$ & $c_{2,3,4} R_y^1(\theta_{15})$ & $[15,30,31]$ & $c_{2,3,4} R_{\mathbb{C}}^1(\theta_{15},\theta_{30},\theta_{31})$\\
\hline
\end{tabular}
\caption{Circuit synthesis from Algorithms~\ref{alg:order_real} and \ref{alg:order_complex}. 
Control, anti-control, and target qubit locations are encoded as bitmasks in $\mathtt{Ctrl}$, $\mathtt{Anti}$, and $\mathtt{Targ}$ (a `$1$' denotes an active qubit). Throughout, we order qubits as $(q_n,\ldots,q_1)$ from left to right in bitmasks. Thus, for $n=4$, $\mathtt{Targ}=1000$ targets qubit $4$ and $\mathtt{Targ}=0001$ targets qubit $1$. For the complex ansatz, an integer index denotes a magnitude-only \(R_y\) gate,
while a length-three index denotes a promoted \(R_{\mathbb C}\) gate carrying
one magnitude angle and two leaf phases.}
\label{tab:hopf4_gate_schedule}
\end{table*}

With the gate order and index convention fixed, we return to the same tree as a coordinate chart. A single root-to-leaf path gives an amplitude product; differentiating that product gives the sparse Jacobian row; and the squared probability flowing through each internal node gives the corresponding diagonal metric entry.

\subsection*{Metric and Jacobian entries}
\label{app:hopf4_metric_jacobian}

As a representative real-state path calculation, consider the leaf
$\ell=10=(1010)_2$. Reading the bits from $q_4$ to $q_1$, the path takes the bit sequence $1,0,1,0$. Hence the path uses the factors $\sin\theta_1$, $\cos\theta_3$, $\sin\theta_6$, and $\cos\theta_{13}$, giving
\begin{equation}
\label{eq:hopf4_x10}
x_{10}
=
\sin\theta_1\cos\theta_3\sin\theta_6\cos\theta_{13}.
\end{equation}
The corresponding nonzero real Jacobian entries in the row for $x_{10}$ are
\begin{equation}
\frac{\partial x_{10}}{\partial\theta_1} = \cot\theta_1\,x_{10}, \quad
\frac{\partial x_{10}}{\partial\theta_3} = -\tan\theta_3\,x_{10}, \quad
\frac{\partial x_{10}}{\partial\theta_6} = \cot\theta_6\,x_{10}, \quad
\frac{\partial x_{10}}{\partial\theta_{13}} = -\tan\theta_{13}\,x_{10}.
\label{eq:hopf4_x10_jacobian}
\end{equation}
These logarithmic-derivative expressions are read as direct product
derivatives where a sine or cosine vanishes. In the complex ansatz,
\begin{equation}
\label{eq:hopf4_x10_complex_phase}
x_{10}^{\mathbb C}=x_{10}e^{\iu\theta_{26}},
\qquad
\frac{\partial x_{10}^{\mathbb C}}{\partial\theta_{26}}
=
\iu\,x_{10}^{\mathbb C}.
\end{equation}

The same tree also gives the metric entries. For an internal node $j$, the real metric entry $g^{\mathbb R}_{j,j}$ is the squared amplitude weight accumulated along the path from the root to the parent of $j$. Let
\[
c_j:=\cos^2\theta_j,
\qquad
s_j:=\sin^2\theta_j.
\]
Then the diagonal real Hopf metric for $n=4$ is listed in Table~\ref{tab:hopf4_metric_entries}. In the complex ansatz, the magnitude block has the same entries, while each phase direction is weighted by the occupation probability of its leaf:
\[
g^{\mathbb C}_{j,j}=g^{\mathbb R}_{j,j}\quad(1\le j\le 15),
\qquad
g^{\mathbb C}_{16+\ell,16+\ell}=|x_\ell|^2.
\]

\begin{table*}[ht]
\centering
\scriptsize
\begin{tabular}{c c c | c c c}
\hline
$j$ & Prefix & $g^{\mathbb R}_{j,j}$ & $j$ & Prefix & $g^{\mathbb R}_{j,j}$ \\
\hline
1 & $\emptyset$ & $1$             & 9  & $001$ & $c_1c_2s_4$ \\
2 & $0$          & $c_1$          & 10 & $010$ & $c_1s_2c_5$ \\
3 & $1$          & $s_1$          & 11 & $011$ & $c_1s_2s_5$ \\
4 & $00$         & $c_1c_2$       & 12 & $100$ & $s_1c_3c_6$ \\
5 & $01$         & $c_1s_2$       & 13 & $101$ & $s_1c_3s_6$ \\
6 & $10$         & $s_1c_3$       & 14 & $110$ & $s_1s_3c_7$ \\
7 & $11$         & $s_1s_3$       & 15 & $111$ & $s_1s_3s_7$ \\
8 & $000$        & $c_1c_2c_4$    &    &       &             \\
\hline
\end{tabular}
\caption{Diagonal metric entries for the real $n=4$ Hopf ansatz. With $d_j=\lfloor\log_2j\rfloor$ the node depth, the prefix is the
length-$d_j$ binary expansion of $j-2^{d_j}$, padded with leading zeros.}
\label{tab:hopf4_metric_entries}
\end{table*}

The metric entries in Table~\ref{tab:hopf4_metric_entries} are precisely the normalization factors for coordinate derivatives. Dividing a raw derivative by $\sqrt{g_{j,j}}$ removes the amplitude weight accumulated before node $j$, leaving a normalized state supported only on the subtree rooted at $j$. We now spell this out for two magnitude parameters and one phase parameter.

\subsection*{Representative normalized coordinate tangents}
\label{app:hopf4_normalized_partials}

First, consider the magnitude parameter $\theta_5$. Its subtree contains the
leaves $\{4,5,6,7\}$. To distinguish the two ansatze in these examples, write
\(\ket{\psi_{\mathbb R}(\boldsymbol{\theta})}\) and
\(\ket{\psi_{\mathbb C}(\boldsymbol{\theta})}\) for the real and complex Hopf
states, respectively.

For the real ansatz, after normalization by
\(\sqrt{g^{\mathbb R}_{5,5}}=\sqrt{c_1s_2}\), all factors from the path above
node \(5\) cancel, and the normalized coordinate tangent is supported only on
that subtree:
\begin{equation}
\ket{e^{\mathbb R}_5(\boldsymbol{\theta})}
=
\sum_{\ell=4}^{7}v^{(5,\mathbb R)}_\ell\ket{b_\ell},
\end{equation}
where
\begin{equation}
(v^{(5,\mathbb R)}_4,v^{(5,\mathbb R)}_5,
 v^{(5,\mathbb R)}_6,v^{(5,\mathbb R)}_7)
=
(-\sin\theta_5\cos\theta_{10},
 -\sin\theta_5\sin\theta_{10},
  \cos\theta_5\cos\theta_{11},
  \cos\theta_5\sin\theta_{11}).
\end{equation}

For the complex ansatz, the magnitude factors are the same, but the inherited
leaf phases on the subtree must be kept:
\begin{equation}
\ket{e^{\mathbb C}_5(\boldsymbol{\theta})}
=
\sum_{\ell=4}^{7}v^{(5,\mathbb C)}_\ell\ket{b_\ell},
\end{equation}
where
\begin{equation}
\begin{aligned}
(v^{(5,\mathbb C)}_4,v^{(5,\mathbb C)}_5,
 v^{(5,\mathbb C)}_6,v^{(5,\mathbb C)}_7)
=
(-\sin\theta_5\cos\theta_{10}e^{\iu\theta_{20}},
  -\sin\theta_5\sin\theta_{10}e^{\iu\theta_{21}},
\ \ \cos\theta_5\cos\theta_{11}e^{\iu\theta_{22}},
   \ \cos\theta_5\sin\theta_{11}e^{\iu\theta_{23}}).
\end{aligned}
\end{equation}

As a final-layer magnitude example, the parameter \(\theta_{13}\) splits only
the sibling leaves \(\ell=10\) and \(\ell=11\). For the real ansatz, after
normalization by
\(\sqrt{g^{\mathbb R}_{13,13}}=\sqrt{s_1c_3s_6}\), the normalized tangent is
\begin{equation}
\ket{e^{\mathbb R}_{13}(\boldsymbol{\theta})}
=
-\sin\theta_{13}\ket{b_{10}}
+
\cos\theta_{13}\ket{b_{11}}.
\end{equation}
For the complex ansatz, the corresponding magnitude tangent keeps the two leaf
phases:
\begin{equation}
\ket{e^{\mathbb C}_{13}(\boldsymbol{\theta})}
=
-\sin\theta_{13}e^{\iu\theta_{26}}\ket{b_{10}}
+
\cos\theta_{13}e^{\iu\theta_{27}}\ket{b_{11}}.
\end{equation}

For a phase parameter, the tangent is supported on a single leaf. In particular,
\(\theta_{26}=\theta_{16+10}\) is attached to \(\ell=10\), so on the regular set
where \(g^{\mathbb C}_{26,26}=|x_{10}|^2>0\),
\begin{equation}
\ket{e^{\mathbb C}_{26}(\boldsymbol{\theta})}
=
\iu e^{\iu\theta_{26}}\ket{b_{10}}.
\end{equation}

Table~\ref{tab:hopf4_tangent_assignments} lists representative gate assignments
that synthesize these normalized directions using the same Hopf gate skeleton.
These assignments are physical gate settings and may lie outside the canonical
coordinate ranges. All unspecified magnitude angles are set to zero; in complex
rows, all unspecified phase angles are also set to zero.

\begin{table*}[ht]
\centering
\scriptsize
\begin{tabular}{c | p{0.64\textwidth} | c}
\hline
Target direction & Nonzero or kept assignments & Support \\
\hline
$\ket{e^{\mathbb R}_5}$
&
$\theta^{(5)}_1=0$,
$\theta^{(5)}_2=\pi/2$,
$\theta^{(5)}_5=\theta_5+\pi/2$;
$\theta^{(5)}_{10}=\theta_{10}$,
$\theta^{(5)}_{11}=\theta_{11}$
& $\ell=4,5,6,7$ \\
\hline
$\ket{e^{\mathbb C}_5}$
&
$\theta^{(5)}_1=0$,
$\theta^{(5)}_2=\pi/2$,
$\theta^{(5)}_5=\theta_5+\pi/2$;
$\theta^{(5)}_{10}=\theta_{10}$,
$\theta^{(5)}_{11}=\theta_{11}$;
$\theta^{(5)}_{20}=\theta_{20}$,
$\theta^{(5)}_{21}=\theta_{21}$,
$\theta^{(5)}_{22}=\theta_{22}$,
$\theta^{(5)}_{23}=\theta_{23}$
& $\ell=4,5,6,7$ \\
\hline
$\ket{e^{\mathbb R}_{13}}$
&
$\theta^{(13)}_1=\pi/2$,
$\theta^{(13)}_3=0$,
$\theta^{(13)}_6=\pi/2$,
$\theta^{(13)}_{13}=\theta_{13}+\pi/2$
& $\ell=10,11$ \\
\hline
$\ket{e^{\mathbb C}_{13}}$
&
$\theta^{(13)}_1=\pi/2$,
$\theta^{(13)}_3=0$,
$\theta^{(13)}_6=\pi/2$,
$\theta^{(13)}_{13}=\theta_{13}+\pi/2$;
$\theta^{(13)}_{26}=\theta_{26}$,
$\theta^{(13)}_{27}=\theta_{27}$
& $\ell=10,11$ \\
\hline
$\ket{e^{\mathbb C}_{26}}$
&
$\theta^{(26)}_1=\pi/2$,
$\theta^{(26)}_3=0$,
$\theta^{(26)}_6=\pi/2$,
$\theta^{(26)}_{13}=0$,
$\theta^{(26)}_{26}=\theta_{26}+\pi/2$
& $\ell=10$ \\
\hline
\end{tabular}
\caption{Representative gate assignments that synthesize normalized coordinate
tangents in the $n=4$ example. The rows distinguish real magnitude tangents
from complex magnitude tangents; the latter keep the original leaf phases on the
relevant subtree, as required by Theorem~\ref{thm:synthesis_normalized_direction_hopf}.}
\label{tab:hopf4_tangent_assignments}
\end{table*}

The preceding examples prepare one tangent direction at a time. The layerwise construction batches all tangent directions at the same tree depth by adding an index register. The $d=2$ layer is the first nontrivial four-label example in the $n=4$ ansatz, so it is the natural case to display explicitly.

\subsection*{Layerwise derivative and branch circuits}
\label{app:hopf4_layerwise_branch}

For depth $d=2$, the magnitude layer is $\mathcal{L}_2=\{4,5,6,7\}$. A two-qubit index register labels the offset within this layer. We write the register as $I_1I_2$, with $I_1$ the most significant label bit, so the four computational-basis labels select the four tangent assignments as follows:
\begin{equation}
00\mapsto\theta_4,
\qquad
01\mapsto\theta_5,
\qquad
10\mapsto\theta_6,
\qquad
11\mapsto\theta_7.
\label{eq:hopf4_layer2_label_map}
\end{equation}

\begin{figure*}[ht]
\centering
\includegraphics[scale=0.25]{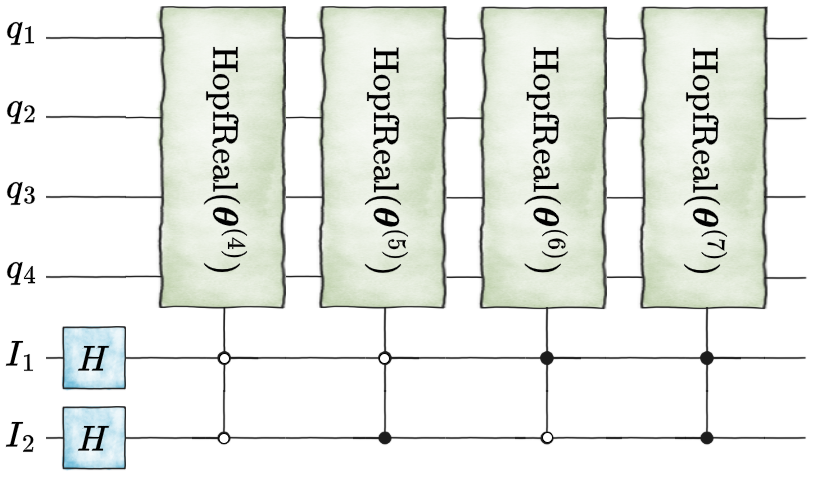}
\caption{Indexed derivative preparation for the $d=2$ magnitude layer of the four-qubit real Hopf ansatz. The index register $I_1I_2$ selects one of the tangent-state assignments $\boldsymbol{\theta}^{(j)}$ with $j\in\{4,5,6,7\}$ according to Eq.~\eqref{eq:hopf4_layer2_label_map}. Conditioned on the observed index, the system register is in the corresponding normalized coordinate tangent.}
\label{fig:hopf4_layerwise_circuit}
\end{figure*}

\begin{figure*}[ht]
\centering
\includegraphics[scale=0.25]{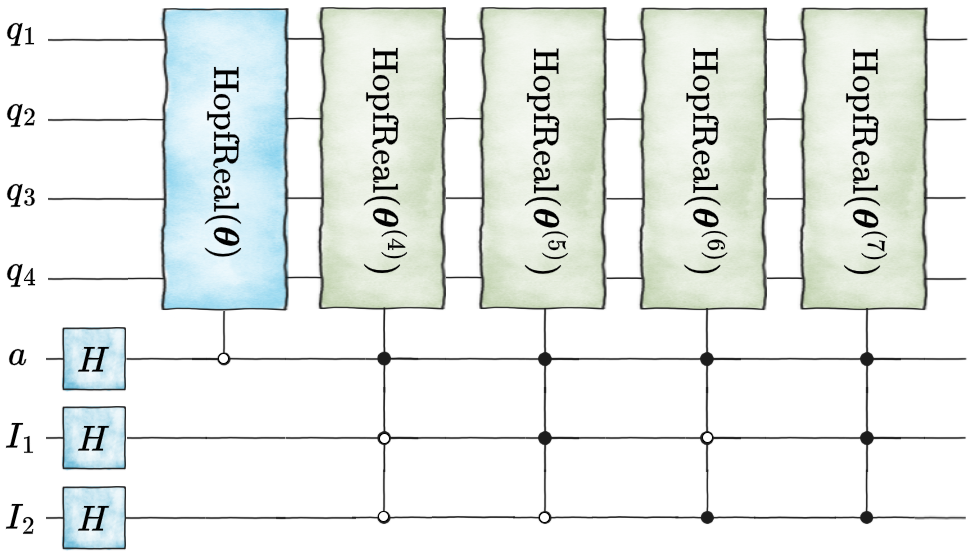}
\caption{Indexed signed-branch preparation for the same $d=2$ magnitude layer. The branch $a=0$ prepares the baseline state $\mathrm{HopfReal}(\boldsymbol{\theta})\ket{0}^{\otimes4}$, while the branch $a=1$ prepares the index-selected tangent state $\mathrm{HopfReal}(\boldsymbol{\theta}^{(j)})\ket{0}^{\otimes4}$. Measuring $a$ in the $X$ basis, equivalently by adding a Hadamard before $Z$-basis readout, yields the signed branch states used in the estimator.}
\label{fig:hopf4_branch_circuit}
\end{figure*}

\FloatBarrier
\section{Geometry-native optimizers for the Hopf ansatz}
\label{app:optimizers}

This appendix records the optimizer layer used with the Hopf ansatz, focusing on
methods whose quantum interface consists of objective evaluations and
Hopf-coordinate gradient evaluations. All geometry-dependent operations below are
classical post-processing: diagonal metric rescaling, Jacobian lifting, exact
sphere exponential maps, exact vector transport, line searches, and the inverse
Hopf map. Adam-type coordinate updates are used in the main text as baselines;
the appendix focuses on optimizers that explicitly use the Hopf metric and
state-sphere geometry.

The starting point is the same cost-and-gradient loop used in standard
variational algorithms. Let \(C\) be a differentiable real-valued objective. At a parameter value \(\boldsymbol{\theta}\), the
quantum device supplies an estimate of
\begin{equation}
C(\boldsymbol{\theta})
=
C\big(\ket{\psi(\boldsymbol{\theta})}\big)
\end{equation}
and of the coordinate-gradient components
\begin{equation}
c_i
:=
\partial_{\theta_i} C(\boldsymbol{\theta}).
\end{equation}
The Hopf structure then converts these coordinate data into state-sphere data
using analytic information from the chart: the inverse map, diagonal metric, and
raw Jacobian directions are all known explicitly.

\subsection*{Jacobian formulas and inverse differential map}

The main text uses only the diagonal metric, but the optimizer layer also uses
the differential of the Hopf chart.  The Jacobian maps parameter variations to
state-space tangent vectors; its columns are the raw coordinate tangents in the
ambient state-vector basis.

\begin{theorem}[Jacobian formulas and Gram forms]
\label{thm:hopf_jacobian_both}
Let \(N=2^n\), and let
\[
\mathbf{x}(\boldsymbol{\theta})
=
(x_0,\ldots,x_{N-1})
\]
be the Hopf coordinate map.

In the real case define
\begin{equation}
J^{\mathbb R}(\boldsymbol{\theta})
\in \mathbb{R}^{N\times(N-1)},
\qquad
(J^{\mathbb R})_{a j}
:=
\frac{\partial x_a}{\partial\theta_j},
\end{equation}
with \(a=0,\ldots,N-1\) and \(j=1,\ldots,N-1\).  In the complex case define
\begin{equation}
J^{\mathbb C}(\boldsymbol{\theta})
\in \mathbb{C}^{N\times(2N-1)},
\qquad
(J^{\mathbb C})_{a j}
:=
\frac{\partial x_a}{\partial\theta_j},
\end{equation}
with \(a=0,\ldots,N-1\) and \(j=1,\ldots,2N-1\).

For a fixed amplitude index \(a\), let \(P(a)\) be the set of internal nodes on
the unique root-to-leaf path that generates \(x_a\).  For \(j\in P(a)\), define
\(\varepsilon_{a j}\in\{0,1\}\) by
\[
\varepsilon_{a j} =
\begin{cases}
0, & \text{if the path contributes a }\cos\theta_j\text{ factor},\\
1, & \text{if the path contributes a }\sin\theta_j\text{ factor}.
\end{cases}
\]
The following logarithmic-derivative formulas hold literally where the displayed \(\tan\) and \(\cot\) factors are finite.  At zeros of the
corresponding sine or cosine factor, they are understood through the direct
product derivative, equivalently by continuous extension.

(A) Real case.  For \(1\le j\le N-1\),
\begin{equation}
\frac{\partial x_a}{\partial\theta_j}=
\begin{cases}
-\tan\theta_j\,x_a, & j\in P(a),\ \varepsilon_{a j}=0,\\[2pt]
\ \ \cot\theta_j\,x_a, & j\in P(a),\ \varepsilon_{a j}=1,\\[2pt]
0, & j\notin P(a).
\end{cases}
\label{eq:jac_real_entries}
\end{equation}
Moreover, the real pullback metric is the Gram matrix
\begin{equation}
g^{\mathbb R}(\boldsymbol{\theta})
=
\big(J^{\mathbb R}(\boldsymbol{\theta})\big)^\top
J^{\mathbb R}(\boldsymbol{\theta}).
\label{eq:g_real_gram}
\end{equation}

(B) Complex case.  Write
\[
x_a(\boldsymbol{\theta})
=
r_a(\boldsymbol{\theta})e^{\iu\theta_{N+a}},
\qquad
r_a\ge0,
\]
where \(r_a\) is produced by the internal-node magnitude angles.  For
\(1\le j\le N-1\), the same path-derivative rule holds, and for the phase block
there is one additional nonzero derivative:
\begin{equation}
\frac{\partial x_a}{\partial\theta_j}=
\begin{cases}
-\tan\theta_j\,x_a, & j\in P(a),\ \varepsilon_{a j}=0,\\[2pt]
\ \ \cot\theta_j\,x_a, & j\in P(a),\ \varepsilon_{a j}=1,\\[2pt]
\iu\,x_a, & j=N+a,\\[2pt]
0, & \text{otherwise}.
\end{cases}
\label{eq:jac_complex_entries}
\end{equation}
Moreover, under the ambient round-sphere convention of
Definition~\ref{def:hopf_metric_convention}, the complex pullback metric is
\begin{equation}
g^{\mathbb C}(\boldsymbol{\theta})
=
\operatorname{Re}\!\left[
\big(J^{\mathbb C}(\boldsymbol{\theta})\big)^\dagger
J^{\mathbb C}(\boldsymbol{\theta})
\right].
\label{eq:g_complex_gram}
\end{equation}
\end{theorem}

Thus the Hopf Jacobian is explicit, and its Gram matrix is exactly the pullback
metric.  Since the metric is diagonal, the differential map can be inverted
componentwise on the regular set, rather than by solving a dense linear system.

Define the regular parameter regions
\[
\mathcal{U}_{\mathbb R}
:=
\left\{
\boldsymbol{\theta}\mid
g^{\mathbb R}_{i,i}(\boldsymbol{\theta})>0\ \forall i
\right\},
\qquad
\mathcal{U}_{\mathbb C}
:=
\left\{
\boldsymbol{\theta}\mid
g^{\mathbb C}_{i,i}(\boldsymbol{\theta})>0\ \forall i
\right\}.
\]
On \(\mathcal{U}_{\mathbb R}\), the Jacobian \(J^{\mathbb R}\) has full column
rank and its Moore--Penrose left inverse is
\[
(J^{\mathbb R})^+
=
(g^{\mathbb R})^{-1}(J^{\mathbb R})^\top,
\qquad
(J^{\mathbb R})^+J^{\mathbb R}=I_{N-1}.
\]
For the complex ansatz, after realification
\[
J_{\mathbb R}
:=
\begin{bmatrix}
\operatorname{Re}J^{\mathbb C}\\
\operatorname{Im}J^{\mathbb C}
\end{bmatrix},
\]
the corresponding left inverse on \(\mathcal{U}_{\mathbb C}\) is
\[
(J_{\mathbb R})^+
=
(g^{\mathbb C})^{-1}(J_{\mathbb R})^\top,
\qquad
(J_{\mathbb R})^+J_{\mathbb R}=I_{2N-1}.
\]
Consequently, on the regular set, the diagonal metric not only gives
componentwise Riemannian rescaling, but also gives an explicit inverse
differential map from state-space tangent vectors back to Hopf-coordinate
variations.

\subsection*{State-sphere conventions}

Let \(\mathbb{S}_\star\) denote the normalized state-vector sphere on which the
optimizer moves.  In the real ansatz,
\(\mathbb{S}_\star=\mathbb{S}^{N-1}\subset\mathbb{R}^{N}\).  In the complex
ansatz, \(\mathbb{C}^{N}\) is regarded as a real Euclidean space and
\(\mathbb{S}_\star=\mathbb{S}^{2N-1}\). We use vectors and kets
interchangeably in this appendix. For ambient vectors \(u,v\), the real pairing is
\begin{equation}
\langle u,v\rangle_\star
=
\begin{cases}
\braket{u|v}, & \text{real case},\\[2pt]
\operatorname{Re}\braket{u|v}, & \text{complex round-sphere case}.
\end{cases}
\label{eq:optimizer_star_pairing}
\end{equation}
Its induced norm is \(\|u\|_\star=\sqrt{\langle u,u\rangle_\star}\).
The tangent space is
\(T_\psi\mathbb S_\star=\{u:\langle\psi,u\rangle_\star=0\}\).
For an ambient vector \(w\), the tangent projection at
\(\ket{\psi}\in\mathbb{S}_\star\) is
\begin{equation}
\Pi_\psi(w)
=
w-
\langle \psi,w\rangle_\star\psi .
\label{eq:optimizer_tangent_projection}
\end{equation}
For \(u\in T_\psi\mathbb{S}_\star\), the exact sphere exponential map is
\begin{equation}
\operatorname{Exp}_{\psi}(u)
=
\begin{cases}
\psi, & \|u\|_\star=0,\\[4pt]
\cos(\|u\|_\star)\psi
+
\dfrac{\sin(\|u\|_\star)}{\|u\|_\star}u,
& \|u\|_\star>0.
\end{cases}
\label{eq:optimizer_sphere_exp}
\end{equation}
For a nonnegative step parameter \(\eta\), the geodesic is
\(\gamma(\eta)=\operatorname{Exp}_{\psi}(\eta u)\), with angular length
\(\eta\|u\|_\star\).  This is the exact-geodesic update used in
Ref.~\cite{ferreira2025quantum}.

The exact parallel transport of a tangent vector
\(\zeta\in T_\psi\mathbb{S}_\star\) along the same geodesic
\(\gamma(\eta)=\operatorname{Exp}_{\psi}(\eta u)\) is
\begin{equation}
\mathcal{T}_{\eta u}(\zeta)
=
\zeta
-
\sin(\eta\|u\|_\star)
\frac{\langle u,\zeta\rangle_\star}{\|u\|_\star}\,
\psi
+
\big[\cos(\eta\|u\|_\star)-1\big]
\frac{\langle u,\zeta\rangle_\star}{\|u\|_\star^2}\,
u,
\label{eq:optimizer_exact_transport}
\end{equation}
with the continuous identity limit when \(\|u\|_\star=0\).  This is the
transport used by Hopf-EGT-CG and by the Riemannian memory methods below.  It
is a vector-transport rule along the accepted geodesic; it is not a second
state-update geodesic.

\subsection*{Lifting Hopf coordinate gradients to the state sphere}

For brevity in this appendix, write
\begin{equation}
\ket{\dot\psi_i}
:=
\partial_{\theta_i}\ket{\psi(\boldsymbol{\theta})},
\qquad
\ket{e_i}
:=
\frac{\ket{\dot\psi_i}}{\sqrt{g_{i,i}(\boldsymbol{\theta})}}
\end{equation}
where \(g_{i,i}\) is the diagonal Hopf metric from
Theorem~\ref{thm:hopf_metric_diagonal}.  On the regular set,
\begin{equation}
\langle e_i,e_j\rangle_\star=\delta_{ij}.
\end{equation}
For a real coordinate displacement \(\delta\boldsymbol\theta\),
the differential \(D_{\boldsymbol\theta}\psi\) maps that displacement to
a state tangent in the orthonormal-frame form
\begin{equation}
D_{\boldsymbol{\theta}}\psi[\delta\boldsymbol{\theta}]
=
\sum_i
\delta\theta_i\sqrt{g_{i,i}(\boldsymbol{\theta})}\,e_i .
\label{eq:optimizer_hopf_differential}
\end{equation}

If the coordinate gradient components
\(c_i=\partial_{\theta_i} C\) are known, the corresponding state-sphere gradient
is
\begin{equation}
G(\boldsymbol{\theta})
:=
\operatorname{grad}_{\mathbb{S}_\star} C
=
\sum_i
\frac{c_i}{\sqrt{g_{i,i}(\boldsymbol{\theta})}}e_i
=
\sum_i
\frac{c_i}{g_{i,i}(\boldsymbol{\theta})}\dot\psi_i.
\label{eq:optimizer_state_gradient_from_hopf}
\end{equation}
In matrix notation this is \(G=Jg^{-1}\nabla_{\theta}C\), where \(J\) is the
Hopf Jacobian and \(g=J^{\top}J\) in the real case, or the corresponding
realified Gram matrix in the complex round-sphere convention.

Equation~\eqref{eq:optimizer_state_gradient_from_hopf} is the optimizer-side use
of the Hopf metric.  Since \(g\) is diagonal, no dense metric inversion is
performed.  In the real tree, the state-gradient lift can be evaluated without
materializing \(J\): for each leaf, only the \(n\) internal nodes on its
root-to-leaf path contribute.  Thus the lift costs \(\bigO(nN)\) classical time
and \(\bigO(N)\) memory for an \(n\)-qubit real Hopf state, after the coordinate
gradient has been supplied.

At iteration \(k\), write \(\psi_k=\psi(\boldsymbol\theta_k)\).
After a state-sphere update produces \(\psi_{k+1}\), the next circuit parameters
are obtained by the inverse Hopf map,
\begin{equation}
\boldsymbol{\theta}_{k+1}
=
\operatorname{Hopf}^{-1}(\psi_{k+1}),
\label{eq:optimizer_inverse_return}
\end{equation}
using Lemma~\ref{lem:inverse_hopf_map}.  This is the step that lets a classical
state-sphere optimizer remain an executable Hopf-circuit optimizer.

\subsection*{Hopf exact geodesic transport with conjugate gradients}

The basic exact-geodesic-transport descent direction is
\begin{equation}
v_k
:=
-G_k,
\label{eq:optimizer_egt_descent}
\end{equation}
where \(G_k=G(\boldsymbol{\theta}_k)\). For a positive step size
\(\eta_k\), the first-order Hopf-EGT update is
\begin{equation}
\psi_{k+1}
=
\operatorname{Exp}_{\psi_k}(\eta_k v_k),
\qquad
\boldsymbol{\theta}_{k+1}
=
\operatorname{Hopf}^{-1}(\psi_{k+1}).
\label{eq:optimizer_egt_update}
\end{equation}
Infinitesimally, this reduces to the diagonal natural-gradient parameter step
\begin{equation}
\delta\theta_i
=
-
\eta_k
\frac{\partial_{\theta_i} C}{g_{i,i}}
+
\bigO(\eta_k^2),
\end{equation}
but Eq.~\eqref{eq:optimizer_egt_update} keeps the state normalized exactly and
uses the full sphere geodesic rather than a flat coordinate step.

Hopf-EGT-CG adds a transported conjugate-gradient memory direction.  Let
\(u_k\in T_{\psi_k}\mathbb{S}_\star\) denote the search direction.  Initialize
\begin{equation}
u_0=v_0=-G_0.
\end{equation}
Given \(u_k\), choose \(\eta_k>0\) and update
\begin{equation}
\psi_{k+1}
=
\operatorname{Exp}_{\psi_k}(\eta_k u_k).
\end{equation}
Then transport the previous direction and previous gradient to
\(T_{\psi_{k+1}}\mathbb{S}_\star\):
\begin{equation}
\widetilde u_k
:=
\mathcal{T}_{\eta_k u_k}(u_k),
\qquad
\widetilde G_k
:=
\mathcal{T}_{\eta_k u_k}(G_k).
\end{equation}
Using the hybrid Dai--Yuan/Hestenes--Stiefel coefficient in gradient notation~\cite{dai1999nonlinear,hestenes1952methods},
define
\begin{equation}
D_k
:=
\langle G_{k+1},\widetilde u_k\rangle_\star
-
\langle G_k,u_k\rangle_\star,
\end{equation}
\begin{equation}
\beta_{k+1}^{\rm DY}
:=
\frac{
\langle G_{k+1},G_{k+1}\rangle_\star
}{D_k},
\qquad
\beta_{k+1}^{\rm HS}
:=
\frac{
\langle G_{k+1},G_{k+1}-\widetilde G_k\rangle_\star
}{D_k},
\end{equation}
and
\begin{equation}
\beta_{k+1}
:=
\max\{0,\min(\beta_{k+1}^{\rm DY},\beta_{k+1}^{\rm HS})\}.
\label{eq:optimizer_hybrid_beta}
\end{equation}
If \(D_k\) is numerically singular or the resulting direction is not a descent
direction, we restart with \(u_{k+1}=-G_{k+1}\).  Otherwise,
\begin{equation}
u_{k+1}
=
-G_{k+1}
+
\beta_{k+1}\widetilde u_k .
\label{eq:optimizer_egtcg_direction}
\end{equation}

The line search is performed on the one-dimensional geodesic
\begin{equation}
\gamma_k(\eta)
=
\operatorname{Exp}_{\psi_k}(\eta u_k).
\end{equation}
For fixed line-search constants \(0<\omega_1<\omega_2<1\), the
strong-Wolfe conditions take the sphere form
\begin{align}
C(\gamma_k(\eta_k))-C(\psi_k)
&\le
\omega_1\eta_k\langle G_k,u_k\rangle_\star,
\label{eq:optimizer_wolfe_armijo}\\
\left|
\left\langle
G_{k+1},
\mathcal{T}_{\eta_k u_k}(u_k)
\right\rangle_\star
\right|
&\le
\omega_2
\left|\langle G_k,u_k\rangle_\star\right|,
\label{eq:optimizer_wolfe_curvature}
\end{align}
These are the cost-and-gradient tests used by the EGT-CG
optimizer of Ref.~\cite{ferreira2025quantum}.  They require additional
objective and, when the curvature condition is checked, gradient evaluations;
they do not require any new quantum primitive beyond the cost and Hopf-gradient
queries.

\subsection*{Hopf Riemannian L-BFGS}

The Riemannian L-BFGS variant~\cite{liu1989limited,absil2008optimization} keeps a limited memory of curvature pairs in
state-sphere tangent spaces.  At iteration \(k\), assume the accepted update was
\begin{equation}
\psi_{k+1}
=
\operatorname{Exp}_{\psi_k}(\eta_k p_k),
\end{equation}
where \(p_k\in T_{\psi_k}\mathbb{S}_\star\) is the L-BFGS search direction.  The
step and gradient-difference vectors at the new point are
\begin{equation}
s_k
:=
\mathcal{T}_{\eta_k p_k}(\eta_k p_k),
\qquad
y_k
:=
G_{k+1}
-
\mathcal{T}_{\eta_k p_k}(G_k).
\label{eq:optimizer_lbfgs_pairs}
\end{equation}
The pair is accepted only when
\begin{equation}
\langle s_k,y_k\rangle_\star>0.
\label{eq:optimizer_lbfgs_curvature}
\end{equation}
All older stored pairs are exactly transported along each newly accepted
geodesic before the next two-loop recursion is applied.  Thus every vector used
by the recursion lies in the current tangent space.

Let \(m_{\mathrm{mem}}\) be the maximum memory size and
\(\ell_k\le\min(k,m_{\mathrm{mem}})\) the number of stored pairs at
iteration \(k\). Label the transported pairs by
\((s_j,y_j)\), \(j=k-\ell_k,\ldots,k-1\), and set
\begin{equation}
\rho_j
:=
\frac{1}{\langle s_j,y_j\rangle_\star}.
\end{equation}
The standard two-loop recursion is then applied with the pairing
\(\langle\cdot,\cdot\rangle_\star\). Starting from the working vector
\(q=G_k\), the backward loop visits pairs from newest to oldest and uses
the auxiliary scalars \(\alpha_j\):
\begin{equation}
\alpha_j
=
\rho_j\langle s_j,q\rangle_\star,
\qquad
q\leftarrow q-
\alpha_j y_j.
\end{equation}
The initial inverse-Hessian scale may be chosen as
\begin{equation}
\sigma_k
=
\frac{\langle s_{k-1},y_{k-1}\rangle_\star}
{\langle y_{k-1},y_{k-1}\rangle_\star},
\end{equation}
when a previous pair exists, and \(\sigma_k=1\) otherwise. Setting
the second working vector to \(h=\sigma_k q\), the forward loop visits
pairs from oldest to newest and uses auxiliary scalars \(\beta_j\):
\begin{equation}
\beta_j
=
\rho_j\langle y_j,h\rangle_\star,
\qquad
h\leftarrow h+s_j(\alpha_j-\beta_j).
\end{equation}
The search direction is
\begin{equation}
p_k
=
-\Pi_{\psi_k}(h),
\label{eq:optimizer_lbfgs_direction}
\end{equation}
with restart to \(-G_k\) if \(p_k\) is not a descent direction.  A geodesic line
search, for instance the strong-Wolfe search
\eqref{eq:optimizer_wolfe_armijo}--\eqref{eq:optimizer_wolfe_curvature}, then
chooses \(\eta_k\).

This optimizer uses the transparent Hopf geometry in two places.  First,
\(G_k\) is obtained from Hopf-coordinate gradients by the analytic diagonal
metric lift \eqref{eq:optimizer_state_gradient_from_hopf}.  Second, the memory
vectors are kept in the correct tangent spaces by exact sphere transport, rather
than by treating the parameter space as flat.

\subsection*{Hopf Riemannian Barzilai--Borwein}

The Riemannian Barzilai--Borwein method~\cite{barzilai1988two} is a spectral-gradient method on the
same state sphere.  It has no multi-vector memory and therefore provides a
simpler third cost-and-gradient optimizer.

Write the previous accepted displacement from \(\psi_{k-1}\) to
\(\psi_k\) as \(\eta_{k-1}p_{k-1}\), where \(p_{k-1}\) is its
tangent search direction and \(\eta_{k-1}>0\) its step size. Define
\begin{equation}
s_{k-1}
:=
\mathcal{T}(\eta_{k-1}p_{k-1}),
\qquad
y_{k-1}
:=
G_k-
\mathcal{T}(G_{k-1}),
\end{equation}
where \(\mathcal{T}\) denotes exact transport along the accepted geodesic.  The
two standard spectral choices are
\begin{equation}
\alpha_k^{\rm BB1}
=
\frac{\langle s_{k-1},s_{k-1}\rangle_\star}
{\langle s_{k-1},y_{k-1}\rangle_\star},
\qquad
\alpha_k^{\rm BB2}
=
\frac{\langle s_{k-1},y_{k-1}\rangle_\star}
{\langle y_{k-1},y_{k-1}\rangle_\star}.
\label{eq:optimizer_bb_steps}
\end{equation}
The implementation may choose \(\alpha_k\) from BB1, BB2, or alternate
between them, with fixed safeguards \(0<\alpha_{\min}\le\alpha_{\max}\)
and clipping to \([\alpha_{\min},\alpha_{\max}]\). A nonpositive or
numerically singular denominator triggers a positive fallback step.
The update is
\begin{equation}
\psi_{k+1}
=
\operatorname{Exp}_{\psi_k}(-\alpha_k G_k),
\qquad
\boldsymbol{\theta}_{k+1}
=
\operatorname{Hopf}^{-1}(\psi_{k+1}).
\label{eq:optimizer_bb_update}
\end{equation}
For stability one may add a geodesic line-search safeguard, such as Armijo
backtracking or the same strong-Wolfe/bounded geodesic search used above, seeded
by the Barzilai--Borwein spectral step.  This safeguard uses only objective and,
when the curvature test is checked, gradient evaluations, and does not change
the quantum-information interface.

\subsection*{Summary of the optimizer family}

The optimizer family used here is therefore:
\begin{itemize}
\item \textbf{Hopf-EGT-CG}: exact sphere geodesic update, exact vector transport,
hybrid Riemannian conjugate-gradient memory, and strong-Wolfe/bounded geodesic
line search;

\item \textbf{Hopf Riemannian L-BFGS}: analytic Hopf metric lift, transported
limited-memory curvature pairs, tangent-space two-loop recursion, and geodesic
line search;

\item \textbf{Hopf Riemannian Barzilai--Borwein}: analytic Hopf metric lift,
transported spectral step-size estimation, and a geodesic line-search safeguard
seeded by the spectral step.
\end{itemize}
All three optimizers use the same quantum process: evaluate the cost and the
Hopf-coordinate gradient at selected Hopf states.  The remaining operations are
classical consequences of the Hopf chart.  Ritz-type and local-subspace updates
are intentionally excluded from this appendix because they require stronger
oracle information or inner subspace optimization beyond the cost-and-gradient
interface considered here.
\FloatBarrier
\section{Metrological objectives as chain-rule observables}
\label{app:metrology}

This appendix states a limited transfer principle: the Hopf gradient-access
framework used for VQE also applies to metrological costs whose local
differentials are expectation-value differentials of Hermitian observables. The
main text gives the estimator for the energy objective
\[
E(\boldsymbol{\theta})
=
\bra{\psi(\boldsymbol{\theta})}H\ket{\psi(\boldsymbol{\theta})},
\]
but the same calculus applies to smooth scalar objectives built from measured
probabilities, expectation values, or local phase derivatives.

We keep QFI and fixed-readout CFI separate. A QFI objective is an ultimate local
probe-state benchmark assuming an optimal measurement; a fixed-readout CFI
objective is the Fisher information of a specified measurement protocol. The
fixed-readout Ramsey construction below optimizes a probe for a fixed parity
readout at a calibrated operating point; noisy optimality, global phase
estimation, and measurement co-optimization are outside this example.

Let \(U_{\boldsymbol\theta}\) denote the known Hopf preparation circuit.
The probe is
\[
\ket{\psi(\boldsymbol{\theta})}
=
U_{\boldsymbol{\theta}}\ket{0}^{\otimes n},
\]
and the normalized coordinate tangent is the same as in
Definition~\ref{def:normalized_partial_general},
\[
\ket{e_i(\boldsymbol{\theta})}
=
\frac{\ket{\dot\psi_i(\boldsymbol{\theta})}}
{\sqrt{g_{i,i}(\boldsymbol{\theta})}}.
\]
No new tangent-state synthesis is needed; only the Hermitian observable in the
transition element changes.

\subsection*{A general chain-rule form}

Let \(K\) be the number of measured moments and
\(f:\mathbb R^K\to\mathbb R\) a differentiable function on the relevant
domain. A scalar cost
\[
\mathcal{C}(\boldsymbol{\theta})
=
f\!\left(
m_1(\boldsymbol{\theta}),\ldots,m_K(\boldsymbol{\theta})
\right)
\]
depends on the expectation values of fixed, parameter-independent
Hermitian observables \(O_\alpha\), \(\alpha=1,\ldots,K\):
\[
m_\alpha(\boldsymbol{\theta})
=
\bra{\psi(\boldsymbol{\theta})}
O_\alpha
\ket{\psi(\boldsymbol{\theta})},
\qquad
O_\alpha=O_\alpha^\dagger .
\]
Define the classical chain-rule coefficients
\[
a_\alpha(\boldsymbol{\theta})
:=
\frac{\partial f}{\partial m_\alpha}
\bigg|_{(m_1,\ldots,m_K)=(m_1(\boldsymbol\theta),\ldots,m_K(\boldsymbol\theta))}
\]
and the associated chain-rule observable
\begin{equation}
O_{\mathcal{C}}(\boldsymbol{\theta})
:=
\sum_{\alpha=1}^K
a_\alpha(\boldsymbol{\theta})\,O_\alpha .
\label{eq:met_chain_rule_observable}
\end{equation}
Then
\begin{equation}
\partial_{\theta_i}\mathcal{C}(\boldsymbol{\theta})
=
2\sqrt{g_{i,i}(\boldsymbol{\theta})}\,
\operatorname{Re}\!\left[
\bra{e_i(\boldsymbol{\theta})}
O_{\mathcal{C}}(\boldsymbol{\theta})
\ket{\psi(\boldsymbol{\theta})}
\right].
\label{eq:met_general_gradient}
\end{equation}
Thus the VQE formula \eqref{eq:hopf_grad_short} is recovered by taking
\(K=1\), \(O_1=H\), and \(f(m_1)=m_1\).

For each \(i\), define the same signed branch state as in the main text,
\[
\ket{\varphi_i^{(s)}(\boldsymbol{\theta})}
=
\frac{
\ket{\psi(\boldsymbol{\theta})}
+
s\ket{e_i(\boldsymbol{\theta})}
}{\sqrt{2}},
\qquad
s\in\{+1,-1\}.
\]
For any normalized state \(\ket\chi\) and Hermitian observable \(O\), write
\[
B_{\chi}[O]
:=
\bra{\chi}O\ket{\chi}.
\]
Then
\begin{equation}
\operatorname{Re}\!\left[
\bra{e_i}O\ket{\psi}
\right]
=
s\left(
B_{\varphi_i^{(s)}}[O]
-
\frac{1}{2}
\left(
B_{\psi}[O]+B_{e_i}[O]
\right)
\right),
\label{eq:met_branch_identity_general}
\end{equation}
where
\[
B_{\psi}[O]
=
\bra{\psi}O\ket{\psi},
\qquad
B_{e_i}[O]
=
\bra{e_i}O
\ket{e_i}.
\]
Applying \eqref{eq:met_branch_identity_general} to
\(O=O_{\mathcal{C}}(\boldsymbol{\theta})\) gives
\begin{equation}
\partial_{\theta_i}\mathcal{C}(\boldsymbol{\theta})
=
2\sqrt{g_{i,i}(\boldsymbol{\theta})}\,
s\left(
B_{\varphi_i^{(s)}}[O_{\mathcal{C}}]
-
\frac{1}{2}
\left(
B_{\psi}[O_{\mathcal{C}}]
+
B_{e_i}[O_{\mathcal{C}}]
\right)
\right).
\label{eq:met_branch_gradient}
\end{equation}
In using this identity, \(O_{\mathcal{C}}(\boldsymbol{\theta})\) is built from
chain-rule coefficients evaluated at the current probe state and then held fixed
while branch-state transition elements are estimated. Thus
\(B_{\varphi_i^{(s)}}[O_{\mathcal{C}}]\) is a transition-moment quantity, not the
nonlinear cost \(\mathcal{C}(\varphi_i^{(s)})\). If the coefficients are
estimated from empirical data, the resulting gradient is generally a stochastic
approximation to \(\nabla\mathcal{C}\); nonlinear functions of noisy estimates
can introduce bias without debiasing. This matters especially for inverse-Fisher
objectives, whose prefactors can be sensitive to small probabilities or slopes.

\subsection*{Pure-state QFI objective}

Consider a single unknown parameter \(\phi\) encoded by a unitary
\[
U_\phi=\exp(-\iu\phi G),
\]
with a fixed Hermitian encoding generator \(G\), independent of the
variational coordinates. In this appendix \(G\) denotes this operator,
not the state-gradient vector of Appendix~\ref{app:optimizers}. For a
pure probe state, the quantum Fisher
information is
\begin{equation}
F_Q(\boldsymbol{\theta})
=
4\operatorname{Var}_{\boldsymbol{\theta}}(G)
=
4\left(
\langle G^2\rangle_{\boldsymbol{\theta}}
-
\langle G\rangle_{\boldsymbol{\theta}}^2
\right).
\label{eq:qfi_single_generator}
\end{equation}
Here and below
\[
\langle O\rangle_{\boldsymbol{\theta}}
:=
\bra{\psi(\boldsymbol{\theta})}O\ket{\psi(\boldsymbol{\theta})}.
\]
For a fixed regularization constant \(\eta>0\), a minimized metrology cost
is the regularized quantum Cram\'er--Rao proxy
\begin{equation}
\mathcal{C}_{Q}(\boldsymbol{\theta})
=
\frac{1}{F_Q(\boldsymbol{\theta})+\eta},
\qquad
\eta>0 .
\label{eq:qfi_cost}
\end{equation}
This cost defines a local probe-state benchmark associated with the
optimal-measurement QFI. A fixed experimental readout realizes this benchmark
only when it saturates the QFI at the operating point. By contrast, the
fixed-readout Ramsey objective below is a classical Fisher-information objective
for a specified parity measurement.
Let
\[
\mu(\boldsymbol{\theta})
=
\langle G\rangle_{\boldsymbol{\theta}}.
\]
From \eqref{eq:qfi_single_generator},
\[
\partial_{\theta_i} F_Q
=
4\,\partial_{\theta_i}\langle G^2\rangle
-
8\mu\,\partial_{\theta_i}\langle G\rangle .
\]
Therefore
\begin{equation}
\partial_{\theta_i} F_Q
=
8\sqrt{g_{i,i}(\boldsymbol{\theta})}
\left(
T_i[G^2]
-
2\mu\,T_i[G]
\right),
\label{eq:qfi_grad_tangent}
\end{equation}
where
\[
T_i[O]
:=
\operatorname{Re}\!\left[
\bra{e_i(\boldsymbol{\theta})}
O
\ket{\psi(\boldsymbol{\theta})}
\right].
\]
The minimized-cost gradient is then
\begin{equation}
\partial_{\theta_i}\mathcal{C}_{Q}
=
-\frac{\partial_{\theta_i} F_Q}
{\left(F_Q+\eta\right)^2}.
\label{eq:qfi_cost_gradient}
\end{equation}
Equivalently, this is the general chain-rule formula
\eqref{eq:met_general_gradient} with the chain-rule observable
\begin{equation}
O_{\mathcal{C}_{Q}}
=
-\frac{
4G^2-8\mu G
}
{\left(F_Q+\eta\right)^2}.
\label{eq:qfi_chain_rule_observable}
\end{equation}
Thus the QFI objective requires the same Hopf tangent-state machinery as VQE,
with the Hamiltonian \(H\) replaced by the observables \(G\) and \(G^2\), or
by their linear combination \eqref{eq:qfi_chain_rule_observable}.

\subsection*{Fixed-readout Ramsey CFI objective}

We next record the exact chain-rule observable for a fixed-readout Ramsey CFI
objective. Let \(W\) be the fixed readout unitary and let
\(\Pi_{\mathrm e}=(I+Z^{\otimes n})/2\) project onto even computational
parity. Both are independent of \(\phi\) and \(\boldsymbol\theta\).
With the same fixed encoding generator \(G\), set
\[
U_\phi=\exp(-\iu\phi G),
\qquad
A_{\mathrm e}=W^\dagger \Pi_{\mathrm e} W,
\]
and define the effective even-parity measurement operator
\[
M_{\mathrm e}(\phi)
=
U_\phi^\dagger A_{\mathrm e} U_\phi .
\]
At a fixed operating point \(\phi_0\), define
\[
p(\boldsymbol{\theta})
=
\bra{\psi(\boldsymbol{\theta})}
M_{\mathrm e}(\phi_0)
\ket{\psi(\boldsymbol{\theta})},
\]
and
\[
q(\boldsymbol{\theta})
=
\left.
\frac{\partial}{\partial \phi}
\bra{\psi(\boldsymbol{\theta})}
M_{\mathrm e}(\phi)
\ket{\psi(\boldsymbol{\theta})}
\right|_{\phi=\phi_0}.
\]
Equivalently,
\[
q(\boldsymbol{\theta})
=
\bra{\psi(\boldsymbol{\theta})}
D_{\mathrm e}(\phi_0)
\ket{\psi(\boldsymbol{\theta})},
\]
where \([G,A_{\mathrm e}]=GA_{\mathrm e}-A_{\mathrm e}G\) is the
commutator, and the derivative observable is
\begin{equation}
D_{\mathrm e}(\phi_0)
:=
\left.
\frac{\partial M_{\mathrm e}(\phi)}{\partial\phi}
\right|_{\phi=\phi_0}
=
\iu U_{\phi_0}^\dagger
\left[
G,A_{\mathrm e}
\right]
U_{\phi_0}.
\label{eq:ramsey_derivative_observable}
\end{equation}
At an interior operating point satisfying
\[
0<p(\boldsymbol{\theta})<1,
\]
the fixed-readout classical Fisher information is
\[
F_C(\boldsymbol{\theta})
=
\frac{q(\boldsymbol{\theta})^2}
{p(\boldsymbol{\theta})[1-p(\boldsymbol{\theta})]},
\]
and the minimized local inverse-CFI cost is
\[
\mathcal{C}_{\mathrm{Ram}}(\boldsymbol{\theta})
=
\frac{1}{F_C(\boldsymbol{\theta})+\eta},
\qquad
\eta>0.
\]
The derivatives below are taken on this interior domain. At the boundary points
\(p(\boldsymbol{\theta})\in\{0,1\}\), the quotient defining \(F_C\) is singular
and must instead be treated by a separate limiting prescription or by explicitly
regularizing the Bernoulli denominator \(p(1-p)\). The additive constant
\(\eta\) in the inverse cost does not by itself regularize this denominator.

For \(0<p(\boldsymbol{\theta})<1\), the chain-rule coefficients with respect to
\(p\) and \(q\) are
\begin{align}
a_p(\boldsymbol{\theta})
&:=
\frac{\partial \mathcal{C}_{\mathrm{Ram}}}{\partial p}
=
\frac{
q(\boldsymbol{\theta})^2
[1-2p(\boldsymbol{\theta})]
}{
[F_C(\boldsymbol{\theta})+\eta]^2
p(\boldsymbol{\theta})^2
[1-p(\boldsymbol{\theta})]^2
},
\nonumber\\
a_q(\boldsymbol{\theta})
&:=
\frac{\partial \mathcal{C}_{\mathrm{Ram}}}{\partial q}
=
-\frac{
2q(\boldsymbol{\theta})
}{
[F_C(\boldsymbol{\theta})+\eta]^2
p(\boldsymbol{\theta})
[1-p(\boldsymbol{\theta})]
}.
\label{eq:ramsey_chain_coefficients}
\end{align}
Thus the exact fixed-readout Ramsey chain-rule observable is
\begin{equation}
O_{\mathcal{C}_{\mathrm{Ram}}}(\boldsymbol{\theta})
=
a_p(\boldsymbol{\theta})M_{\mathrm e}(\phi_0)
+
a_q(\boldsymbol{\theta})D_{\mathrm e}(\phi_0),
\label{eq:ramsey_chain_rule_observable}
\end{equation}
and the Hopf gradient is
\[
\partial_{\theta_i}\mathcal{C}_{\mathrm{Ram}}
=
2\sqrt{g_{i,i}(\boldsymbol{\theta})}\,
\operatorname{Re}\!\left[
\bra{e_i(\boldsymbol{\theta})}
O_{\mathcal{C}_{\mathrm{Ram}}}(\boldsymbol{\theta})
\ket{\psi(\boldsymbol{\theta})}
\right].
\]

For a centered finite-difference implementation, choose a fixed phase
increment \(\delta>0\), independent of \(\boldsymbol\theta\), and replace
\(D_{\mathrm e}(\phi_0)\) by
\begin{equation}
D_{\mathrm e,\delta}
=
\frac{
M_{\mathrm e}(\phi_0+\delta)
-
M_{\mathrm e}(\phi_0-\delta)
}{2\delta}.
\label{eq:ramsey_finite_difference_observable}
\end{equation}
The directly executable three-point scheme defines the surrogate slope,
Fisher quantity, and cost by
\[
q_\delta=\bra\psi D_{\mathrm e,\delta}\ket\psi,\qquad
F_{C,\delta}=\frac{q_\delta^2}{p(1-p)},\qquad
\mathcal C_{\mathrm{Ram},\delta}=\frac{1}{F_{C,\delta}+\eta}.
\]
All quantities are evaluated at the current \(\boldsymbol\theta\), with
\(0<p<1\). Define \(a_{p,\delta}\) and \(a_{q,\delta}\) by
Eq.~\eqref{eq:ramsey_chain_coefficients} with \(q\) replaced by
\(q_\delta\) and \(F_C\) by \(F_{C,\delta}\), leaving \(p\) unchanged.
The surrogate chain-rule observable is then
\[
O_{\mathcal{C}_{\mathrm{Ram},\delta}}(\boldsymbol{\theta})
=
a_{p,\delta}(\boldsymbol{\theta})M_{\mathrm e}(\phi_0)
+
a_{q,\delta}(\boldsymbol{\theta})D_{\mathrm e,\delta}.
\]
With these coefficients held fixed during branch measurements, the
transition-moment identity gives the exact coordinate gradient of
\(\mathcal C_{\mathrm{Ram},\delta}\), not of the original inverse-CFI
cost. For fixed smooth encoding and \(p\) bounded away from zero and one,
the surrogate and its gradient approach their exact counterparts with
\(\bigO(\delta^2)\) truncation error as \(\delta\to0\). Finite-shot
slope estimates also divide probability noise by \(2\delta\), so reducing
\(\delta\) trades truncation error against amplified sampling noise.
Using empirical chain-rule coefficients can add the plug-in bias described
above; that effect is separate from finite-difference truncation.

\subsection*{Compiled-setting scaling}

The metrology transfer does not alter the Hopf state preparation or gradient
access circuits.  The indexed tangent-state preparations
\(\ket{\Phi_d(\boldsymbol{\theta})}\), the indexed branch states
\(\ket{\Omega_d(\boldsymbol{\theta})}\), and, for the complex ansatz, the
corresponding phase-block states
\(\ket{\Phi_{\mathrm{ph}}(\boldsymbol{\theta})}\) and
\(\ket{\Omega_{\mathrm{ph}}(\boldsymbol{\theta})}\), are exactly the same
objects as in Section~\ref{sec:gradient_access}.  Only the observable or
readout configuration used to evaluate the transition moments changes:
\[
H
\quad\longrightarrow\quad
O_{\mathcal C}(\boldsymbol{\theta}).
\]

Let \(K_{\mathrm{read}}\) denote the number of logical readout configurations
or moment types required by the chosen metrological cost. For the pure-QFI
formula above, \(K_{\mathrm{read}}\) includes the moments needed for \(G\) and
\(G^2\), together with any additional grouping cost from their observable
decompositions. For a centered finite-difference Ramsey implementation, the
logical phase settings are
\[
\phi_0,
\qquad
\phi_0+\delta,
\qquad
\phi_0-\delta,
\]
so \(K_{\mathrm{read}}=3\) before any further hardware-specific repetition or
sampling allocation.

Write \(N_{\mathrm{set,met}}^{\mathbb R}\) and
\(N_{\mathrm{set,met}}^{\mathbb C}\) for the total preparation-setting
counts when every logical readout is included. The Hopf count appears per
logical readout configuration. For the real ansatz,
\[
N_{\mathrm{set,met}}^{\mathbb{R}}
=
K_{\mathrm{read}}(1+2n),
\]
and for the complex ansatz,
\[
N_{\mathrm{set,met}}^{\mathbb{C}}
=
K_{\mathrm{read}}\bigl[1+2(n+1)\bigr].
\]
For a chosen metrological cost \(\mathcal C_{\mathrm{met}}\), let
\(C_{\mathrm{cfg}}^{\mathrm{read}}\) bound the assigned Hopf preparation
charge per setting and logical readout. Let
\(C_{\mathrm{grad}}(\nabla\mathcal C_{\mathrm{met}})\) sum these charges
over the required settings, each counted once. Then
\begin{equation}
\boxed{
C_{\mathrm{grad}}(\nabla \mathcal{C}_{\mathrm{met}})
=
\bigO\!\left(K_{\mathrm{read}}n\right)
C_{\mathrm{cfg}}^{\mathrm{read}}
=
\bigO\!\left(K_{\mathrm{read}}\log N\right)
C_{\mathrm{cfg}}^{\mathrm{read}}.
}
\label{eq:met_cfg_scaling}
\end{equation}
When \(K_{\mathrm{read}}\) is independent of \(N\), this preserves the
\(\bigO(\log N)\) compiled-setting scaling of the VQE case. If the chosen
metrological observable or readout scheme requires a number of logical
settings that grows with problem size, that cost enters multiplicatively.

Equation~\eqref{eq:met_cfg_scaling} gives the assigned Hopf preparation
charge across distinct settings, not hardware routing costs or finite-shot
work. The corresponding
sampling budget is componentwise and is determined by the chosen metrological
readout: by the measurements needed for \(G\) and \(G^2\) in the pure-QFI
formulation, or by the Ramsey probability and local phase-slope measurements in
the fixed-readout formulation above.

\FloatBarrier
\section{Proofs of the main results}
\label{app:proof}

\subsection*{Closed-form product formula for Hopf amplitudes}

\begin{lemma}[Closed-form product formula for Hopf amplitudes]
\label{lem:hopf_product_formula}
Fix $n\in\mathbb{N}$, set $N=2^n$, and let $i\in\{0,\ldots,N-1\}$ have binary expansion
\[
i=(q_n q_{n-1}\cdots q_1)_2,\qquad q_k\in\{0,1\},
\]
using the qubit order $(q_n,\ldots,q_1)$ (MSB to LSB).
Define the sequence of internal-node indices along the root-to-leaf construction by
\[
s_t(i)
:=2^{t-1}+\big(q_n q_{n-1}\cdots q_{n-t+2}\big)_2,
\qquad t=1,\ldots,n,
\]
where for $t=1$ the empty bit string is interpreted as $0$, so that $s_1(i)=1$.
Equivalently, $s_1(i)=1$ and
\[
s_{t+1}(i)=2\,s_t(i)+q_{n-t+1},
\qquad t=1,\ldots,n-1.
\]

(A) Real Hopf coordinates.
The amplitude produced by Definition~\ref{def:hopf_coordinates}(A) satisfies
\[
x_i(\boldsymbol{\theta})
=\prod_{t=1}^{n}
\begin{cases}
\cos\theta_{s_t(i)}, & q_{n-t+1}=0,\\[2pt]
\sin\theta_{s_t(i)}, & q_{n-t+1}=1.
\end{cases}
\]
In particular,
\[
x_i(\boldsymbol{\theta})^2
=\prod_{t=1}^{n}
\begin{cases}
\cos^2\theta_{s_t(i)}, & q_{n-t+1}=0,\\[2pt]
\sin^2\theta_{s_t(i)}, & q_{n-t+1}=1.
\end{cases}
\]

(B) Complex Hopf coordinates.
Let $r_i(\boldsymbol{\theta})\ge 0$ denote the magnitude produced from the internal-node
angles as in (A), using $\theta_1,\ldots,\theta_{N-1}$.
Then Definition~\ref{def:hopf_coordinates}(B) yields
\[
x_i(\boldsymbol{\theta})
=
r_i(\boldsymbol{\theta})\,e^{\iu\theta_{N+i}},
\qquad
|x_i(\boldsymbol{\theta})|^2=r_i(\boldsymbol{\theta})^2.
\]
\end{lemma}

\begin{proof}
Fix $i$ and run the update rule in Definition~\ref{def:hopf_coordinates}.
At the start, $x_i=1$ and $j=1$.
At step $t$ (corresponding to $k=n-t+1$), the algorithm multiplies $x_i$ by
$\cos\theta_j$ if $q_{n-t+1}=0$ and by $\sin\theta_j$ if $q_{n-t+1}=1$,
and then updates $j\leftarrow 2j$ or $2j+1$ accordingly.
By construction, the value of $j$ used at step $t$ is exactly $s_t(i)$.
Therefore after $n$ steps,
$x_i$ equals the stated product of $n$ cosine/sine factors, proving (A),
and squaring gives the formula for $x_i^2$.

In the complex case, the same internal-node recursion produces the nonnegative magnitude
$r_i(\boldsymbol{\theta})$, after which Definition~\ref{def:hopf_coordinates}(B) multiplies
by the leaf phase $e^{\iu\theta_{N+i}}$, yielding (B).
\end{proof}

\subsection*{Proof of Theorem~\ref{thm:hopf_metric_diagonal}}

\begin{proof}

We write
\(
\ket{\psi(\boldsymbol{\theta})}=\sum_{k=0}^{N-1}x_k(\boldsymbol{\theta})\ket{b_k}.
\)
By Lemma~\ref{lem:hopf_product_formula}, each amplitude $x_k$ is a product of
$\cos\theta_j$ and $\sin\theta_j$ factors along a unique root-to-leaf path,
and in the complex case an additional leaf phase factor.

(A) Real case.
Here $x_k\in\mathbb{R}$ and the metric is
\(
g^{\mathbb{R}}_{i\ell}=\braket{\dot\psi_i|\dot\psi_\ell}
=\sum_{k}\partial_{\theta_i} x_k\,\partial_{\theta_\ell} x_k
\)
(cf.\ Definition~\ref{def:hopf_metric_convention}(A)).

Step 1: Support structure of $\ket{\dot\psi_i}$.
Fix an internal node index $i\in\{1,\ldots,N-1\}$ and let $\mathcal{L}(i)$ be
the set of leaf indices in the subtree rooted at $i$.
If $k\notin\mathcal{L}(i)$ then the product formula for $x_k$ contains no
factor $\cos\theta_i$ or $\sin\theta_i$, hence $\partial_{\theta_i} x_k=0$.
Therefore $\ket{\dot\psi_i}$ is supported only on $\{\ket{b_k}:k\in\mathcal{L}(i)\}$.

Moreover, writing $\mathcal{L}(i)=\mathcal{L}_L(i)\,\dot\cup\,\mathcal{L}_R(i)$
as the leaves in the left and right child subtrees of $i$, Lemma~\ref{lem:hopf_product_formula}
implies that for $k\in\mathcal{L}_L(i)$ the amplitude has the form
\[
x_k(\boldsymbol{\theta})=\alpha_k(\boldsymbol{\theta})\cos\theta_i,
\]
where $\alpha_k$ does not depend on $\theta_i$, and for $k\in\mathcal{L}_R(i)$,
\[
x_k(\boldsymbol{\theta})=\beta_k(\boldsymbol{\theta})\sin\theta_i,
\]
where $\beta_k$ does not depend on $\theta_i$.
Thus
\begin{equation}
\partial_{\theta_i} x_k=
\begin{cases}
-\alpha_k\sin\theta_i=-\tan\theta_i\,x_k, & k\in\mathcal{L}_L(i),\\[2pt]
\ \ \beta_k\cos\theta_i=\ \cot\theta_i\,x_k, & k\in\mathcal{L}_R(i),\\[2pt]
0, & k\notin\mathcal{L}(i).
\end{cases}
\label{eq:real_support_and_scaling}
\end{equation}

Step 2: Diagonal entries.
Define the incoming mass at node $i$ by
\[
S_i(\boldsymbol{\theta}) := \sum_{k\in\mathcal{L}(i)} x_k(\boldsymbol{\theta})^2.
\]
Using \eqref{eq:real_support_and_scaling},
\begin{align*}
g^{\mathbb{R}}_{ii}
&=\sum_{k=0}^{N-1}(\partial_{\theta_i} x_k)^2\\
&=\tan^2\theta_i\sum_{k\in\mathcal{L}_L(i)} x_k^2
\;+\;\cot^2\theta_i\sum_{k\in\mathcal{L}_R(i)} x_k^2.
\end{align*}
But from the factorization $x_k=\alpha_k\cos\theta_i$ on $\mathcal{L}_L(i)$ and
$x_k=\beta_k\sin\theta_i$ on $\mathcal{L}_R(i)$ we have
\[
\sum_{k\in\mathcal{L}_L(i)}x_k^2=\cos^2\theta_i\sum_{k\in\mathcal{L}_L(i)}\alpha_k^2,
\qquad
\sum_{k\in\mathcal{L}_R(i)}x_k^2=\sin^2\theta_i\sum_{k\in\mathcal{L}_R(i)}\beta_k^2.
\]
Furthermore, by construction the descendant angles below $i$ parameterize unit vectors
on each child subtree, so both sums of squares of $\alpha_k$ and $\beta_k$ equal the same incoming mass:
\[
\sum_{k\in\mathcal{L}_L(i)}\alpha_k^2=\sum_{k\in\mathcal{L}_R(i)}\beta_k^2=S_i(\boldsymbol{\theta}).
\]
Substituting gives
\[
g^{\mathbb{R}}_{ii}
=\tan^2\theta_i\cdot \cos^2\theta_i\,S_i+\cot^2\theta_i\cdot \sin^2\theta_i\,S_i
=\sin^2\theta_i\,S_i+\cos^2\theta_i\,S_i
=S_i(\boldsymbol{\theta}).
\]
Finally, $S_i$ depends only on the ancestor split probabilities above $i$.
Let $d_i=\lfloor\log_2 i\rfloor$ be the depth of node $i$ (root at depth $0$),
and let $s_1,\ldots,s_{d_i}$ be the ancestor indices from the root $s_1=1$ to the
parent of $i$, with corresponding left/right bits $q_1,\ldots,q_{d_i}$ (MSB first)
as in the theorem statement. Each ancestor contributes a factor $\cos^2\theta_{s_m}$
if the path goes left and $\sin^2\theta_{s_m}$ if it goes right, hence
\[
g^{\mathbb{R}}_{ii}(\boldsymbol{\theta})=S_i(\boldsymbol{\theta})
=\prod_{m=1}^{d_i}
\begin{cases}
\cos^2\theta_{s_m}, & q_m=0,\\
\sin^2\theta_{s_m}, & q_m=1,
\end{cases}
\]
which is \eqref{eq:g_real_diag}. For $i=1$ the product is empty, so $g^{\mathbb{R}}_{11}=1$.

Step 3: Off-diagonal entries (diagonality).
Let $i\neq \ell$ be internal nodes. If $\mathcal{L}(i)\cap\mathcal{L}(\ell)=\varnothing$,
then $\ket{\dot\psi_i}$ and $\ket{\dot\psi_\ell}$ have disjoint supports, hence
$g^{\mathbb{R}}_{i\ell}=0$.

Otherwise, one of the nodes is an ancestor of the other; assume $\ell$ lies in the
subtree of $i$. Then on the entire leaf set $\mathcal{L}(\ell)$, the scaling factor in
\eqref{eq:real_support_and_scaling} is constant (all leaves of $\mathcal{L}(\ell)$ lie
entirely in either $\mathcal{L}_L(i)$ or $\mathcal{L}_R(i)$). Therefore there exists a real
constant $c$ such that for all $k\in\mathcal{L}(\ell)$,
\(
\partial_{\theta_i} x_k=c\,x_k
\)
and $\partial_{\theta_i} x_k=0$ for $k\notin\mathcal{L}(i)\supseteq\mathcal{L}(\ell)$.
Hence
\[
g^{\mathbb{R}}_{i\ell}=\sum_{k}\partial_{\theta_i} x_k\,\partial_{\theta_\ell} x_k
=c\sum_{k\in\mathcal{L}(\ell)} x_k\,\partial_{\theta_\ell} x_k.
\]
But the vector of amplitudes restricted to $\mathcal{L}(\ell)$ has fixed squared norm
$S_\ell(\boldsymbol{\theta})$, and $\partial_{\theta_\ell}$ varies only the internal angle at node $\ell$,
so differentiating $S_\ell=\sum_{k\in\mathcal{L}(\ell)}x_k^2$ with respect to $\theta_\ell$ yields
\(
0=\partial_{\theta_\ell} S_\ell=2\sum_{k\in\mathcal{L}(\ell)} x_k\,\partial_{\theta_\ell} x_k.
\)
Thus $\sum_{k\in\mathcal{L}(\ell)} x_k\,\partial_{\theta_\ell} x_k=0$ and therefore
$g^{\mathbb{R}}_{i\ell}=0$. This proves that $\boldsymbol{g}^{\mathbb{R}}$ is diagonal.

(B) Complex case.
Write $x_k(\boldsymbol{\theta})=r_k(\boldsymbol{\theta})e^{\iu\phi_k}$ with
$\phi_k=\theta_{N+k}$ and $r_k\ge 0$ determined by the internal-node angles.
By Lemma~\ref{lem:hopf_product_formula}, the magnitudes $r_k$ obey the same
cosine--sine tree factorization as in the real case.

For $1\le i,\ell\le N-1$ (magnitude parameters), we have
\(
\partial_{\theta_i} x_k=(\partial_{\theta_i} r_k)e^{\iu\phi_k}
\)
and therefore
\[
\braket{\dot\psi_i|\dot\psi_\ell}
=\sum_k \overline{\partial_{\theta_i} x_k}\,\partial_{\theta_\ell} x_k
=\sum_k (\partial_{\theta_i} r_k)(\partial_{\theta_\ell} r_k),
\]
which is real. The same support/cancellation argument from (A) applies to $\{r_k\}$,
so the magnitude block is diagonal and
\(
g^{\mathbb{C}}_{ii}=g^{\mathbb{R}}_{ii}
\)
for $1\le i\le N-1$.

For phase parameters $i=N+k$,
\(
\partial_{\theta_{N+k}}x_m=\iu\,x_m\,\delta_{mk},
\)
so
\[
g^{\mathbb{C}}_{N+k,\,N+k}
=\mathrm{Re}\,\braket{\dot\psi_{N+k}|\dot\psi_{N+k}}
=\mathrm{Re}\,(|\iu x_k|^2)
=|x_k|^2,
\]
and for $k\neq m$ the phase derivatives are orthogonal, hence the phase block is diagonal.

Finally, for a magnitude parameter $i\le N-1$ and a phase parameter $N+k$,
\[
\braket{\dot\psi_i|\dot\psi_{N+k}}
=\overline{\partial_{\theta_i} x_k}\,(\iu x_k)
=\iu\,(\partial_{\theta_i} r_k)\,r_k,
\]
which is purely imaginary; taking the real part (Definition~\ref{def:hopf_metric_convention}(B))
gives $g^{\mathbb{C}}_{i,\,N+k}=0$. Hence the full complex metric is diagonal and
\eqref{eq:g_complex_diag} holds.
\end{proof}

\subsection*{Proof of Theorem~\ref{thm:gradient_estimation_short}}

\begin{proof}
Let $\ket{\psi(\boldsymbol{\theta})}=U_{\boldsymbol{\theta}}\ket{0}^{\otimes n}$ be differentiable
and define
\(
E(\boldsymbol{\theta})=\bra{\psi(\boldsymbol{\theta})}H\ket{\psi(\boldsymbol{\theta})}
\)
with $H=H^\dagger$.
Using the product rule,
\begin{equation}
\partial_{\theta_i} E(\boldsymbol{\theta})
=
\partial_{\theta_i}\big(\bra{\psi(\boldsymbol{\theta})}H\ket{\psi(\boldsymbol{\theta})}\big)\\
=
\bra{\dot\psi_i(\boldsymbol{\theta})}H\ket{\psi(\boldsymbol{\theta})}
+
\bra{\psi(\boldsymbol{\theta})}H\ket{\dot\psi_i(\boldsymbol{\theta})}.
\end{equation}
Since $H$ is Hermitian,
\(
\bra{\psi(\boldsymbol{\theta})}H\ket{\dot\psi_i(\boldsymbol{\theta})}
=
\overline{\bra{\dot\psi_i(\boldsymbol{\theta})}H\ket{\psi(\boldsymbol{\theta})}},
\)
hence
\[
\partial_{\theta_i} E(\boldsymbol{\theta})
=
2\,\mathrm{Re}\!\left[\bra{\dot\psi_i(\boldsymbol{\theta})}H\ket{\psi(\boldsymbol{\theta})}\right].
\]
Whenever $g_{i,i}(\boldsymbol{\theta})>0$, Definition~\ref{def:normalized_partial_general} gives
\[
\ket{e_i(\boldsymbol{\theta})}
=
\frac{\ket{\dot\psi_i(\boldsymbol{\theta})}}{\sqrt{g_{i,i}(\boldsymbol{\theta})}},
\qquad
\ket{\dot\psi_i(\boldsymbol{\theta})}
=
\sqrt{g_{i,i}(\boldsymbol{\theta})}\,\ket{e_i(\boldsymbol{\theta})},
\]
so substituting yields
\[
\partial_{\theta_i} E(\boldsymbol{\theta})
=
2\sqrt{g_{i,i}(\boldsymbol{\theta})}\;
\mathrm{Re}\!\left[\bra{e_i(\boldsymbol{\theta})}\,H\,\ket{\psi(\boldsymbol{\theta})}\right],
\]
which is \eqref{eq:hopf_grad_short}.
\end{proof}

\subsection*{Proof of Theorem~\ref{thm:synthesis_normalized_direction_hopf}}

\begin{proof}
We treat magnitude parameters and phase parameters separately.
Throughout, write
\(
\ket{\psi(\boldsymbol{\theta})}=\sum_{k=0}^{N-1}x_k(\boldsymbol{\theta})\ket{b_k}.
\)

In this proof, the vectors \(\boldsymbol{\theta}^{(i)}\) are interpreted as gate-parameter assignments on the fixed Hopf circuit skeleton. They are not required to satisfy the canonical coordinate-range restrictions of Definition~\ref{def:hopf_coordinates}; only the resulting circuit action matters.

(A) Magnitude parameters ($1\le i\le N-1$).
Fix an internal node index $i$ and let $\mathcal{L}(i)$ be the set of leaves in the subtree
rooted at $i$. Let $\mathcal{L}_L(i)$ and $\mathcal{L}_R(i)$ denote the leaves in the left and
right child subtrees of $i$.
By Lemma~\ref{lem:hopf_product_formula}, every amplitude $x_k$ factorizes as a product of
trigonometric terms along the unique root-to-leaf path.

Step 1: Isolating the subtree rooted at $i$.
Let $A_L^{(i)}$ and $A_R^{(i)}$ be the ancestor sets defined in the theorem statement.
Setting
\(
\theta^{(i)}_a=0
\)
for $a\in A_L^{(i)}$ forces the corresponding ancestor factor to be $\cos\theta^{(i)}_a=\cos 0=1$
and the sibling-branch factor to be $0$ (since $\sin 0=0$), thereby eliminating amplitudes
supported on the ``wrong'' sibling subtrees.
Similarly, setting
\(
\theta^{(i)}_a=\pi/2
\)
for $a\in A_R^{(i)}$ forces $\sin\theta^{(i)}_a=\sin(\pi/2)=1$ and $\cos(\pi/2)=0$, again
eliminating all sibling subtrees not on the path to $i$.
Consequently,
\[
x_k\!\big(\boldsymbol{\theta}^{(i)}\big)=0
\qquad\text{for all }k\notin\mathcal{L}(i),
\]
so the prepared state has support only on $\{\ket{b_k}:k\in\mathcal{L}(i)\}$.

Moreover, for $k\in\mathcal{L}(i)$ the product of all ancestor factors above $i$ equals $1$
under $\boldsymbol{\theta}^{(i)}$, so the remaining dependence of
$x_k(\boldsymbol{\theta}^{(i)})$ comes only from the angle at $i$ and from strict descendants of
$i$. Since the construction keeps all descendant angles unchanged
(\eqref{eq:desc_keep_mag}), the relative amplitude pattern inside the subtree is preserved.

Step 2: A two-branch decomposition at node $i$.
There exist vectors supported on the two child subtrees,
\[
\ket{\psi_L}:=\sum_{k\in\mathcal{L}_L(i)} \alpha_k(\boldsymbol{\theta})\ket{b_k},
\qquad
\ket{\psi_R}:=\sum_{k\in\mathcal{L}_R(i)} \beta_k(\boldsymbol{\theta})\ket{b_k},
\]
such that $\|\psi_L\|=\|\psi_R\|=1$ and, on the subtree $\mathcal{L}(i)$,
\begin{equation}
\sum_{k\in\mathcal{L}(i)} x_k(\boldsymbol{\theta})\ket{b_k}
=
\gamma_i(\boldsymbol{\theta})
\Big(\cos\theta_i\,\ket{\psi_L}+\sin\theta_i\,\ket{\psi_R}\Big),
\label{eq:subtree_decomp_original}
\end{equation}
where the nonnegative scalar $\gamma_i(\boldsymbol{\theta})$ is the common product of ancestor
trigonometric factors above node $i$ (hence independent of $\theta_i$).
Equivalently,
\(
\gamma_i(\boldsymbol{\theta})^2=\sum_{k\in\mathcal{L}(i)}|x_k(\boldsymbol{\theta})|^2.
\)

Differentiating \eqref{eq:subtree_decomp_original} with respect to $\theta_i$ yields
\begin{equation}
\ket{\dot\psi_i(\boldsymbol{\theta})}
=
\gamma_i(\boldsymbol{\theta})
\Big(-\sin\theta_i\,\ket{\psi_L}+\cos\theta_i\,\ket{\psi_R}\Big),
\label{eq:subtree_derivative}
\end{equation}
and therefore
\[
\|\ket{\dot\psi_i(\boldsymbol{\theta})}\|^2=\gamma_i(\boldsymbol{\theta})^2.
\]
In the real case this equals $g^{\mathbb{R}}_{i,i}(\boldsymbol{\theta})$, and in the complex
case (magnitude block) it equals $g^{\mathbb{C}}_{i,i}(\boldsymbol{\theta})$
(consistent with Theorem~\ref{thm:hopf_metric_diagonal}).
Hence, whenever $g_{i,i}(\boldsymbol{\theta})>0$, we have
\begin{equation}
\ket{e_i(\boldsymbol{\theta})}
=
\frac{\ket{\dot\psi_i(\boldsymbol{\theta})}}{\sqrt{g_{i,i}(\boldsymbol{\theta})}}
=
-\sin\theta_i\,\ket{\psi_L}+\cos\theta_i\,\ket{\psi_R}.
\label{eq:normalized_derivative_subtree}
\end{equation}

Step 3: Real ansatz synthesis.
By Step~1, under $\boldsymbol{\theta}^{(i)}$ all amplitude weight is clamped into the subtree
$\mathcal{L}(i)$ and the ancestor prefactor becomes $1$.
By \eqref{eq:desc_keep_mag}, the descendant angles below $i$ are unchanged, so the same
$\ket{\psi_L}$ and $\ket{\psi_R}$ appear.
Finally, by \eqref{eq:ti_shift_mag} we shift $\theta_i\mapsto\theta_i+\pi/2$, and using
\(
\cos(\theta_i+\pi/2)=-\sin\theta_i
\)
and
\(
\sin(\theta_i+\pi/2)=\cos\theta_i,
\)
the resulting subtree superposition is exactly
\(
-\sin\theta_i\,\ket{\psi_L}+\cos\theta_i\,\ket{\psi_R}.
\)
Therefore,
\[
\mathrm{HopfReal}\!\big(\boldsymbol{\theta}^{(i)}\big)\ket{0}^{\otimes n}
=
\ket{e_i(\boldsymbol{\theta})}.
\]

Step 4: Complex ansatz synthesis (magnitude block).
For $\mathrm{HopfComplex}$, the magnitude recursion is identical to the real case, and the
phase assignment \eqref{eq:phase_keep_on_subtree} keeps the original leaf phases on
$\mathcal{L}(i)$ while setting all other phases to $0$.
Since Step~1 makes $x_k(\boldsymbol{\theta}^{(i)})=0$ for $k\notin\mathcal{L}(i)$, those
outside phases are irrelevant.
Inside the subtree, the same trigonometric shift at node $i$ produces the same linear
combination \eqref{eq:normalized_derivative_subtree}, now with the correct inherited leaf
phases. Hence,
\[
\mathrm{HopfComplex}\!\big(\boldsymbol{\theta}^{(i)}\big)\ket{0}^{\otimes n}
=
\ket{e_i(\boldsymbol{\theta})}.
\]

Finally, if node $i$ has depth $d$ (root at depth $0$), its subtree contains exactly
$N/2^d$ leaves. The resulting state is supported only on the corresponding
computational basis states, so its support size is at most $N/2^d$.
Equality holds whenever its amplitude on every leaf in $\mathcal{L}(i)$ is nonzero.

(B) Phase parameters ($i=N+\ell$).
Fix $\ell\in\{0,\ldots,N-1\}$ and write
\(
x_\ell(\boldsymbol{\theta})=r_\ell(\boldsymbol{\theta})e^{\iu\theta_{N+\ell}}
\)
with $r_\ell\ge 0$.
Then
\[
\ket{\dot\psi_{N+\ell}(\boldsymbol{\theta})}
=
\iu\,x_\ell(\boldsymbol{\theta})\,\ket{b_\ell},
\qquad
\left\|\ket{\dot\psi_{N+\ell}(\boldsymbol{\theta})}\right\|^2
=
|x_\ell(\boldsymbol{\theta})|^2,
\]
so
\begin{equation}
\ket{e_{N+\ell}(\boldsymbol{\theta})}
=
\frac{\iu\,x_\ell(\boldsymbol{\theta})}{|x_\ell(\boldsymbol{\theta})|}\ket{b_\ell}.
\label{eq:phase_normalized_target}
\end{equation}

Now construct $\boldsymbol{\theta}^{(i)}$ as in \eqref{eq:phase_case_clamp}--\eqref{eq:phase_case_shift}.
The ancestor clamping forces the magnitude to be supported on the single leaf $\ell$ with
$r_\ell(\boldsymbol{\theta}^{(i)})=1$ and $r_m(\boldsymbol{\theta}^{(i)})=0$ for $m\neq \ell$.
The phase shift sets
\(
\theta^{(i)}_{N+\ell}=\theta_{N+\ell}+\pi/2,
\)
so the prepared state is
\[
\mathrm{HopfComplex}\!\big(\boldsymbol{\theta}^{(i)}\big)\ket{0}^{\otimes n}
=
\iu\,e^{\iu\theta_{N+\ell}}\ket{b_\ell}.
\]
Since
\[
x_\ell(\boldsymbol{\theta})
=
r_\ell(\boldsymbol{\theta})e^{\iu\theta_{N+\ell}},
\qquad
|x_\ell(\boldsymbol{\theta})|=r_\ell(\boldsymbol{\theta}),
\]
we have
\[
\frac{\iu\,x_\ell(\boldsymbol{\theta})}{|x_\ell(\boldsymbol{\theta})|}
=
\iu\,e^{\iu\theta_{N+\ell}}.
\]
Hence
\[
\mathrm{HopfComplex}\!\big(\boldsymbol{\theta}^{(i)}\big)\ket{0}^{\otimes n}
=
\iu\,e^{\iu\theta_{N+\ell}}\ket{b_\ell}
=
\frac{\iu\,x_\ell(\boldsymbol{\theta})}{|x_\ell(\boldsymbol{\theta})|}\ket{b_\ell}
=
\ket{e_{N+\ell}(\boldsymbol{\theta})}.
\]
This completes the proof.
\end{proof}

\subsection*{Proof of Theorem~\ref{thm:hopf_gradient_setting_cost}}

\begin{proof}
We prove the per-setting assigned CNOT bound separately for
magnitude-layer settings and phase-layer settings.

Throughout, the forward Hopf ansatz uses a fixed binary-tree gate skeleton with
\(N-1\) single-qubit rotation locations, one for each internal node of the Hopf tree
(Definition~\ref{def:hopf_ansatz}).
At each location, the original gate already carries some data-register controls;
their number is at most \(n-1\).

(A) Magnitude-layer settings: \(\ket{\Phi_d}\) and \(\ket{\Omega_d}\).
Fix a depth \(d\in\{0,1,\ldots,n-1\}\).
The reduced index register \(I\) contains \(d\) qubits and labels the nodes in the layer
\(
\mathcal{L}_d=\{2^d,\ldots,2^{d+1}-1\}
\)
through the offset \(r=i-2^d\).

For a fixed magnitude parameter \(i\in\mathcal{L}_d\), Theorem~\ref{thm:synthesis_normalized_direction_hopf}(A)
shows that the tangent-state preparation is obtained from the same Hopf skeleton by replacing each magnitude angle
by one of a constant-size set of values:
\[
0,\qquad \tfrac{\pi}{2},\qquad \theta_j,\qquad \theta_j+\tfrac{\pi}{2}.
\]
When the preparation is made coherent over all \(i\in\mathcal{L}_d\), the
nonzero assignment at a gate of depth \(h\) is selected by a single
mixed-polarity prefix condition on the layer index: selection of its right
subtree for \(h<d\), equality with the gate for \(h=d\), or equality with its
depth-\(d\) ancestor for \(h>d\). Hence each gate location requires only
\(\bigO(1)\) index-conditioned cases, with no angle multiplexing.

By Lemma~\ref{lem:cnot_model}, the CNOT cost of a controlled single-qubit gate is linear in the total number of controls once the small-control regime is absorbed into constants.
Hence each index-conditioned case costs \(\bigO(n)\) CNOTs.
Since there are \(N-1=\bigO(N)\) gate locations and only \(\bigO(1)\) cases per location, the total assigned CNOT charge of the indexed tangent setting \(\ket{\Phi_d(\boldsymbol{\theta})}\) is
\[
\bigO(N)\times \bigO(n)=\bigO(nN).
\]

For the branch setting \(\ket{\Omega_d(\boldsymbol{\theta})}\), one ancilla qubit is added to coherently combine the baseline state and the indexed derivative preparation.
This adds at most one extra control to the same family of gates and therefore does not change the asymptotic bound:
\[
C_{\mathrm{cfg}}\big(\ket{\Omega_d(\boldsymbol{\theta})}\big)=\bigO(nN).
\]

This proves the real-case claim and the magnitude-block part of the complex-case claim.

(B) Phase-layer settings: \(\ket{\Phi_{\mathrm{ph}}}\) and \(\ket{\Omega_{\mathrm{ph}}}\).
Now consider the complex Hopf ansatz and the phase parameters
\(
\theta_{N+\ell}
\),
\(
\ell=0,\ldots,N-1
\).
Let \(P\) be the \(n\)-qubit phase-index register.

For each fixed leaf label
\(
\ell=(b_n b_{n-1}\cdots b_1)_2,
\)
Theorem~\ref{thm:synthesis_normalized_direction_hopf}(B) gives an explicit parameter vector
\(
\boldsymbol{\theta}^{(N+\ell)}
\)
such that
\[
\mathrm{HopfComplex}\!\big(\boldsymbol{\theta}^{(N+\ell)}\big)\ket{0}^{\otimes n}
=
\ket{e_{N+\ell}(\boldsymbol{\theta})}.
\]
Thus the indexed phase derivative preparation is obtained by implementing the map
\(
\ell\mapsto \boldsymbol{\theta}^{(N+\ell)}
\)
coherently under control of \(P\).

For an internal-node gate location \(j\), the phase-case construction
\eqref{eq:phase_case_clamp} sets its magnitude angle either to \(0\) or to \(\pi/2\), depending only on whether \(j\) lies on the selected root-to-leaf path and on the next branch bit.
Hence each internal-node gate location requires at most one nontrivial index-conditioned case.

For a promoted final-layer \(R_{\mathbb C}\) gate corresponding to a sibling leaf pair, the phase-case construction
\eqref{eq:phase_case_shift} requires at most two nontrivial cases:
one selecting the left leaf and one selecting the right leaf.
Each such case uses only the two local leaf phases already attached to that gate location, together with the constant shift \(+\pi/2\).
Again, no angle multiplexing is required.

Therefore the indexed phase-layer implementation uses the same Hopf skeleton as the forward complex ansatz, with only \(\bigO(1)\) case splitting at each original gate location.
Each case inherits the original data-register controls and acquires at most \(n\) additional controls from the phase-index register \(P\).
By Lemma~\ref{lem:cnot_model}, each such case costs \(\bigO(n)\) CNOTs.
Since there are \(\bigO(N)\) gate locations and \(\bigO(1)\) cases per location, the indexed phase-derivative setting satisfies
\[
C_{\mathrm{cfg}}\big(\ket{\Phi_{\mathrm{ph}}(\boldsymbol{\theta})}\big)=\bigO(nN).
\]

For the phase-branch setting \(\ket{\Omega_{\mathrm{ph}}(\boldsymbol{\theta})}\), the additional ancilla used to form the signed superposition adds at most one extra control and therefore preserves the same asymptotic bound:
\[
C_{\mathrm{cfg}}\big(\ket{\Omega_{\mathrm{ph}}(\boldsymbol{\theta})}\big)=\bigO(nN).
\]

Combining parts (A) and (B), every gradient-access configuration used in the Hopf protocol (whether derivative or branch, magnitude or phase) has assigned CNOT charge \(\bigO(nN)\).
\end{proof}

\subsection*{Proof of Lemma~\ref{lem:inverse_hopf_map}}

\begin{proof}
Let \(N=2^n\).

\textbf{(A) Real case.}
Let \(\mathbf{x}\in\mathbb{R}^{N}\) satisfy
\(\sum_{\ell=0}^{N-1}x_\ell^2=1\).
Use a tree-indexed array of weights, with leaf weights
\[
w_{N+\ell}:=x_\ell^2,
\qquad
\ell=0,\ldots,N-1.
\]
For internal nodes \(j=N-1,N-2,\ldots,1\), define
\[
w_j:=w_{2j}+w_{2j+1}.
\]
Then \(w_j\) is the total squared mass in the subtree rooted at \(j\). In
particular,
\[
w_j=\sum_{\ell\in\mathcal{L}(j)}x_\ell^2,
\qquad
w_{2j}=\sum_{\ell\in\mathcal{L}_L(j)}x_\ell^2,
\qquad
w_{2j+1}=\sum_{\ell\in\mathcal{L}_R(j)}x_\ell^2.
\]
Thus
\[
S_L(j)=\sqrt{w_{2j}},
\qquad
S_R(j)=\sqrt{w_{2j+1}},
\qquad
S(j)=\sqrt{w_j}.
\]

For every non-final internal node \(1\le j\le N/2-1\), set
\[
\theta_j=
\begin{cases}
\operatorname{atan2}\!\big(\sqrt{w_{2j+1}},\sqrt{w_{2j}}\big),
& w_j>0,\\[2pt]
0, & w_j=0.
\end{cases}
\]
Equivalently, when \(w_j>0\),
\[
\cos\theta_j=\frac{\sqrt{w_{2j}}}{\sqrt{w_j}},
\qquad
\sin\theta_j=\frac{\sqrt{w_{2j+1}}}{\sqrt{w_j}}.
\]
Therefore the forward Hopf recursion sends the incoming magnitude
\(\sqrt{w_j}\) at node \(j\) to child magnitudes
\[
\sqrt{w_j}\cos\theta_j=\sqrt{w_{2j}},
\qquad
\sqrt{w_j}\sin\theta_j=\sqrt{w_{2j+1}}.
\]
By induction from the root down to depth \(n-1\), the recursion reproduces the
correct incoming magnitude for every final-layer sibling pair.

It remains only to recover the signs inside each final sibling pair.  For a
final-layer node \(j=N/2+k\), the two descendant leaves are \(2k\) and \(2k+1\).
Let
\[
S_k:=\sqrt{x_{2k}^2+x_{2k+1}^2}.
\]
If \(S_k>0\), set
\[
\theta_{N/2+k}
=
\operatorname{atan2}\!\big(x_{2k+1},x_{2k}\big)
\pmod{2\pi}.
\]
Then
\[
S_k\cos\theta_{N/2+k}=x_{2k},
\qquad
S_k\sin\theta_{N/2+k}=x_{2k+1}.
\]
If \(S_k=0\), both leaf amplitudes vanish and the angle is arbitrary; the
convention in the lemma sets it to zero.  Hence the real forward Hopf map
reproduces the signed vector \(\mathbf{x}\).  Nonuniqueness occurs exactly at
zero-subtree or zero-leaf-pair points.

\textbf{(B) Complex case.}
Let \(\mathbf{x}\in\mathbb{C}^{N}\) satisfy
\(\sum_{\ell=0}^{N-1}|x_\ell|^2=1\), and write
\[
x_\ell=r_\ell e^{\iu\phi_\ell},
\qquad
r_\ell=|x_\ell|\ge0,
\]
with the convention \(\phi_\ell=0\) when \(x_\ell=0\).

Apply the same bottom-up construction to the nonnegative leaf weights
\[
w_{N+\ell}:=r_\ell^2=|x_\ell|^2.
\]
For every internal node \(j=1,\ldots,N-1\), set
\[
\theta_j=
\begin{cases}
\operatorname{atan2}\!\big(\sqrt{w_{2j+1}},\sqrt{w_{2j}}\big),
& w_j>0,\\[2pt]
0, & w_j=0.
\end{cases}
\]
Since all magnitudes are nonnegative, these angles lie in \([0,\pi/2]\), and
the forward magnitude recursion reproduces exactly the magnitudes
\(r_0,\ldots,r_{N-1}\).

Finally set the leaf phases to
\[
\theta_{N+\ell}:=\phi_\ell,
\qquad
\ell=0,\ldots,N-1.
\]
Then the complex forward Hopf map gives
\[
r_\ell e^{\iu\theta_{N+\ell}}
=
r_\ell e^{\iu\phi_\ell}
=
x_\ell
\]
for every leaf.  If \(r_\ell=0\), the corresponding phase is irrelevant, so the
inverse is not unique at zero-amplitude leaves.

\textbf{Complexity.}
The bottom-up computation of all weights uses one traversal of a complete
binary tree with \(\bigO(N)\) nodes.  Once the weights are known, each magnitude
angle is obtained in \(\bigO(1)\) time.  In the real case, the final-layer signed
angles are obtained directly from the pairs
\((x_{2k},x_{2k+1})\).  In the complex case, the leaf phases require one
additional linear pass.  Therefore
\[
\text{Time}_{\mathrm{Hopf}^{-1}}(N)=\bigO(N),
\qquad
\text{Memory}_{\mathrm{Hopf}^{-1}}(N)=\bigO(N).
\]
This proves the lemma.
\end{proof}

\subsection*{Proof of Theorem~\ref{thm:hopf_jacobian_both}}

\begin{proof}
Fix a leaf (basis) index $i\in\{0,\ldots,N-1\}$.
By Lemma~\ref{lem:hopf_product_formula}, the Hopf amplitude $x_i(\boldsymbol{\theta})$
is a product of trigonometric factors contributed by the internal nodes on the unique
root-to-leaf path of $i$, and (in the complex case) an additional leaf phase factor.

Let $P(i)$ denote the set of internal nodes on the root-to-parent path that generates $x_i$
(as in the theorem statement), and for $j\in P(i)$ let $\varepsilon_{ij}\in\{0,1\}$ indicate
whether $x_i$ contains the factor $\cos\theta_j$ ($\varepsilon_{ij}=0$) or $\sin\theta_j$
($\varepsilon_{ij}=1$).

(A) Real case.
Assume $x_i(\boldsymbol{\theta})\in\mathbb{R}$.
If $j\notin P(i)$, then the product formula for $x_i$ contains no factor depending on
$\theta_j$, hence $\partial x_i/\partial\theta_j=0$.

Now assume $j\in P(i)$.
Write the product formula as
\[
x_i(\boldsymbol{\theta})
=
f_{ij}(\theta_j)\,u_{ij}(\boldsymbol{\theta}),
\]
where \(u_{ij}\) is the product of all factors in \(x_i\) except the one at node
\(j\), so \(u_{ij}\) is independent of \(\theta_j\).  Here
\[
f_{ij}(\theta_j)=
\begin{cases}
\cos\theta_j, & \varepsilon_{ij}=0,\\
\sin\theta_j, & \varepsilon_{ij}=1.
\end{cases}
\]
The direct product derivative is always
\[
\frac{\partial x_i}{\partial\theta_j}
=
f'_{ij}(\theta_j)\,u_{ij}(\boldsymbol{\theta}).
\]
Where \(f_{ij}(\theta_j)\neq0\), this may be written as
\[
\frac{\partial x_i}{\partial\theta_j}
=
\frac{f'_{ij}(\theta_j)}{f_{ij}(\theta_j)}
x_i(\boldsymbol{\theta}).
\]
Thus, if \(\varepsilon_{ij}=0\), then
\[
\frac{\partial x_i}{\partial\theta_j}
=
-\tan\theta_j\,x_i,
\]
while if \(\varepsilon_{ij}=1\), then
\[
\frac{\partial x_i}{\partial\theta_j}
=
\cot\theta_j\,x_i.
\]
At points where \(f_{ij}(\theta_j)=0\), the displayed \(\tan\) or \(\cot\)
formula is interpreted by the direct product derivative above, equivalently by
continuous extension.  This proves \eqref{eq:jac_real_entries} with the stated
product-derivative convention.

Finally, since in the real case
\(
g^{\mathbb{R}}_{j\ell}(\boldsymbol{\theta})
=\sum_{i=0}^{N-1}\frac{\partial x_i}{\partial\theta_j}\frac{\partial x_i}{\partial\theta_\ell}
\)
(cf.\ \eqref{eq:real_metric_JTJ}), the matrix identity
\(
\boldsymbol{g}^{\mathbb{R}}=(\boldsymbol{J}^{\mathbb{R}})^\top\boldsymbol{J}^{\mathbb{R}}
\)
is exactly the statement that $g^{\mathbb{R}}_{j\ell}$ is the Euclidean inner product
of columns $j$ and $\ell$ of $\boldsymbol{J}^{\mathbb{R}}$, proving \eqref{eq:g_real_gram}.

(B) Complex case.
Write
\[
x_i(\boldsymbol{\theta}) = r_i(\boldsymbol{\theta})\,e^{\iu\theta_{N+i}},
\qquad r_i(\boldsymbol{\theta})\ge 0,
\]
where $r_i$ is built from internal-node angles as in Lemma~\ref{lem:hopf_product_formula}(B).

For $1\le j\le N-1$ (magnitude parameters), the phase
$e^{\iu\theta_{N+i}}$ is independent of $\theta_j$, so the same
product/log-derivative argument as in (A) applies verbatim and yields the two
path-derivative cases in \eqref{eq:jac_complex_entries}, together with the
zero case when \(j\notin P(i)\), now with \(x_i\) complex-valued.

For the phase parameter $j=N+i$,
\[
\frac{\partial x_i}{\partial\theta_{N+i}}
=
r_i(\boldsymbol{\theta})\cdot \iu e^{\iu\theta_{N+i}}
= \iu\,x_i(\boldsymbol{\theta}),
\]
and for $j=N+m$ with $m\neq i$ the derivative is $0$ since $x_i$ does not depend on that phase.
This completes \eqref{eq:jac_complex_entries}.

Finally, by Definition~\ref{def:hopf_metric_convention}(B),
\[
g^{\mathbb{C}}_{j\ell}(\boldsymbol{\theta})
=\mathrm{Re}\,\langle \partial_{\theta_j}\psi,\partial_{\theta_\ell}\psi\rangle
=\mathrm{Re}\sum_{i=0}^{N-1}\overline{\frac{\partial x_i}{\partial\theta_j}}
\frac{\partial x_i}{\partial\theta_\ell},
\]
which is precisely the $(j,\ell)$ entry of
$\mathrm{Re}\big((\boldsymbol{J}^{\mathbb{C}})^\dagger\boldsymbol{J}^{\mathbb{C}}\big)$,
proving \eqref{eq:g_complex_gram}.
\end{proof}

\end{document}